\RequirePackage{fix-cm}
\documentclass[a4paper]{article}
\usepackage{graphicx}
\usepackage[margin=1in]{geometry}

\usepackage[utf8]{inputenc}
\usepackage{color}
\usepackage{amsmath}
\usepackage{amsthm}
\usepackage{amssymb}
\usepackage{amsopn}
\usepackage{caption}
\usepackage{subcaption}
\usepackage{gensymb}
\usepackage[inline]{enumitem}
\usepackage{todonotes}
\usepackage{hyperref}
\usepackage{sidecap}
\usepackage{wrapfig}
\usepackage{geometry}
\usepackage{paralist}
\usepackage{thmtools}
\usepackage{thm-restate}
\usepackage{amsmath}
\usepackage{amsthm}

\newtheorem{observation}{Observation}

\newtheorem{definition}{Definition}
\newtheorem{lemma}{Lemma}
\newtheorem{corollary}{Corollary}
\newtheorem{proposition}{Proposition}
\declaretheorem[name=Theorem]{thm}

\newcommand{\embedding}{\mathcal{E}}

\let\doendproof\endproof
\renewcommand\endproof{\qed\doendproof}

\newcommand{\flip}[1]{\mathrm{flip}(#1)}
\newcommand{\mirror}[1]{\mathrm{mirror}(#1)}

\newcommand{\T}{\mathrm{T}}
\newcommand{\B}{\mathrm{B}}

\newcommand{\widebar}[1]{\mkern 
  1.5mu\overline{\mkern-1.5mu#1\mkern-1.5mu}\mkern 1.5mu}
\newcommand{\subpath}[2]{\ensuremath{#1[#2]}}

\newcommand{\join}{+}
\newcommand{\reverse}[1]{\ensuremath{\widebar{#1}}}
\newcommand{\NP}{\ensuremath{\mathcal{NP}}}

\renewcommand{\O}{\ensuremath{\mathcal{O}}}

\DeclareMathOperator{\rot}{rot}
\DeclareMathOperator{\dir}{dir}

\usepackage{etoolbox}
\newcounter{claimcount}
\newenvironment{Claim}{\refstepcounter{claimcount}\noindent\textit{Claim \arabic{claimcount}:}}{}
\AtBeginEnvironment{proof}{\setcounter{claimcount}{0}}
\begin{document}

\title{A Topology-Shape-Metrics Framework for\\ Ortho-Radial
  Graph~Drawing\thanks{This manuscript is based on the two conference papers (1) \emph{L. Barth, B. Niedermann, I. Rutter, and M. Wolf. Towards a topology-shape-metrics framework for ortho-radial drawings. In Leibniz International Proceedings in Informatics. Proc. 33rd Annual ACM Symposium on Computational Geometry (SoCG '17), pages 14:1-14:16. 2017.} and (2) \emph{ 
 B. Niedermann, I. Rutter, and M. Wolf. Efficient algorithms for ortho-radial graph drawing. In volume 129 of Leibniz International Proceedings in Informatics. Proc. 35th Annual ACM Symposium on Computational Geometry (SoCG '19). Schloss Dagstuhl-Leibniz-Zentrum fuer Informatik, 2019.}}
}

\author{Lukas Barth\thanks{Karlsruhe Institute of Technology,  lukas.barth@kit.edu} \and Benjamin Niedermann\thanks{University of Bonn, niedermann@uni-bonn.de} \and Ignaz Rutter\thanks{University of Passau, rutter@fim.uni-passau.de} \and Matthias 
  Wolf\thanks{Karlsruhe Institute of Technology, matthias.wolf@kit.edu}
}

\date{}

\maketitle

\begin{abstract}
  Orthogonal drawings, i.e., embeddings of graphs into grids, are a
  classic topic in Graph Drawing.  Often the goal is to find a drawing
  that minimizes the number of bends on the edges.  A key ingredient
  for bend minimization algorithms is the existence of an
  \emph{orthogonal representation} that allows to describe such
  drawings purely combinatorially by only listing the angles between
  the edges around each vertex and the directions of bends on the
  edges, but neglecting any kind of geometric information such as
  vertex coordinates or edge lengths.

  In this work, we generalize this idea to \emph{ortho-radial
    representations} of \emph{ortho-radial drawings}, which are
  embeddings into an ortho-radial grid, whose gridlines are concentric
  circles around the origin and straight-line spokes emanating from
  the origin but excluding the origin itself.  Unlike the orthogonal
  case, there exist ortho-radial representations that do not admit a
  corresponding drawing, for example so-called strictly monotone
  cycles.  An
  ortho-radial drawing is called \emph{valid} if it does not contain a
  strictly monotone cycle.  Our first main result is that an
  ortho-radial
  representation admits a corresponding drawing if and only if it is
  valid.  Previously such a characterization was only known for
  ortho-radial drawings of paths, cycles, and theta
  graphs~\cite{hht-orthoradial-09}, and in the special case of
  rectangular drawings of cubic graphs~\cite{hhmt-rrdcp-10}, where the
  contour of each face is required to be a rectangle.  Additionally,
  we give a quadratic-time algorithm that tests for a given
  ortho-radial representation whether it is valid, and we show how to
  draw a valid ortho-radial representation in the same running time.

  Altogether, this reduces the problem of computing a minimum-bend
  ortho-radial drawing to the task of computing a valid ortho-radial
  representation with the minimum number of bends, and hence
  establishes an ortho-radial analogue of the topology-shape-metrics
  framework for planar orthogonal drawings by Tamassia~\cite{t-emn-87}.
\end{abstract}

\section{Introduction}
\label{sec:introduction}
Grid drawings of graphs embed graphs into grids such that vertices map
to grid points and edges map to internally disjoint curves on the grid
lines that connect their endpoints.  Orthogonal grids, whose grid
lines are horizontal and vertical lines, are popular and widely used
in graph drawing. Among other applications, orthogonal graph drawings are
used in VLSI design (e.g.,~\cite{Valiant1981,Bhatt1984}), diagrams
(e.g.,~\cite{Batini1986,Gutwenger2003,Eiglsperger2004,Wybrow2010}),
and network layouts (e.g.,~\cite{Ruegg2014,Kieffer2016}).  They have
been extensively studied with respect to their construction and
properties
(e.g.,~\cite{Tamassia1991,Biedl1996,Biedl1998,Papakostas1998,Alam2017}).
Moreover, they have been generalized to arbitrary planar graphs with
degree higher than four
(e.g.,~\cite{Tamassia1988,Fossmeier1996,Biedl1997}).

\begin{figure}[t]
\begin{minipage}[b]{0.52\textwidth}
  \centering
  \begin{subfigure}[b]{0.48\textwidth}
    \centering
    \includegraphics{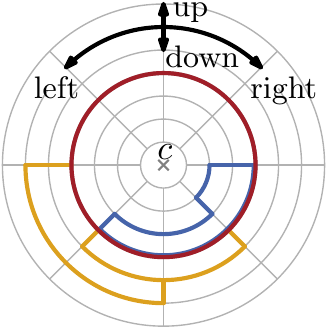}
    \caption{Ortho-radial grid.}
    \label{fig:pre:drawing-grid}
  \end{subfigure}
  \hfill
  \begin{subfigure}[b]{0.44\textwidth}
    \centering
    \includegraphics{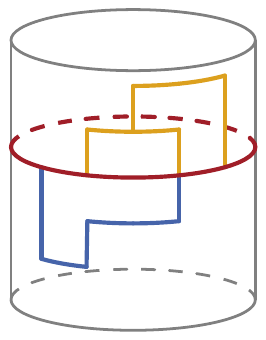}
    \caption{Cylinder drawing.}
    \label{fig:pre:drawing-cylinder}
  \end{subfigure}
  \hfill
  \caption{An ortho-radial drawing of a graph on a grid
    \protect(\subref{fig:pre:drawing-grid}) and its equivalent interpretation
    as an orthogonal drawing on a cylinder
    \protect(\subref{fig:pre:drawing-cylinder}).}
  \label{fig:pre:drawing}
\end{minipage}
\hfill
 \begin{minipage}[b]{0.46\textwidth}
   \centering
   \includegraphics[width=.9\textwidth]{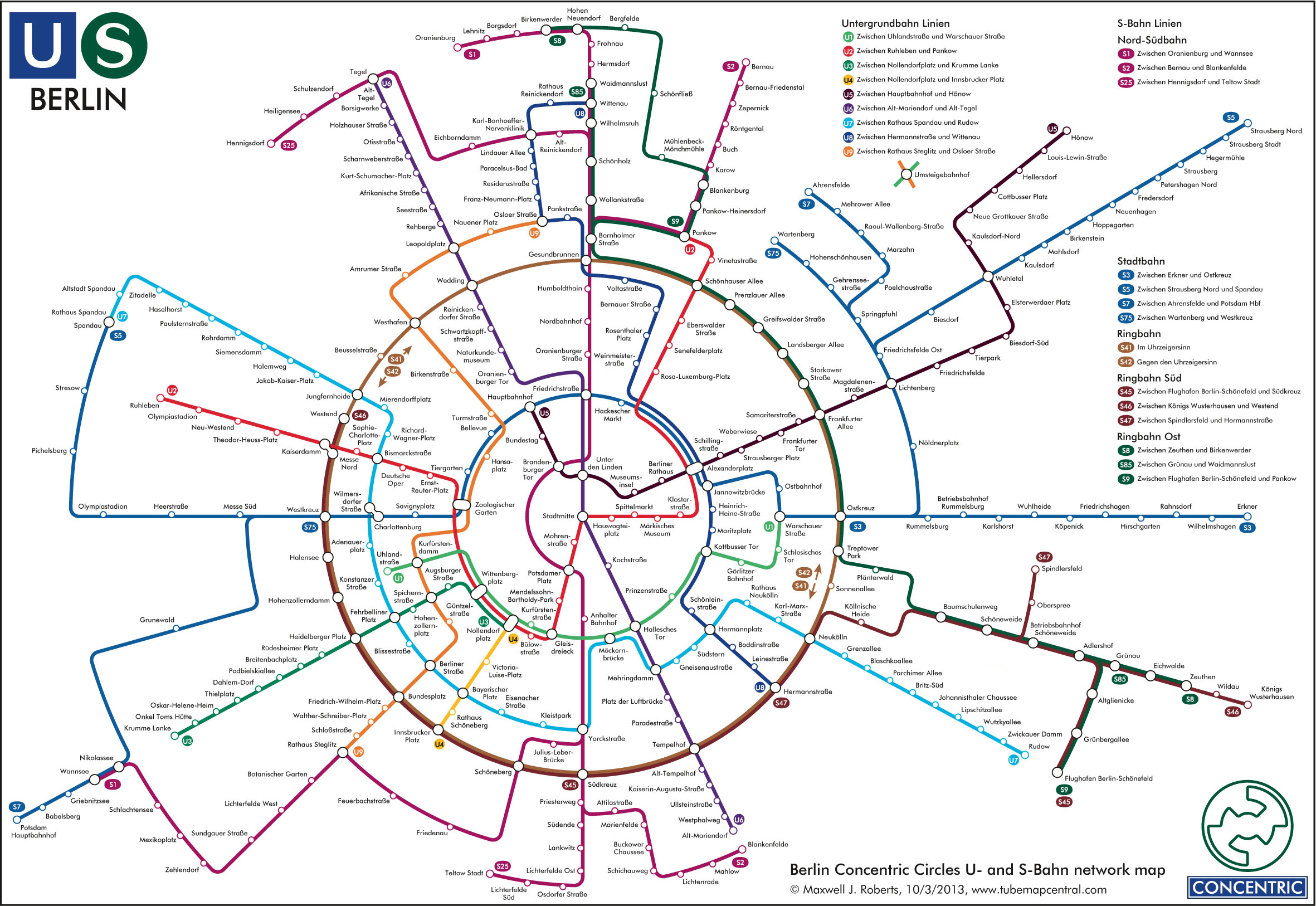}
   
   \caption{Metro map of Berlin using an ortho-radial 
   layout\protect\footnotemark{}.
   Image
     copyright by Maxwell J.~Roberts. Reproduced with permission.}
   \label{fig:intro:berlin}
 \end{minipage} 
\end{figure}
Ortho-radial drawings are a generalization
of orthogonal drawings to grids that are formed by concentric circles
around the origin and straight-line spokes from the origin, but
excluding the origin.  Equivalently, they can be viewed as graphs
drawn in an orthogonal fashion on the surface of a standing cylinder,
see Figure~\ref{fig:pre:drawing}, or a sphere without the
poles. Hence, they naturally bring orthogonal graph drawings to the
third dimension.

Among other applications, ortho-radial drawings are used to visualize
network maps; see Figure~\ref{fig:intro:berlin}.  Especially, for
metro systems of metropolitan areas they are highly suitable. Their
inherent structure emphasizes the city center, the metro lines that
run in circles as well as the metro lines that lead to suburban
areas. While the automatic creation of metro maps has been extensively
studied for other layout
styles~(e.g.,~\cite{Hong2006,Noellenburg2011,Wang2011,Fink2013}), this
is a new and wide research field for ortho-radial drawings~\cite{finketal2014concentric}.

\footnotetext{Note that ortho-radial drawings exclude the center of
  the grid, which is slightly different to the concentric circles maps
  by Maxwell J.\ Roberts.}

Adapting existing techniques and objectives from orthogonal graph
drawings is a promising step to open up that field.  One main
objective in orthogonal graph drawing is to minimize the number of
bends on the edges.  The key ingredient of a large fraction of the algorithmic
work on this problem is the \emph{orthogonal representation},
introduced by Tamassia~\cite{t-emn-87}, which describes orthogonal
drawings by listing
\begin{inparaenum}[(i)]
\item the angles formed by consecutive edges around each vertex and
\item the directions of bends along the edges.
\end{inparaenum}
Such a representation is \emph{valid} if
\begin{inparaenum}[(I)]
\item the angles around each vertex sum to $360\degree$, and
\item the sum of the angles around each face with $k$ vertices is $(k-2)\cdot 
180\degree$ for internal faces and $(k+2)\cdot 180\degree$ for the outer face.
\end{inparaenum}
The necessity of the first condition is obvious and the necessity of
the latter follows from the sum of inner/outer angles of any polygon
with $k$ corners.  It is thus clear that any orthogonal drawing yields
a valid orthogonal representation, and Tamassia~\cite{t-emn-87} showed
that the converse holds true as well; for a valid orthogonal
representation there exists a corresponding orthogonal drawing that
realizes this representation.  Moreover, the proof is constructive and
allows the efficient construction of such a drawing, a process that is
referred to as \emph{compaction}.

Altogether, this enables a three-step approach for computing orthogonal
drawings, the so-called \emph{Topology-Shape-Metrics Framework}, which
works as follows.  First, fix a \emph{topology}, i.e., a combinatorial
embedding of the graph in the plane (possibly planarizing it if it is
non-planar); second, determine the \emph{shape} of the drawing by
constructing a valid orthogonal representation with few bends; and
finally, compactify the orthogonal representation by assigning
suitable vertex coordinates and edge lengths (\emph{metrics}).  As
mentioned before, this reduces the problem of computing an orthogonal
drawing of a planar graph with a fixed embedding to the purely
combinatorial problem of finding a valid orthogonal representation,
preferably with few bends.  The task of actually creating a
corresponding drawing in polynomial time is then taken over by the
framework.  It is this approach that is at the heart of a large body
of literature on bend minimization algorithms for orthogonal drawings
(e.g.,
\cite{Bertolazzi2000,Eiglsperger2003,Cornelsen2012,Felsner2014,%
  Blasius2016,Blasius2016b,Chang2017}).

\paragraph{Contribution and Outline.}

In this paper we establish an analogous drawing framework
for ortho-radial drawings.  To this end, we introduce so-called
ortho-radial representations, which give a combinatorial description
of ortho-radial drawings, and therefore can be used to substitute
orthogonal representations in the Topology-Shape-Metrics Framework.

More precisely, our contributions are as follows.  We show that a
natural generalization of the validity conditions (I) and (II) above
is not sufficient, and introduce a third, less local condition that
excludes so-called \emph{strictly monotone cycles}, which do not admit
an ortho-radial drawing.  We prove that these three conditions together
fully characterize ortho-radial drawings.  Before that,
characterizations for bend-free ortho-radial drawings were only known
for paths, cycles and theta graphs~\cite{hht-orthoradial-09}. Further,
for the special case that each internal face is a rectangle, a
characterization for cubic graphs was known~\cite{hhmt-rrdcp-10}.

On the algorithmic side, we show that testing whether a given
ortho-radial representation is drawable can be done in $\O(n^2)$ time,
and a corresponding drawing can be obtained in the same running time.
While this does not yet directly allow us to compute ortho-radial
drawings with few bends, our result paves the way for a purely
combinatorial treatment of bend minimization in ortho-radial drawings,
thus enabling the same type of tools that have proven highly
successful in minimizing bends in orthogonal drawings. Recently,
Niedermann and Rutter~\cite{Niedermann2020} presented such a tool based on an integer
linear programming formulation showing that the topology-shape-metrics
framework for ortho-radial drawings is capable of handling real-world
networks such as metro systems.

We formally introduce ortho-radial drawings and ortho-radial
representations in Section~\ref{sec:representations-drawings}, where
we also establish basic properties that will be used throughout this
paper.  Section~\ref{sec:labeling-properties} introduces basic
properties of labelings that are used to describe ortho-radial
representations.  In Sections~\ref{sec:characterization-rect}
and~\ref{sec:rectangulation} we prove that ortho-radial
representations are drawable if and only if they are valid.  In
Section~\ref{sec:finding_monotone_cycles} we give a validity test for
ortho-radial representations that runs in $\O(n^{2})$ time.
Afterwards, in Section~\ref{sec:efficient_rectangulation}, we revisit
the rectangulation procedure from~Section~\ref{sec:rectangulation} and
show that using the techniques from
Section~\ref{sec:finding_monotone_cycles} it can be implemented to run
in $\O(n^{2})$ time, improving over a naive application which would
yield running time $\O(n^{4})$.  This enables a purely combinatorial
treatment of ortho-radial drawings.
In Section~\ref{sec:bend-minimization} we show that computing bend-minimal
ortho-radial representations is $\NP$-complete regardless of whether the
embedding of the graph is fixed or not.  We conclude with a summary and
some open questions in Section~\ref{sec:conclusion}.

\section{Preliminaries}
\label{sec:preliminaries}
 Let $G$ be a plane graph with combinatorial embedding
$\embedding$ and outer face $f_o$. The embedding $\embedding$ fixes
for each vertex $v$ of $G$ the counterclockwise order of the edges
incident to $v$ around the vertex $v$.  A \emph{path} in $G$ may
contain vertices multiple times, and a \emph{cycle}~$C$ may contain
vertices multiple times but may not cross itself in the sense that the
pairs of edges along which $C$ enters and leaves a vertex $v$ do not
alternate in the cyclic order of edges around $v$ in the
embedding~$\embedding$. We consider all paths and cycles to be
directed. We represent a path $P$ as the sequence $v_1\dots v_k$ of
its vertices in the order as they appear on $P$. Similarly, we
represent a cycle as the sequence $v_1\dots v_k$ of its vertices in
the order as they appear on $C$, where $v_1$ is arbitrarily chosen.
For any path $P=v_1\dots v_k$ its \emph{reverse} is
$\reverse{P}=v_k\dots v_1$. The concatenation of two paths $P_1$ and
$P_2$ is written as $P_1+P_2$. For two edges $uv$ and $wx$ on a path
$P$ the \emph{subpath} from $uv$ to $wx$ is the unique path on $P$
that starts with $uv$ and ends with $wx$, and we denote it by
$P[uv, wx]$.  If $P$ contains $u$ (or $x$) only once, we may write $u$
instead of $uv$ (or $x$ instead of $wx$). In particular, if $P$ is
simple $P[u, x]$ denotes the subpath of $P$ from $u$ to $x$.  For a
cycle~$C$, we similarly denote its reverse by $\reverse{C}$, and for
edges $uv$ and $wx$ on $C$ the subpath of $C$ from $uv$ to $wx$ in the
direction of $C$ is denoted by $C[uv, wx]$.

Moreover, a path or cycle is \emph{simple} if it contains all vertices
at most once.  A \emph{facial walk} $C$ of a face~$f$ is a cycle in
$G$ that describes the boundary of $f$, i.e., the cycle $C$ consists
of edges of $f$ and for any subpath $uvw$ of $C$ the edge $uv$
precedes $vw$ in the cyclic order of edges around $v$ that is defined
by $\embedding$.  Any simple cycle $C$ separates two sets of
faces. One of these sets contains the outer face~$f_o$, and we call
these faces together with the vertices and edges incident to them the
\emph{exterior} of $C$.  Conversely the faces of the other set and
their incident vertices and edges form the \emph{interior} of
$C$. Note that $C$ belongs to both its interior and its exterior.
Unless specified explicitly, a simple cycle~$C$ is directed such that
its interior lies to the right of~$C$. Finally, a path $P$ \emph{respects} a
cycle $C$ if $P$ lies in the exterior of $C$.

\section{Ortho-Radial Drawings and Representations}
\label{sec:representations-drawings}

Let $G=(V,E)$ be a planar, connected 4-graph with $n$ vertices, where
a graph is a 4-graph if it has maximum degree four. An
\emph{ortho-radial drawing}~$\Delta$ of $G$ is a plane drawing on an
ortho-radial grid $\mathcal G$ such that each vertex of $G$ is a grid
point of $\mathcal G$ and each edge of $G$ is a curve on $\mathcal
G$. We observe that in any ortho-radial drawing there is an unbounded
face~$f_o$ and a face~$f_c$ that contains the center of the grid; we
call the former the \emph{outer face} and the latter the \emph{central
  face}; in our figures we mark the central face using a small
``x''. All other faces are \emph{regular}. We remark that $f_c$ and
$f_o$ are not necessarily distinct. We further distinguish
two types of simple cycles. If the central face lies in the interior
of a simple cycle, the cycle is \emph{essential} and otherwise
\emph{non-essential}.

In this paper, we assume that we are given $G$, a fixed combinatorial
embedding $\mathcal E$ of $G$ and two (not necessarily distinct) faces
$f_c$ and $f_o$ of $\mathcal E$.  We seek an ortho-radial drawing
$\Delta$ of $G$ such that the combinatorial embedding of $\Delta$ is
$\embedding$, the face $f_c$ is the central face of $\Delta$ and $f_o$
is the outer face of $\Delta$. We call the tuple
$I=(G,\mathcal E,f_c,f_o)$ an \emph{instance} of ortho-radial graph
drawing and $\Delta$ a drawing of $I$.

We observe that the definition of ortho-radial drawings allows edges
to have bends, i.e., an edge may consist of a sequence of
straight-line segments and circular arcs. In this paper, we focus on
ortho-radial drawings without bends; we call such drawings
\emph{bend-free}. Hence, each edge is either part of a radial ray or
of a concentric circle of $\mathcal G$.  This is not a restriction as
any ortho-radial drawing can be turned into a bend-free drawing by
replacing bends with subdivision vertices.

In a bend-free ortho-radial drawing of $G$ each edge has a
\emph{geometric direction} in the sense that is drawn either clockwise,
counterclockwise, towards the center or away from the center.  Hence,
using the metaphor of a cylinder, the edges point \emph{right},
\emph{left}, \emph{down} or \emph{up}, respectively. Moreover,
\emph{horizontal edges} point left or right, while \emph{vertical
  edges} point up or down; see Figure~\ref{fig:pre:drawing}.

\begin{figure}[t]
\centering
  \includegraphics{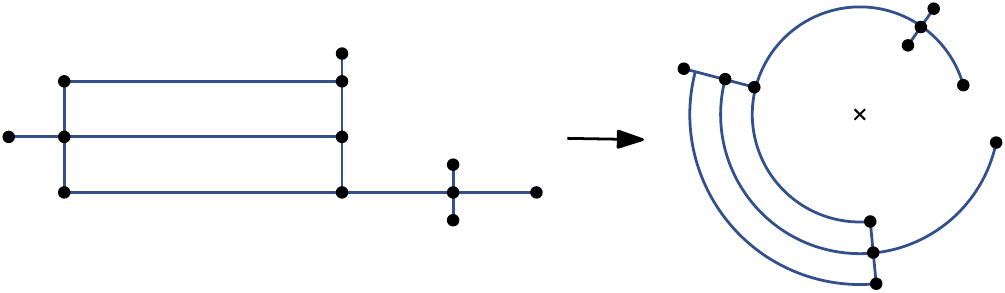}
  \caption{The orthogonal drawing is transformed into an ortho-radial drawing by \emph{bending} the horizontal edges into concentric circular arcs, while vertical edges become segments of rays that emanate from the center of the ortho-radial grid.}
  \label{fig:orthogonal-to-orthoradial}
\end{figure}

We further observe that if the central and outer face are identical
then an ortho-radial drawing can be intepreted as a distorted
orthogonal drawing, in which the horizontal edges are bended to
circular arcs, while the vertical edges remain straight segments; see
Figure~\ref{fig:orthogonal-to-orthoradial} for an example. Hence, utilizing the framework for
orthogonal drawings by Tamassia~\cite{t-emn-87}, this allows us to
easily create an ortho-radial drawing of an instance $I$ for the case
that the central face $f_c$ and the outer face~$f_o$ are the same. Hence, we assume
$f_c\neq f_o$ in the remainder of this work, which changes the problem
of finding an ortho-radial drawing of $I$ substantially.

\begin{figure}[t]
  \centering
  \begin{subfigure}[b]{0.32\textwidth}
      \includegraphics[page=1]{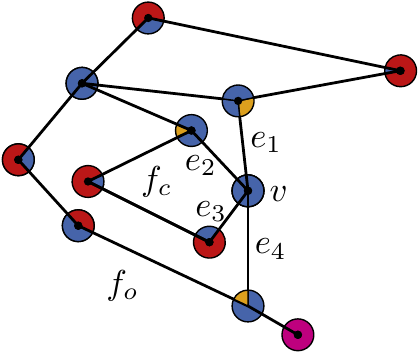}      
\end{subfigure}
    \begin{subfigure}[b]{0.32\textwidth}
\includegraphics[page=2]{./fig/prelim}
\end{subfigure}
    \begin{subfigure}[b]{0.10\textwidth}
\includegraphics[page=3]{./fig/prelim}
\end{subfigure}
\hfill
  \caption{A combinatorial embedding (left) and a ortho-radial drawing (right) of a graph. For each combinatorial angle a rotation is given.}
  \label{fig:prelim:combi}
\end{figure}

We first introduce concepts that help us to combinatorially describe
the ortho-radial drawing $\Delta$. Let $v$ be a vertex of $G$ and let
$\mathcal E(v)$ be the counterclockwise order of the edges in
$\mathcal E$ around $v$. A \emph{combinatorial angle at~$v$} is a pair
of edges~$(e_1,e_2)$ that are both incident to~$v$ and such that~$e_1$
immediately precedes~$e_2$ in~$\mathcal E(v)$; see Figure~\ref{fig:prelim:combi}.  An
\emph{angle assignment}~$\Gamma$ of an instance
$I=(G,\embedding,f_c,f_o)$ assigns to each combinatorial angle
$(e_1,e_2)$ of $\embedding$ a \emph{rotation}
$\rot(e_1,e_2) \in \{-2,-1,0,1\}$.  For an ortho-radial drawing $\Delta$
of $I$ we can derive an angle assignment that defines
$\rot(e_1,e_2) = 2-2\alpha/\pi$ for each angle~$(e_1,e_2)$ at~$v$,
where $\alpha$ is the counterclockwise geometric angle between~$e_1$
and~$e_2$ in~$\Delta$. Hence, the rotation of a combinatorial angle counts the
number of right turns that are taken when going from $e_1$ to $e_2$
via $v$, where negative numbers correspond to left turns; see
Figure~\ref{fig:prelim:combi}. In particular, in case that $e_1=e_2$ we derive
$\rot(e_1,e_2)=-2$ from $\Delta$, i.e., $v$ contributes two left
turns.  But conversely, we cannot derive an ortho-radial drawing from
every angle assignment.

For a face $f$ of $\embedding$ with facial walk $v_1\dots v_{k}$
around $f$ (where $f$ is oriented in clockwise order) we define
$\rot(f)=\sum_{i=1}^{k} \rot(v_{i-1}v_{i},v_{i}v_{i+1})$, where we define $v_0:=v_k$ and $v_{k+1}:=v_1$.  Every angle assignment~$\Gamma$ that is derived from a
bend-free ortho-radial drawing is locally consistent in the following
sense~\cite{hht-orthoradial-09}.

\begin{definition}\label{def:local-conditions}
  An angle assignment is \emph{locally consistent} if it
  satisfies the following two conditions.
\begin{compactenum}
\item\label{cond:repr:sum_of_angles} For each vertex, the sum of the
  rotations around~$v$ is~$2(\deg(v)-2)$.
\item\label{cond:repr:rotation_faces}
  For each face $f$, we have
  \[
  \rot(f)=
  \begin{cases}
  4, & \text{$f$ is a regular face} \\
  0, & \text{$f$ is the outer or the central face but not both} \\
  -4,& \text{$f$ is both the outer and the central face.} \\
  \end{cases}
  \]
\end{compactenum}
\end{definition}

\begin{figure}
  \centering
  \includegraphics[scale=1.2]{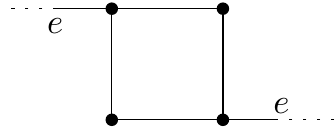}
  \caption{In this drawing, the angles around vertices sum up to $360^\circ$, and also the sum of angles for each face is as expected for an ortho-radial drawing. However, the graph does not have an ortho-radial drawing without bends.}
  \label{fig:intro:noteasy}
\end{figure} 

We call a locally consistent angle assignment of~$I$ an
\emph{ortho-radial representation} of $I$.  Unlike for orthogonal
representations Condition~(\ref{cond:repr:sum_of_angles}) and
Condition~(\ref{cond:repr:rotation_faces}) do not guarantee that for
an ortho-radial representation of $I$ there is an ortho-radial drawing
of $I$ having the same angles; see Figure~\ref{fig:intro:noteasy}.  In
this paper, we introduce a third more global condition that
characterizes all ortho-radial representations of $I$ that can be
drawn. To that end, we first introduce  basic concepts on rotations and
directions in ortho-radial representations in
Section~\ref{sec:rotations-directions}, which we then use to
define this global condition in
Section~\ref{sec:drawable-ortho-radial-represenations}.

\subsection{Rotations and Directions in
  Ortho-Radial Representations}\label{sec:rotations-directions}

We transfer two basic properties of ortho-radial drawings to
ortho-radial representations. First, the rotations of all cycles are
either $0$ or $4$. Second, fixing the geometric direction of a single
edge $e^\star$, fixes the geometric directions of all edges.  We call
$e^\star$ a \emph{reference edge} and assume that it points to the
right and lies on the outer face of $\embedding$.

For two edges $e=uv$ and $e'=vw$ we define the rotation between them
as $\rot(uvw)=\sum_{i=1}^{k-1}\rot(e_i,e_{i+1})-2(k - 2)$, where
$e=e_1,\dots,e_k=e'$ are the edges that are incident to $v$ and lie
between $e$ and $e'$ in counterclockwise order; see Figure~\ref{fig:prelim:rotation-path}(a). 

\begin{figure}[t]
  \centering
  \begin{subfigure}[t]{0.49\textwidth}
    \centering
     \includegraphics[page=4]{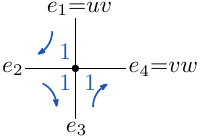}
     \caption{}
   \end{subfigure}
   \begin{subfigure}[t]{0.49\textwidth}
     \centering
     \includegraphics[page=5]{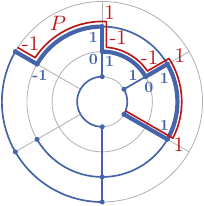}
     \caption{}
   \end{subfigure}
      \caption{Generalizaton of rotations. (a)~The rotation of the two edges $e_1=uv$ and $e_2=vw$ is $\rot(uvw)=\sum_{i=1}^3 1-2 (4-2) =-1$. (b)~The rotation of the path $P$ is $\rot(P)=-1+1-1-1+1+1=0$.}
      \label{fig:prelim:rotation-path}
\end{figure}

The rotation of a path $P=v_1\dots v_k$ is the sum of the rotations
at its internal vertices, that is
$\rot(P)=\sum_{i=2}^{k-1}\rot(v_{i-1}v_iv_{i+1})$; see Figure~\ref{fig:prelim:rotation-path}(b).
\begin{observation}\label{obs:rot_splitting_path}
  Let $P$ be a path with start vertex $s$ and end vertex
  $t$.
  \begin{compactenum}
  \item It is $ \rot(\reverse{P}) = -\rot(P)$.
  \item For every edge $e$ on $P$ it is $ \rot(P) =
    \rot(\subpath{P}{s, e}) + \rot(\subpath{P}{e,t})$.
\end{compactenum}
\end{observation}
Similarly, for a cycle $C = v_1\ldots v_k$, its rotation is the
sum of the rotations at all its vertices (where we define $v_0 = v_k$
and $v_{k+1}=v_1$), i.e.,
$\rot(C) = \sum_{i=1}^{k} \rot(v_{i-1}v_{i}v_{i+1})$. We observe that the
rotation of a face~$f$ is equal to the rotation of the cycle that we
obtain from the facial walk around $f$.

\begin{figure}[t]
      \centering
      \includegraphics{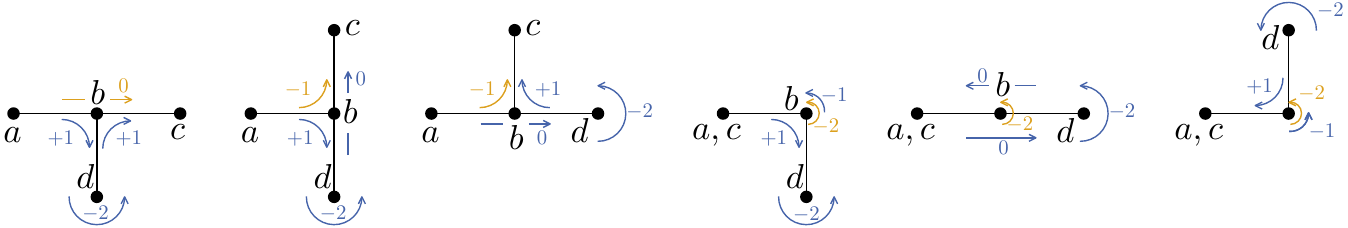}
      \caption{Illustration of proof for
        Lemma~\ref{lem:rotation-cycle}. The rotation $\rot(abc)$
        (orange) is equal to the rotation $\rot(abd)+\rot(dbc)-2$
        (blue) for the cases that $d\neq a,c$.}
      \label{fig:prelim:rotation-edges}
    \end{figure}

\begin{figure}[t]
  \centering
      \begin{subfigure}[c]{0.48\textwidth}
      \centering
      \includegraphics[page=2]{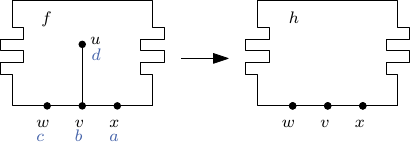}
      \caption{Rule 1}
      \label{fig:prelim:rules:rule1}
    \end{subfigure}
     \hfill
      \begin{subfigure}[c]{0.48\textwidth}
      \centering
      \includegraphics[page=3]{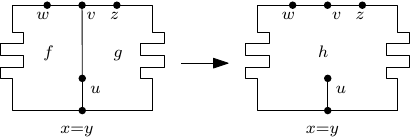}
      \caption{Rule 2}
      \label{fig:prelim:rules:rule2}
      \end{subfigure}
      \caption{Illustration of the rules introduced in proof for
        Lemma~\ref{lem:rotation-cycle}.~(\subref{fig:prelim:rules:rule1})~The
        degree-1 vertex~$u$ is
        removed.~(\subref{fig:prelim:rules:rule1})~The edge $uv$
        separating the faces $f$ and $g$ is removed. }
      \label{fig:prelim:rules}
\end{figure}

\begin{lemma}
  \label{lem:rotation-cycle}
  Let $C$ be a simple cycle in an ortho-radial representation $\Gamma$. Then,
  $\rot(C)=0$ if $C$ is
  essential and $\rot(C)=4$ if $C$ is non-essential.
\end{lemma}

\begin{proof}
  Let $H$ be the sub-graph of $G$ that is contained in the interior of
  $C$; we note that $C$ belongs to $H$. We remove $H^\circ=H-C$ from
  the ortho-radial representation by deleting the edges and vertices
  of $H^\circ$ successively. We show that in each deletion step we
  again obtain an ortho-radial representation; in particular
  Condition~\ref{cond:repr:sum_of_angles} and
  Condition~\ref{cond:repr:rotation_faces} hold. After $H^\circ$ has
  been completely removed, the cycle $C$ is a face of the resulting
  ortho-radial representation. By
  Condition~\ref{cond:repr:rotation_faces} the statement follows.
  
  We do the removal of $H^\circ$ based on two rules. The first rule removes
  a degree-1 vertex from $H^\circ$ adapting the angle assignment
  accordingly. The second rule removes an edge of $H^\circ$ that lies on a
  cycle of $H$ adapting the angle assignment accordingly.

  We show the following claims.

  \begin{Claim}\label{claim:rotation:applicable}
    If $H^\circ$ is not empty, then the first or second rule is applicable. 
  \end{Claim}

  \begin{Claim}\label{claim:rotation:first-rule}
    Applying the first rule yields a connected graph and an
    ortho-radial representation of this graph.
  \end{Claim}

  \begin{Claim}\label{claim:rotation:second-rule}
    Applying the second rule yields a connected graph and an
    ortho-radial representation of this graph.
  \end{Claim}

  We now prove Claim~\ref{claim:rotation:applicable}. Assume that the
  second rule is not applicable, but $H^\circ$ contains at least one
  vertex. We contract $C$ and the sub-graph in the exterior of $C$ to
  one vertex. As the second rule is not applicable, the result is a
  tree, which shows that there is a degree-1 vertex in $H^\circ$. Hence,
  the first rule is applicable.

  For proving Claim~\ref{claim:rotation:first-rule} and
  Claim~\ref{claim:rotation:second-rule} we first show a general
  statement on rotations. Let $a$, $b$, $c$ and $d$ vertices of $G$
  such that $d\neq c,a$ and there are the edges $ab$, $bc$ and
  $bd$. We assume that $ab$, $bc$ and $bd$ appear around $b$ in
  clockwise order. Figure~\ref{fig:prelim:rotation-edges} shows the six cases that are
  possible, from which we derive
  \begin{align}\label{lem:rotation:helping-equation}
    \rot(abc) = \rot(abd) +\rot(dbc)-2 
  \end{align}
  We now prove Claim~\ref{claim:rotation:first-rule}.  Let $u$ be a
  degree-1 vertex and let $v$ be the adjacent vertex of $u$; see
  Figure~\ref{fig:prelim:rules:rule1}. The edge $uv$ lies on a face
  $f$. Let $xv$ and $vw$ be the preceding and succeeding edges of $uv$
  on $f$, respectively. We note that, possibly, $x=w$. Let $h$ be the
  new face after deleting $v$. As $u$ has degree one, the resulting
  graph after deleting $u$ is still connected. Further, we define the
  deletion of $u$ such that the resulting angle assignment is locally
  correct, i.e., such that it satisfies
  Condition~\ref{cond:repr:sum_of_angles} at $v$.  In order to prove
  Condition~\ref{cond:repr:rotation_faces} we apply
  Equation~\eqref{lem:rotation:helping-equation} with $a=x$, $b=v$,
  $c=w$ and $d=u$ as follows.
  \begin{align*}
    \rot(f) = \rot(f[vw,xv])+\rot(xvu)-2+\rot(uvw)
             = \rot(f[vw,xv])+\rot(xvw) = \rot(h)
  \end{align*}
  If $f$ is a regular face, then $h$ is also a regular face. Hence,
  we obtain $\rot(h)=4$.  If $f$ is the central face, then $h$ is
  also the central face so that $\rot(h)=0$.

  Finally, we prove Claim~\ref{claim:rotation:second-rule}. Let $uv$
  be the edge that is removed; see
  Figure~\ref{fig:prelim:rules:rule2}. As $uv$ lies on a cycle in $H$,
  it separates two faces $f$ and $g$. We assume that $f$ lies locally
  to the left of $uv$ and $g$ lies locally to the right of $uv$. Let
  $w$ and $x$ be the preceding and succeeding vertex on $f$,
  respectively. Further, let $y$ and $z$ be the preceding and
  succeeding vertex on $g$. We note that possibly~ $w=z$ or~$x=y$.
  Let $h$ be the new face after deleting $uv$, whose boundary consists of the
  two paths $f[ux,wv]$ and $g[vz,yu]$.  As $uv$ lies on a cycle, the
  resulting graph after deleting $uv$ is still connected.  Further, we
  define the deletion of $uv$ such that the resulting angle assignment
  is locally correct, i.e., such that it satisfies
  Condition~\ref{cond:repr:sum_of_angles} at $u$ and $v$.  We prove
  Condition~\ref{cond:repr:rotation_faces} by showing
  \begin{align}\label{lem:rotation:helping-equation2}
    \rot(h) = \rot(f)+\rot(g)-4.
   \end{align}
   Since $f$ and $g$ lie in the interior of the essential cycle $C$,
   neither $f$ nor $g$ is the outer face.  If both are regular faces,
   then $h$ is also regular. From
   Equation~\eqref{lem:rotation:helping-equation2} we correctly obtain
   $\rot(h)=4$.  If one of them is the central face, then $h$ is the
   new central face. From
   Equation~\eqref{lem:rotation:helping-equation2} we obtain
   $\rot(h)=0$.  In the remainder of the proof we show
   $\rot(h) = \rot(f)+\rot(g)-4$.  For $f$, $g$ and $h$ we obtain the
   following rotations.
   \begin{align*}
     \rot(f) &= \rot(f[ux,wv])+\rot(wvu)+\rot(vux)\\
     \rot(g) &= \rot(g[vz,yu])+\rot(yuv)+\rot(uvz)\\
     \rot(h) &= \rot(f[ux,wv])+\rot(wvz)+\rot(g[vz,yu])+\rot(yux)
   \end{align*}
   Replacing $\rot(f[ux,wv])$ and $\rot(g[vz,yu])$ in the last
   equation, we obtain the next equation.
   \begin{align*}
     \rot(h) =& \rot(f)-\rot(wvu)-\rot(vux)+\rot(wvz)+\\
              & \rot(g)-\rot(yuv)-\rot(uvz)+\rot(yux)\\
             =& \rot(f)+\rot(wvz)-\rot(wvu)-\rot(uvz)+\\
              & \rot(g)+\rot(yux)-\rot(yuv)-\rot(vux)
   \end{align*}
   Applying Equation~\eqref{lem:rotation:helping-equation} twice, we replace
   $\rot(wvz)-\rot(wvu)-\rot(uvz)$ and $\rot(yux)-\rot(yuv)-\rot(vux)$
   with $-2$ each, which yields $\rot(h) = \rot(f)+\rot(g)-4$ as
   desired.
\end{proof}

  \begin{figure}[t]
  \centering
  \includegraphics[page=4]{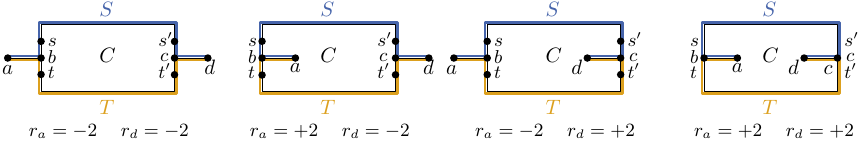}
  \caption{Illustration for proof of
    Lemma~\ref{lem:direction_cycle}. Cycle $C$ with two
    additional edges $ab$ and $cd$. If $a$ lies in the exterior of
    $C$ then $r_a=+2$ and otherwise $r_a=-2$. Similarly, if $d$ lies in the 
    exterior of $C$ then $r_d=+2$ and $r_d=-2$ otherwise.}
  \label{fig:measure-direction:cycle}
\end{figure}

The next lemma relates the rotations of two paths~$S$ and $T$ that use the same 
edges except on a cycle~$C$; see Figure~\ref{fig:measure-direction:cycle}.

\begin{lemma}\label{lem:direction_cycle}
  Let $C$ be a cycle and let $ab$ and $cd$ be two edges (with
  $b\neq c$) such that $b$ and $c$ lie on $C$, but $a$ and $d$ do not.
  Further, let $S=ab+C[b,c]+cd$ and $T=ab+\reverse{C}[b,c]+cd$. Then
  \[\rot(S)-\rot(T)=\rot(C)+r_a+r_d\,,\] where  for $z\in \{a,d\}$ we
  define $r_z=+2$ if $z$ lies in the interior of $C$ and $r_z=-2$ if
  $z$ lies in the exterior of $C$.
\end{lemma}
\begin{proof}
 \begin{figure}[t]
   \centering
   \includegraphics[page=5]{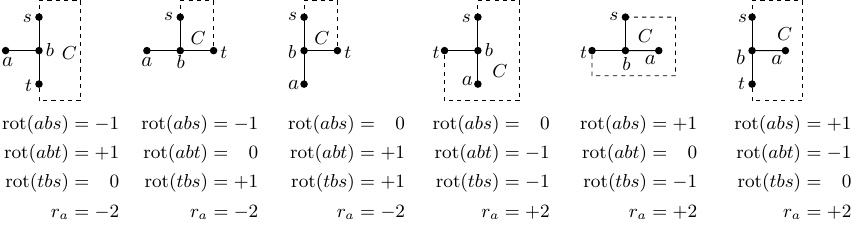}
   \caption{{Illustration for proof of
       Lemma~\ref{lem:direction_cycle}}. In all six cases it holds $\rot(abs) 
       -
     \rot(abt) = \rot(tbs) + r_a$}
   \label{fig:measure-direction:cycle:cases}
 \end{figure}
 
 Let $t$ and $s$ be the vertices immediately before and after~$b$ on~$C$. We
 observe that $a$ lies in the interior of $C$ if and only if $ab$ lies locally
 to the right of the path $sbt$. Considering all six cases how the edges $ab$,
 $sb$, and $tb$ can be arranged (see 
 Figure~\ref{fig:measure-direction:cycle:cases}), we obtain $\rot(abs) -
 \rot(abt) = \rot(tbs) + r_a$.
 Similarly, we define $s'$ and $t'$ as the vertices immediately before and
 after $c$ on $C$. Considering all cases as above we get $\rot(s'cd) -
 \rot(t'cd) = \rot(s'bt') + r_d$.
 
 Splitting $S$ and $T$ into three parts (see 
 Figure~\ref{fig:measure-direction:cycle}), we have
 \begin{align*}
 \rot(S) - \rot(T) &= \rot(abs) + \rot(C[b, c]) + \rot(s'cd)
 -\rot(abt) - \rot(\reverse{C}[b,c]) - \rot(t'cd).
 \end{align*}
 Combining the rotations at $b$ and $c$ using the observations from above, we
 get
 \begin{align*}
 \rot(S) - \rot(T) &= \rot(tbs) + r_a + \rot(C[b, c]) + \rot(C[c,b]) +
 \rot(s'ct') + r_d \\
 &= \rot(C) + r_a + r_d.
 \end{align*}
\end{proof}

\begin{figure}[t]
  \centering
      \begin{subfigure}[c]{\textwidth}
      \centering
      \includegraphics[page=1]{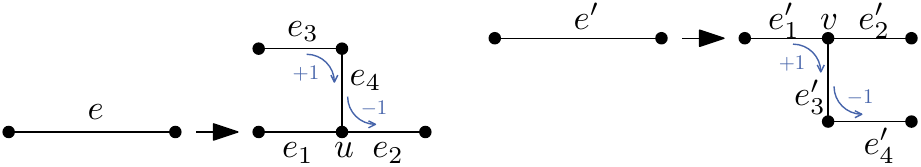}
      \caption{}
      \label{fig:measure-direction:construction}
    \end{subfigure}
    
      \begin{subfigure}[c]{\textwidth}
      \centering
      \includegraphics[page=2]{fig/lemma_directions_to_labels.pdf}
      \caption{}
      \label{fig:measure-direction:paths}
      \end{subfigure}
      \caption{Illustration of the proof for        Lemma~\ref{lem:measure-direction}.~(\subref{fig:measure-direction:construction})~The edges $e$ and $e'$ are replaced by the depicted construction to reduce the number of cases to be considered.~(\subref{fig:measure-direction:paths})~Four cases are considered how the edges $e$ and $e'$ can be connected by the path $S$. The paths $S_1$ and $S_2$ are defined depending on the particular case.}
      \label{fig:measure-direction}
\end{figure}

For two edges $e$ and $e'$ let $P$ be an arbitrary path that starts at
the source or target of $e$ and ends at the source or target of $e'$,
and that neither contains $e$ nor $e'$.  We call $P$ a \emph{reference
  path} from $e$ to $e'$. We define the \emph{combinatorial direction}
of $e'=xy$ with respect to $e=uv$ and $P$ as
\[
  \dir(e,P,e') = \begin{cases}
    \rot(e+P+e') & \text{$P$ starts at $v$ and ends at $x$},\\
    \rot(\reverse{e}+P+e')+2 & \text{$P$ starts at $u$ and ends at $x$},\\
    \rot(e+P+\reverse{e'})-2 & \text{$P$ starts at $v$ and ends at $y$},\\
    \rot(\reverse{e}+P+\reverse{e'}) & \text{$P$ starts at $u$ and ends at $y$.}
  \end{cases}
\]

With the fixed direction of the reference edge~$e^\star$, it is
natural to determine the direction of any other edge $e$ by
considering the direction of any reference path from $e^\star$ to
$e$. In order to get consistent results, any two reference paths $P$
and $Q$ from $e^\star$ to $e$ must induce the same direction of~$e$,
which means that $\dir(e^\star, P, e)$ and $\dir(e^\star, Q, e)$ may
only differ by a multiple of 4.  In the following lemma we show that
this is indeed the case.

\begin{lemma}\label{lem:measure-direction}
  Let $e$ and $e'$ be two edges of an ortho-radial representation
  $\Gamma$, and let $P$ and $Q$ be two reference paths from $e$ to
  $e'$.
  \begin{compactenum}
  \item It holds $\dir(e,P,e')\equiv \dir(e,Q,e') \pmod{4}$.
  \item It holds $\dir(e,P,e') = \dir(e,Q,e')$, if there are two essential cycles $C$ and $C'$ such that
    \begin{compactenum}
    \item $C'$ lies in the interior of $C$,
    \item $e$ lies on $C$ and $e'$ lies on $C'$, and 
    \item $P$ and $Q$ lie in the interior of $C$ and in the exterior of $C'$.    
    \end{compactenum}
  \end{compactenum} 
\end{lemma}

\begin{proof}
  First, we define a construction that helps us to reduce the
  number of cases to be considered. We subdivide $e$ by a vertex $u$
  into two edges $e_1$ and $e_2$; see
  Figure~\ref{fig:measure-direction:construction}. Further, we add a
  path consisting of two edges $e_3$ and $e_4$ such that the target of
  $e_4$ is $u$. We define that $\rot(e_3e_4)=1$ and
  $\rot(e_4e_2)=-1$. Similarly, we subdivide $e'$ by a vertex $v$ into
  two edges $e'_1$ and $e'_2$. Further, we add a path consisting of
  two edges $e'_3$ and $e'_4$ such that the source of $e'_3$ is
  $v$. We define that $\rot(e'_1e'_3)=1$ and $\rot(e'_3e'_4)=-1$.  Let
  $S$ be a reference path from $e$ to $e'$; see
  Figure~\ref{fig:measure-direction:paths}. Let $S_1$ be the path that
  starts at the source of $e_3$ and ends at the starting point of $S$
  only using edges from~$\{e_1,e_2,e_3,e_4\}$. Further, let $S_2$ be
  the path that starts at the end point of $S$ and ends at the target of
  $e'_4$ only using edges from~$\{e'_1,e'_2,e'_3,e'_4\}$. The
  \emph{extension} $S'$ of $S$ is the path $S_1+S+S_2$. The following
  claim shows that we can consider $S'$ instead of $S$ such that it is
  sufficient to consider the rotation of $S'$ instead of the
  direction~$\dir(e,S,e')$, which distinguishes four cases.
  
  \begin{Claim}\label{claim:directions:simplifications}
    $\dir(e,S,e')=\rot(S')$
  \end{Claim}

  The detailed proof of Claim~\ref{claim:directions:simplifications}
  is found at the end of this proof. In the following let $P'$ and
  $Q'$ be the extensions of $P$ and $Q$, respectively. We show that
  $\rot(P')\equiv \rot(Q') \pmod{4}$.  Moreover, for the case that
  $e=e^\star$ and $e'$ lies on an essential cycle that is respected by
  $P$ and $Q$, we show that $\rot(P') = \rot(Q')$. Altogether, due to
  Claim~\ref{claim:directions:simplifications} this proves
  Lemma~\ref{lem:measure-direction}.  We show $\rot(P') = \rot(Q')$ by
  converting $P'$ into $Q'$ successively. More precisely, we construct
  paths $R_1\dots R_k$ such that $R_i$ consists of a prefix of
  $P'$ followed by a suffix of $Q'$ such that with increasing $i$ the
  used prefix of $P'$ becomes longer, while the used suffix of $Q'$
  becomes shorter. In particular, we have $R_1=Q'$ and $R_k=P'$. We
  show that $\rot(R_i)\equiv\rot(R_{i+1}) \pmod{4}$.  If
  $R_{i}\neq P'$, we construct $R_{i+1}$ from $R_i$ as follows; see
   Figure~\ref{fig:measure-direction:R}.

  \begin{figure}[t]
  \centering
      \includegraphics[page=3]{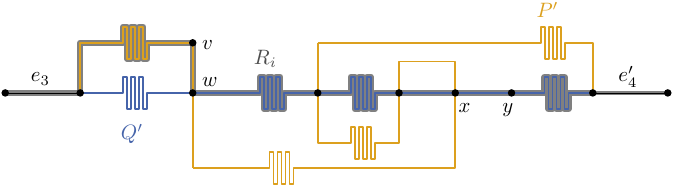}
      \caption{Illustration for proof of
        Lemma~\ref{lem:measure-direction}. The path $R_i$ (gray)
        consists of a prefix of the path $P'$ (orange) and a suffix of
        the path $Q'$ (blue).}
    \label{fig:measure-direction:R}
  \end{figure}
  
  There is a first edge $vw$ on $P'$ such that the following edge does 
  not lie on $R_{i}$. Let~$x$ be the first vertex on $P'$ after $w$
  that lies on $R_i$ and let $y$ be the vertex on $R_i$ that follows
  $x$ immediately.  As both $P'$ and $R_i$ end at the same edge, these
  vertices always exist. We define $R_{i+1} = P[e_3,x]+R_i[x,e_4']$.  We
  observe that $R_i[x,e_4']=Q'[x,e_4']$, as $R_i[w,e_4']=Q'[w,e_4']$ and $x$
  occurs on $Q'$ after $w$. Further, $R_{i+1}$ is a path as we can argue
  as follows.  We can decompose $R_{i+1}$ into three paths:
  $R^1_{i+1}=P'[e_3,w]=R_{i}[e_3,w]$, $R^2_{i+1}=P'[w,x]$ and
  $R^3_{i+1}=Q'[x,e_4']=R_{i}[x,e_4']$. The paths $R^1_{i+1}$ and
  $R^3_{i+1}$ do not intersect as both also belong to $R_i$, which is
  a path by induction. The paths $R^1_{i+1}$ and $R^2_{i+1}$ do not
  intersect (except at their common vertex $w$), because both belong
  to $P'$. The paths $R^2_{i+1}$ and $R^3_{i+1}$ do not intersect
  (except at their common vertex $x$), because by the definition of
  $x$ no vertex of $P'$ between $w$ and $x$ lies on $R_i$.

  Next, we show that $\rot(R_i)\equiv\rot(R_{i+1}) \pmod{4}$.  To that
  end consider the cycle $C_{i}$ that consists of the two paths
  $R_i[w,x]$ and $R_{i+1}[w,x]$. We orient $C_{i}$ such that the
  interior of the cycle locally lies to the right of its edges.  By
  the definition of $R_{i+1}$ we obtain
  $\rot(R_i[e_3,vw]) = \rot(R_{i+1}[e_3,vw])$ and
  $\rot(R_i[xy,e_4']) = \rot(R_{i+1}[xy,e_4'])$, as these subpaths of
  $R_i$ and $R_{i+i}$ coincide, respectively. Hence, it remains to
  show that $\rot(R_i[vw,xy]) \equiv \rot(R_{i+1}[vw,xy]) \pmod{4}$
  and $\rot(R_i[vw,xy]) = \rot(R_{i+1}[vw,xy])$ in the special case
  that $e = e^\star$ and $e'$ lies on an
  essential cycle that is respected by $P$ and $Q$.

  In general we can describe the obtained situation as follows. We are
  given a cycle and two edges $ab$ and $cd$ (with $b\neq c$) such
  that $b$ and $c$ lie on that cycle, but $a$ and $d$ not; see also 
  Figure~\ref{fig:measure-direction:cycle}. Hence, we can apply 
  Lemma~\ref{lem:direction_cycle}.

  We distinguish two cases: if the interior of
  $C_i$ lies locally to the right of $R_i$ we define 
  $S=vw+R_i[w,x]+xy$ and $T=vw+R_{i+1}[w,x]+xy$, and otherwise
  $S=vw+R_{i+1}[w,x]+xy$ and $T=vw+R_{i}[w,x]+xy$.  We only consider
  the first case, as the other case is symmetric. By
  Lemma~\ref{lem:direction_cycle} we obtain
  $\rot(S)-\rot(T)=\rot(C_i)+r_v+r_y$. As $\rot(C_i)\equiv 0 \pmod{4}$
  and $r_v,r_y\equiv 2 \pmod{4}$ we obtain
  $\rot(S)\equiv \rot(T) \pmod{4}$ and with this
  $\rot(R_i[vw,xy]) \equiv \rot(R_{i+1}[vw,xy]) \pmod{4}$. Altogether,
  in the general case we obtain $\rot(R_i) \equiv \rot(R_{i+1})$.

  Finally, we prove the second statement of the lemma.  Hence, there
  are two essential cycles $C$ and $C'$ such that
  \begin{compactenum}
    \item $C'$ lies in the interior of $C$,
    \item $e$ lies on $C$ and $e'$ lies on $C'$, and 
    \item $P$ and $Q$ lie in the interior of $C$ and in the exterior of $C'$.    
  \end{compactenum}
  In particular, the paths $P$ and $Q$ respect $C'$ as they lie in the
  exterior of $C'$.  We show that $\rot(S)=\rot(T)$. First, we observe
  that $R_i[e,vw]$ respects the cycle $C_i$ by the simplicity
  of $R_i$ and $R_{i+1}$. In particular, $vw$ lies in the exterior of $C_i$ so
  that $r_v=-2$.  We distinguish the two cases whether $C_i$ is
  essential or non-essential.  If $C_i$ is a non-essential cycle, the
  edge $xy$ is also contained in the exterior of $C_i$ as $P$ and $Q$
  end on the essential cycle~$C'$, which both respect. Hence, we
  obtain $r_y=-2$. Thus, we get by Lemmas~\ref{lem:rotation-cycle}
  and~\ref{lem:direction_cycle} that
  $\rot(S)-\rot(T)=\rot(C_i)+r_v+r_y=4-2-2=0$.  If $C_i$ is an
  essential cycle, then the cycle $C'$ is contained in the interior of
  $C_i$ as both contain the central face, and $C_i$ is composed by
  parts of paths that respect $C'$. Consequently, the edge $xy$ lies
  in the interior of $C_i$ so that we obtain $r_y=2$. By
  Lemmas~\ref{lem:rotation-cycle} and~\ref{lem:direction_cycle} we get
  $\rot(S)-\rot(T)=\rot(C_i)+r_v+r_y=0-2+2=0$.  It remains to prove
  Claim~\ref{claim:directions:simplifications}.

  \textit{Proof of Claim~\ref{claim:directions:simplifications}.} We
  distinguish the four cases given by the definition of
  $\dir(e,S,e')$; see Figure~\ref{fig:measure-direction:paths}.  If $S$ starts at the target of $e$
  and ends at the source of $e'$, we obtain
  $\dir(e,S,e')=\rot(e+S+e')=\rot(e_2+S+e'_1)$ as subdividing $e$ and
  $e'$ transfers the directions of $e$ and $e'$ to $e_2$ and $e'_1$,
  respectively. Hence, we obtain
  \begin{align*}
    \rot(S') & = \rot(S_1 + S +  S_2) \\
             & = \rot(e_3+e_4+e_2+ S + e'_1 + e'_3 +e'_4) \\
             & = \rot(e_3+e_4+e_2)+\rot(e_2+ S + e'_1)+ \rot(e'_1 + e'_3 +e'_4)\\
             & = 0 + \dir(e,S,e') + 0  = \dir(e,S,e'). 
  \end{align*}
  If $S$ starts at the source of $e$ and ends at the source of $e'$, we
  obtain
  $\dir(e,S,e')=\rot(\reverse{e}+S+e')+2=\rot(\reverse{e_1}+S+e'_1)+2$
  and with this we obtain
  \begin{align*}
    \rot(S') & = \rot(S_1 + S +  S_2) \\
             & = \rot(e_3+e_4+\reverse{e_1}+ S + e'_1 + e'_3 +e'_4) \\
             & = \rot(e_3+e_4+\reverse{e_1})+\rot(\reverse{e_1} + S + e'_1)+ \rot(e'_1 + e'_3 +e'_4)\\
             & = 2 + \dir(e,S,e')-2 + 0  = \dir(e,S,e'). 
  \end{align*}
  If $S$ starts at the target of $e$ and ends at the target of $e'$, we
  obtain
  $\dir(e,S,e')=\rot(e+S+\reverse{e'})-2=\rot(e_2+S+\reverse{e'_2})-2$
  and with this we obtain
    \begin{align*}
    \rot(S') & = \rot(S_1 + S +  S_2) \\
             & = \rot(e_3+e_4+e_2+ S + \reverse{e'_2} + e'_3 +e'_4) \\
             & = \rot(e_3+e_4+e_2)+\rot(e_2+ S + \reverse{e'_2})+ \rot(\reverse{e'_2} + e'_3 +e'_4)\\
             & = 0 + \dir(e,S,e')+ 2 - 2  = \dir(e,S,e'). 
    \end{align*}
  If $S$ starts at the source of $e$ and ends at the target of $e'$, we
  obtain
  $\dir(e,S,e')=\rot(\reverse{e}+S+\reverse{e'})=\rot(\reverse{e_1}+S+\reverse{e'_2})$
  and with this we obtain
    \begin{align*}
      \rot(S') & = \rot(S_1 + S +  S_2) \\
               & = \rot(e_3+e_4+\reverse{e_1}+ S + \reverse{e'_2} + e'_3 +e'_4) \\
               & = \rot(e_3+e_4+\reverse{e_1})+\rot(\reverse{e_1}+ S + \reverse{e'_2})+ \rot(\reverse{e'_2} + e'_3 +e'_4)\\
               & = 2 + \dir(e,S,e') - 2  = \dir(e,S,e'). 
  \end{align*}
  Altogether, this shows the claim $\dir(e,S,e')=\rot(S')$.
\end{proof}

\begin{corollary}\label{cor:measure-direction}
  If $e$ is the reference edge
  $e^\star$ and $e'$ lies on an essential cycle that is respected by
  $P$ and $Q$, then $\dir(e,P,e') = \dir(e,Q,e')$.
\end{corollary}

\begin{proof}
  The statement directly follows from the second statement of
  Lemma~\ref{lem:measure-direction} by assuming that $C$ is the outermost essential cycle
  of $\Gamma$ and $C'$ is the cycle containing $e'$.
\end{proof}

Using this result, the geometric directions of all edges of a given
ortho-radial representation $\Gamma$ can be determined as follows.
Let $P$ be any reference path from the reference edge $e^\star$ to any
edge $e$, the edge $e$ points right, down, left, and up if
$\dir(e^\star, P, e)$ is congruent to 0, 1, 2, and 3,
respectively. Lemma~\ref{lem:measure-direction} ensures that the
result is independent of the choice of the reference path.  In fact,
Lemma~\ref{lem:measure-direction} even gives a stronger result as we
can infer the geometric direction of one edge from the geometric
direction of another edge locally without having to resort to paths to
the reference edge. We often implicitly make use of
this observation in our proofs.

 \begin{figure}[t]
  \centering
  \includegraphics[]{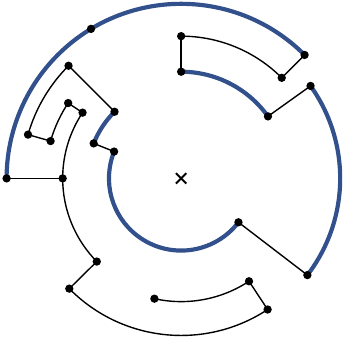}
  \caption{The outer face of an ortho-radial drawing. All outlying
    edges are marked blue.}
  \label{fig:outlying-edges}
 \end{figure}

\subsection{Drawable Ortho-Radial Representations}\label{sec:drawable-ortho-radial-represenations}
 
In this section we introduce concepts that help us to characterize the
ortho-radial representations that have an ortho-radial drawing. To
that end, consider an arbitrary bend-free ortho-radial
drawing~$\Delta$ of a plane 4-graph~$G$. As we assume throughout this
work that the outer and central face are not the same, there is an
essential cycle $C$ that lies on the outer face~$f$ of $\Delta$. Let
$e$ be a horizontal edge of $C$ that points to the right and that lies
on the outermost circle of the ortho-radial grid among all such edges,
and let $R_\Delta$ be the set of all edges $e'$ with $\dir(e,P,e')=0$
for any path $P$ on $f$; see Figure~\ref{fig:outlying-edges}. We
observe that $R_\Delta$ is independent of the choice of $e$, because
if there are multiple choices for $e$, then all of these edges are
contained in $R_\Delta$. We call the edges in $R_\Delta$ the
\emph{outlying edges} of $\Delta$. Without loss of generality, we
require that the reference edge of the according ortho-radial
representation~$\Gamma$ stems from~$R_\Delta$; as $R_\Delta$ is not
empty and all of the edges in $R_\Delta$ are possible candidates for
being a reference edge, we can always change the reference edge of
$\Gamma$ to one of the edges in $R_\Delta$. These outlying edges
possess the helpful properties that they all lie on the outermost
essential cycle, which makes them the ideal choice as reference edges.

An ortho-radial representation $\Gamma$ of a graph $G$ with reference
edge $e^\star$ is \emph{drawable} if there exists a bend-free
ortho-radial drawing $\Delta$ of $G$ embedded as specified by $\Gamma$
such that the corresponding angles in $\Delta$ and $\Gamma$ are equal
and the edge $e^\star$ is an outlying edge, i.e.,
$e^\star\in R_\Delta$.
Unlike for orthogonal representations
Condition~(\ref{cond:repr:sum_of_angles}) and
Condition~(\ref{cond:repr:rotation_faces}) do not guarantee that the
ortho-radial representation is drawable; see Figure~\ref{fig:intro:noteasy}.
Therefore, we introduce a
third condition, which is formulated in terms of labelings of
essential cycles.

Let $e$ be an edge on an essential cycle $C$ in $G$ and let $P$ be a
reference path from the reference edge $e^\star$ to $e$ that respects
$C$.  We define the \emph{label} of $e$ on $C$ as
$\ell_C(e)=\dir(e^\star,P,e)$.  By
Corollary~\ref{cor:measure-direction} the label $\ell_C(e)$ of $e$
does not depend on the choice of $P$.  We call the set of all labels
of an essential cycle its \emph{labeling}.

We call an essential cycle \emph{monotone} if either all its labels
are non-negative or all its labels are non-positive.  A monotone cycle
is a \emph{decreasing} cycle if it has at least one strictly positive
label, and it is an \emph{increasing} cycle if it has at least one
strictly negative label.  We also refer to increasing and decreasing
cycles as \emph{strictly monotone}.  An ortho-radial representation is
\emph{valid} if it contains no strictly monotone cycles.
The validity of an ortho-radial representation ensures that on each
essential cycle with at least one non-zero label there is at least one
edge pointing up and one pointing down.

A main goal of this paper is to show that a graph with a given
ortho-radial representation can be drawn if and only if the
representation is valid. Further, we show how to test validity of a given representation and how to obtain a bend-free
ortho-radial drawing from a valid ortho-radial representation in
quadratic time. Altogether, this yields our main results:

\begin{restatable*}{thm}{mrDrawable}\label{thm:main-result:drawable}
  An ortho-radial representation is drawable if and only if it is
  valid. 
\end{restatable*}

\begin{restatable*}{thm}{mrValidity}\label{thm:main-result:validity}
    Given an ortho-radial representation $\Gamma$, it can be determined
  in $\O(n^2)$ time whether $\Gamma$ is valid.  In the negative case a
  strictly monotone cycle can be computed in $\O(n^2)$  time.
\end{restatable*}

\begin{restatable*}{thm}{mrDrawing}\label{thm:main-result:drawing}
 Given a valid ortho-radial representation, a corresponding drawing can be constructed in $\O(n^2)$ time.
\end{restatable*}

We prove the three theorems in the given
order. Section~\ref{sec:transformations}--\ref{sec:characterization-rect}
deals with Theorem~\ref{thm:main-result:drawable}. In
Section~\ref{sec:rectangulation} we prove
Theorem~\ref{thm:main-result:validity}. In particular, together with
the proof of Theorem~\ref{thm:main-result:drawable} this already leads
to a version of Theorem~\ref{thm:main-result:drawing}, but with a
running time of $O(n^4)$. In Section~\ref{sec:efficient_rectangulation} we show how to achieve
$O(n^2)$ running time proving Theorem~\ref{thm:main-result:drawing}.

\section{Transformations of Ortho-Radial Representations}\label{sec:transformations}
In this section we introduce helpful tools that we use throughout this work. In the remainder of this work we assume
that we are given an ortho-radial representation~$\Gamma$ with
reference edge $e^\star$.

Since the reference edge lies on an essential cycle by definition, we can
compute the labelings of essential cycles via the rotation of paths as shown in
the following lemma. This simplifies the arguments of our proofs.

\begin{figure}[t]
  \centering
      \begin{subfigure}[c]{0.48\textwidth}
      \centering
      \includegraphics[page=1]{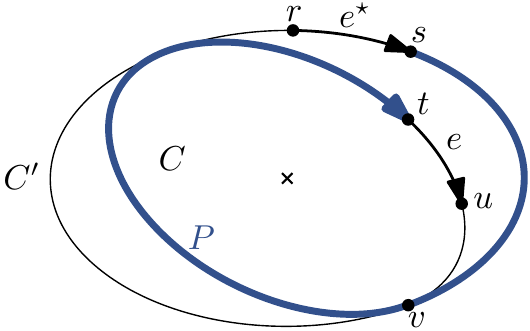}
      \caption{}
      \label{fig:reference-path:common-vertices}
    \end{subfigure}
      \begin{subfigure}[c]{0.48\textwidth}
      \centering
      \includegraphics[page=2]{fig/lemma_reference_path.pdf}
      \caption{}
      \label{fig:reference-path:disjoint}
      \end{subfigure}
      \caption{Illustration of proof for
        Lemma~\ref{lem:label-simple-path}. The outermost cycle $C'$
        contains the essential cycle $C$. There is always a reference
        path from $s$ on $C'$ to $t$ on $C$ respecting
        $C$.~(\subref{fig:reference-path:common-vertices}) The cycles
        $C$ and $C'$ have vertices in
        common.~(\subref{fig:reference-path:disjoint}) The cycles $C$
        and $C'$ are disjoint.  }
      \label{fig:reference-path}
\end{figure}

\begin{lemma}
  \label{lem:label-simple-path}
  For every edge $e$ on
  an essential cycle $C$ there is a reference path~$P$ from $e^\star$
  to $e$ such that $P$ respects $C$, starts at the target of $e^\star$
  and ends on the source of $e$. Moreover,
  $\ell_C(e)=\rot(e^\star+P+e)$.
\end{lemma}
\begin{proof}
  Let $Q$ be a reference path from the reference edge~$e^\star=rs$ to
  $e=tu$ respecting $C$. We construct the desired reference path $P$
  as follows.  Let $C'$ be the outermost essential cycle, i.e., $C'$
  is the unique essential cycle such that every edge of $C'$ bounds
  the outer face.  If $C$ and $C'$ have a common vertex, we define $v$
  to be the first common vertex on $C'$ after $s$; see
  Figure~\ref{fig:reference-path:common-vertices}. We set
  $P=C'[s,v] + C[v, t]$. By the choice of $v$ this concatenation is a 
  path. Moreover, it is a reference path from $e^\star$ to $e$ that
  respects $C$.  If $C$ and $C'$ are disjoint, let $v$ be the first
  vertex of $Q$ lying on $C$ and let $w$ be the last vertex of $Q$
  before $v$ that lies on $C'$; see
  Figure~\ref{fig:reference-path:disjoint}.  We set
  $P=C'[s,w] + Q[w,v] + C[v,t]$. We observe that the concatenation $P$
  is a path by the choice of $v$ and $w$, and that $P$ is a reference
  path from $e^\star$ to $e$ respecting $C$.
  
  By Corollary~\ref{cor:measure-direction} we have $\ell_{C}(e) = \dir(e^\star, P, 
  e)$ since $P$ is a reference path from $e^\star$ to $e$ respecting $C$.
  Further, since $P$ starts at the target $s$ of the reference edge $e^\star$ 
  and ends at the source $t$ of the edge $e$, we can express the direction of 
  $P$ as $\dir(e^\star, P, e)=\rot(e^\star + P + e)$. 
\end{proof}

The next lemma shows how we can change the reference edge $e^\star$ on
the essential cycle $C_o$ that is part of the outer face;
Figure~\ref{fig:reference-edge:change} illustrates the lemma.

\begin{figure}[t]
  \centering
      \begin{subfigure}[c]{0.48\textwidth}
      \centering
      \includegraphics[page=1]{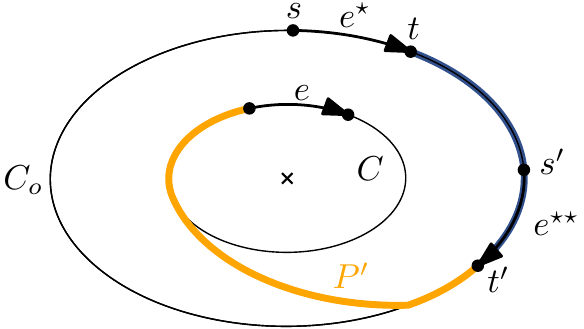}
      \caption{}
      \label{fig:reference-edge:change:case1}
    \end{subfigure}
      \begin{subfigure}[c]{0.48\textwidth}
      \centering
      \includegraphics[page=2]{fig/lemma_reference_edge.pdf}
      \caption{}
      \label{fig:reference-edge:change:case2}
      \end{subfigure}
      \caption{Illustration of proof for
        Lemma~\ref{lem:reference-edge:change}. The outermost cycle
        $C_o$ contains the essential cycle $C$. Further, there is a
        reference path $P$ from $e^\star$ to an edge $e$ on $C$.
        (\subref{fig:reference-path:common-vertices}) The path $P$
        contains the edge
        $e^{\star\star}$.~(\subref{fig:reference-path:disjoint}) The
        path $P$ does not contain the edge $e^{\star\star}$.  }
      \label{fig:reference-edge:change}
\end{figure}

\begin{lemma}\label{lem:reference-edge:change}
  Let $\Gamma$ be an ortho-radial representation and let $e^\star$ be
  the reference edge of $\Gamma$. Further, let $C_o$ be the essential
  cycle that lies on the outer face and let $e^{\star\star}$ be an edge on $C_o$ such
  that $\rot(C_o[e^\star,e^{\star\star}])=0$. For every edge $e$ on an essential cycle $C$ of $\Gamma$ it holds $\ell_C(e)=\overline{\ell}_C(e)$, where $\overline{\ell}_C$ is the labeling of $C$ with respect to $e^{\star\star}$.

  In particular, $\Gamma$ with reference edge $e^\star$ is valid if and only if $\Gamma$ with reference edge $e^{\star\star}$ is valid.
\end{lemma}

\begin{proof}
  Let $e$ be an edge of an essential cycle $C$ in $\Gamma$ and let $P$
  be a reference path from $e^\star$ to $e$. By
  Lemma~\ref{lem:label-simple-path} we assume that $P$ starts at the
  target of $e^\star$ and ends at the source of $e$. Without loss
  of generality we assume that only a prefix of $P$ is part of
  $C_o$. Further, let $e^\star=st$ and $e^{\star\star}=s't'$.

  If $P$ contains $e^{\star\star}$, then  
  \[\ell_C(e)=\rot(e^\star+P+e)=\rot(e^\star+P[t,e^{\star\star}])+\rot(e^{\star\star}+P'+e]),\] where $P'$ is  the suffix of $P$ that
  starts at $t'$; see Figure~\ref{fig:reference-edge:change:case1}. By
  assumption $\rot(e^\star+P[t,e^{\star\star}])=0$ and
  $\overline{\ell}_C(e)=\rot(e^{\star\star}+P'+e])$.

  If $P$ does not contain $e^{\star\star}$, then $C_o[t',e^\star]+P$
  is a reference path of $e$ with respect to $e^{\star\star}$, which
  contains $e^\star$; see Figure~\ref{fig:reference-edge:change:case2}. Swapping $e^\star$ and $e^{\star\star}$ in the
  argument of the previous case yields the claim. Note that
  $\rot(C_o[e^{\star\star},e^\star])=0$ as $\rot(C_o)=0$.
\end{proof}

In our arguments we frequently exploit certain symmetries.  For
an ortho-radial representation $\Gamma$ we introduce two new
ortho-radial representations, its \emph{flip} $\flip{\Gamma}$ and its
\emph{mirror} $\mirror{\Gamma}$.  Geometrically, viewed as a drawing on a
cylinder, a flip corresponds to rotating the cylinder by $180\degree$
around a line perpendicular to the axis of the cylinder so that is
upside down, see Figure~\ref{fig:flipping-cylinder}, whereas mirroring
corresponds to mirroring it at a plane that is parallel to the axis of
the cylinder; see Figure~\ref{fig:mirroring-cylinder}.  Intuitively,
the first transformation exchanges left/right and top/bottom,
and thus preserves monotonicity of cycles, while the second
transformation exchanges left/right but not top/bottom, and
thus maps increasing cycles to decreasing ones and vice versa.
This intuition indeed holds with the correct definitions of
$\flip{\Gamma}$ and $\mirror{\Gamma}$, but due to the non-locality of
the validity condition for ortho-radial
representations and the dependence on a reference edge this requires
some care. The following two lemmas formalize flipped and mirrored
ortho-radial representations. We denote the reverse of a directed edge $e$ by
$\reverse{e}$.

  \begin{figure}[t]
    \begin{subfigure}[c]{0.35\textwidth}
      \centering
      \includegraphics[width=\textwidth]{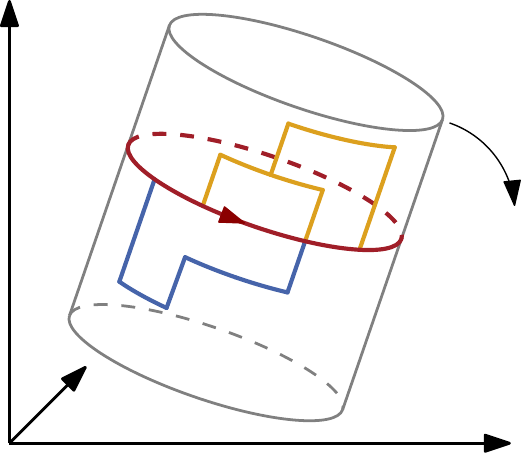}
\subcaption{Flipping the cylinder.}
    \end{subfigure} \hspace{5ex}
    \begin{subfigure}[c]{0.55\textwidth}
      \centering
      \includegraphics[page=2,width=\textwidth]{fig/flipping_cylinder.pdf}
\subcaption{The ortho-radial drawing before and after flipping.}
    \end{subfigure}
    \caption{Illustration of flipping the cylinder.}
   \label{fig:flipping-cylinder}
 \end{figure}
 
 \begin{figure}[t]
  \centering
  \includegraphics[]{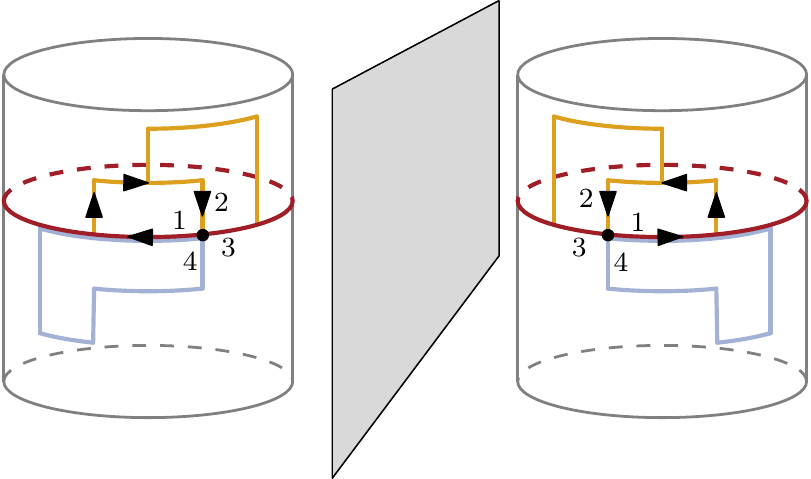}
  \caption{Mirroring the cylinder.}
  \label{fig:mirroring-cylinder}
 \end{figure}

\begin{lemma}[Flipping]
  \label{lem:flip_label}
  Let $\Gamma$ be an ortho-radial
  representation with outer face $f_o$ and central face~$f_c$. If the cycle 
  bounding the central face is not strictly monotone, there exists an ortho-radial
  representation $\flip{\Gamma}$ such that
  \begin{compactenum}
  \item $\reverse{f}_c$ is the outer face of $\flip{\Gamma}$ and 
  $\reverse{f}_o$ is the
    central face of $\flip{\Gamma}$,
  \item $\reverse{\ell}_{\reverse{C}}(\reverse{e})=\ell_{C}(e)$ for all
  essential cycles $C$ and edges $e$ on $C$, where $\reverse{\ell}$
  is the labeling in $\reverse{\Gamma}$.
  \end{compactenum}
  In particular, increasing and decreasing cycles of $\Gamma$ correspond to
  increasing and decreasing cycles of $\flip{\Gamma}$, respectively. 
\end{lemma}

\newcommand{\doublestar}{\star\star}

\begin{proof}
  We define $\flip{\Gamma}$ as follows.  The central face of $\Gamma$
  becomes the outer face of $\flip{\Gamma}$ and the outer face of
  $\Gamma$ becomes the central face of $\flip{\Gamma}$.  Further, we
  choose an arbitrary edge $e^{\doublestar}$ on the central face~$f_c$ of
  $\Gamma$ with $\ell_{f_c}(e^{\doublestar}) = 0$ (such an edge exists
  since the cycle bounding the central face is not strictly monotone), and
  choose $\reverse{e^{\doublestar}}$ as the reference edge of
  $\flip{\Gamma}$.  All other information of $\Gamma$ is transferred
  to $\flip{\Gamma}$ without modification. As the local structure is unchanged, 
  $\flip{\Gamma}$ is an ortho-radial representation.
  
  The essential cycles in $\Gamma$ bijectively correspond to the essential 
  cycles in $\flip{\Gamma}$ by reversing the direction of the cycles. That 
  is, any essential cycle $C$ in $\Gamma$ corresponds to the cycle 
  $\reverse{C}$ 
  in $\flip{\Gamma}$. Note that the reversal is necessary since we always 
  consider essential cycles to be directed such that the center lies in its 
  interior, which is defined as the area locally to the right of the cycle.
  
  Consider any essential cycle $C$ in $\Gamma$. We denote the labeling
  of $C$ in $\Gamma$ by $\ell_C$ and the labeling of $\reverse{C}$ in
  $\flip{\Gamma}$ by $\reverse{\ell}_{\reverse{C}}$. We show that for
  any edge $e$ on $C$,
  $\ell_C(e) = \reverse{\ell}_{\reverse{C}}(\reverse{e})$, which in
  particular implies that any monotone cycle in $\Gamma$ corresponds
  to a monotone cycle in $\reverse{\Gamma}$ and vice versa.  By
  Lemma~\ref{lem:label-simple-path} there is a reference path $P$ from
  the target of $e^\star$ to the source of the edge~$e$
  respecting~$C$.  Similarly, there is a path $Q$ from the target of
  $e$ to the source of $e^{\doublestar}$ that lies in the interior of
  $C$. The path $\reverse{Q}$ is a reference path from
  $e^{\doublestar}$ to $\reverse{e}$ respecting $\reverse{C}$ in
  $\flip{\Gamma}$.
    
  Assume for now that $P+e+Q$ is simple. We shall see at the end how the proof 
  can be extended if this is not the case. By the choice of $e^{\doublestar}$, 
  we have
  \begin{align}
  0 &= \ell_{f_c}(e^{\doublestar}) = \rot(e^\star + P + e + Q + e^{\doublestar}) 
  = \rot(e^\star + P + e) + \rot(e + Q + e^{\doublestar})
  \end{align}
  Hence, 
  $ \rot(e^\star + P + e) = -\rot(e+Q+e^{\doublestar}) =\rot(\reverse{e^{\doublestar}} +
  \reverse{Q} + \reverse{e})$ and in total
   \begin{align}
  \reverse{\ell}_{\reverse{C}}(\reverse{e}) 
  = \rot(\reverse{e^{\doublestar}} + \reverse{Q} + 
  \reverse{e})
  = \rot(e^\star + P + e) 
  = \ell_C(e).
  \end{align}
  Thus, any monotone cycle in $\Gamma$ corresponds to a monotone cycle in 
  $\flip{\Gamma}$ and vice versa.
  
  \begin{figure}[t]
    \centering
    \begin{subfigure}[c]{0.48\textwidth}
      \centering
      \includegraphics[page=1]{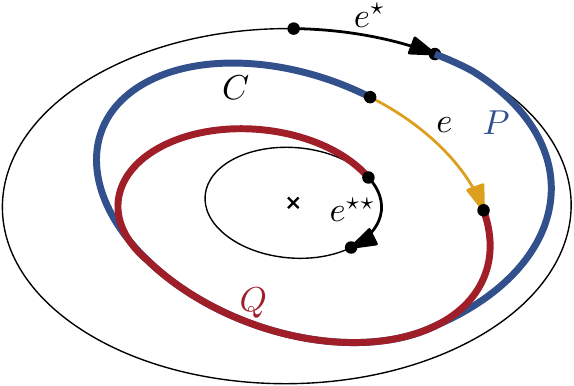}
      \subcaption{}
      \label{fig:lemma:flipping:original}
    \end{subfigure}
    \hspace{2ex}
    \begin{subfigure}[c]{0.48\textwidth}
      \centering
      \includegraphics[page=2]{fig/lemma_flipping.pdf}
      \subcaption{}
       \label{fig:lemma:flipping:copy}
    \end{subfigure}
    \caption{Illustration of proof for
      Lemma~\ref{lem:flip_label}.~(\subref{fig:lemma:flipping:original})~The
      path $P$ is a reference path from $e^{\star}$ to $e$ in $\Gamma$
      respecting $C$ and the path $Q$ is a reference path from
      $e^{\doublestar}$ to $\reverse{e}$ in $\flip{\Gamma}$ respecting
      $\reverse{C}$.~(\subref{fig:lemma:flipping:copy}) The edge $e$
      is subdivided by a vertex $x$. Afterwards, $G$ is cut at $C$
      such that the interior and the exterior get their own copies of
      $C$. The copies are connected by an edge between $x$ and $x'$.}
    \label{fig:cutting}
  \end{figure}
  
  If $P+e+Q$ is not simple, we make it simple by cutting $G$ at $C$
  such that the interior and the exterior of $C$ get their own copies
  of $C$; see Figure~\ref{fig:cutting}. We connect the two parts by an
  edge between two new vertices $x$ and $x'$ on the two copies of $e$,
  which we denote by $vw$ in the exterior part and $v'w'$ in the
  interior part.  The new edge is placed perpendicular to these
  copies.  The path $P+vxx'w'+Q$ is simple and its rotation is
  $0$. Hence, the argument above implies
  $\reverse{\ell}_{\reverse{C}}(\reverse{e}) = \ell_C(e)$.
\end{proof}

\begin{lemma}[Mirroring]
  \label{lem:mirroring_label}
  Let $\Gamma$ be an ortho-radial
  representation with outer face $f_o$ and central face $f_c$. There
  exists an ortho-radial representation $\mirror{\Gamma}$ such that
  \begin{compactenum}
  \item $\reverse{f}_o$ is the outer face of $\mirror{\Gamma}$ and 
  $\reverse{f}_c$ is the
    central face of $\mirror{\Gamma}$,
  \item $\mirror{\ell}_{\reverse{C}}(\reverse{e})=-\ell_{C}(e)$ for all
  essential cycles $C$ and edges $e$ on $C$, where $\mirror{\ell}$
  is the labeling in $\mirror{\Gamma}$.
  \end{compactenum}
  
  In particular, increasing and decreasing cycles of $\Gamma$ correspond to
  decreasing and increasing cycles of $\mirror{\Gamma}$, respectively. 
\end{lemma}
\begin{proof}
  We define $\mirror{\Gamma}$ as follows.  We reverse the direction of
  all faces and reverse the order of the edges around each vertex.
  The outer and central face are equal to those in $\Gamma$ (except
  for the directions) and the reference edge is $\reverse{e^\star}$.
  By this definition $(e_1, e_2)$ is a combinatorial angle in $\Gamma$
  if and only if $(\reverse{e}_2, \reverse{e}_1)$ is a combinatorial
  angle in $\mirror{\Gamma}$. We define
  $\rot_{\mirror{\Gamma}}(\reverse{e}_2,\reverse{e}_1) =
  \rot_\Gamma(e_1, e_2)$, where the subscript indicates the
  ortho-radial representation that defines the rotation. Thus,
  edges that point left
  in $\Gamma$ point right in $\mirror{\Gamma}$ and vice versa, but the
  edges that point up (down) in $\Gamma$ also point up (down) in
  $\mirror{\Gamma}$. Note that this construction satisfies the
  conditions for ortho-radial representations.
 
  Let $e=tu$ be an edge on an essential cycle $C$ and let $P$ be a
  reference path from $e^\star=rs$ to $e$ that respects $C$; by
  Lemma~\ref{lem:label-simple-path} we assume that $P$ starts at~$s$
  and ends at $t$.  After mirroring, $P$ still is a reference path
  from $e^\star$ to $e$, but its rotation in~$\mirror{\Gamma}$ may be
  different from its rotation in $\Gamma$.  As above, to distinguish the
  directions and rotations of paths in $\mirror{\Gamma}$ from the ones
  in $\Gamma$, we include $\mirror{\Gamma}$ and $\Gamma$ as subscripts
  to $\rot$ and $\dir$. 

  As $P$ starts at $s$ and ends at $t$ we have by definition   
  \begin{align*}
  \dir_{\Gamma}(e^\star, P, e)&=\rot_{\Gamma}(e^\star + P + e),\\
  \dir_{\mirror{\Gamma}}(\reverse{e^\star}, P, \reverse{e})&=\rot_{\mirror{\Gamma}}(e^\star + P + e).
  \end{align*}
  As for any path $Q$ we have
  $\rot_{\Gamma}(Q) = -\rot_{\mirror{\Gamma}}(Q)$, we obtain
  $\dir_{\mirror{\Gamma}}(\reverse{e^\star}, P, \reverse{e}) =
  -\dir_{\Gamma}(e^\star, P, e)$.
  By the definition of labels as directions of reference paths we
  directly obtain that
  $\mirror{\ell}_{\reverse{C}}(\reverse{e}) = -\ell_{C}(e)$. In
  particular, if $C$ is increasing (decreasing) in $\Gamma$, then
  $\reverse{C}$ is decreasing (increasing) in $\mirror{\Gamma}$.
\end{proof}

\section{Properties of Labelings}\label{sec:labeling-properties}
In this section we study the properties of labelings in more detail to
derive useful tools for proving Theorem~\ref{thm:main-result:drawable}.
Throughout this section, we are given an instance
$(G,\mathcal E,f_c,f_o)$ with an ortho-radial representation~$\Gamma$
and a reference edge $e^\star$. The following observation follows
immediately from the definition of labels and the fact that the
rotation of any essential cycle is $0$.
\begin{observation}\label{obs:repr:label_difference}  
  Let $C$ be an essential cycle.  Then, for any two edges $e$ and $e'$ on $C$, it
  holds that $\rot(\subpath{C}{e,e'})=\ell_C(e')-\ell_C(e)$.
\end{observation}

Note that, if an edge $e$ is contained in two essential cycles $C_1$
and~$C_2$, then their labelings may generally differ, i.e.,
$\ell_{C_1}(e) \ne \ell_{C_2}(e)$.  In fact,
Figure~\ref{fig:repr:intersection_cycles} shows that two cycles
$C_1,C_2$ may share two edges $e,e'$ such that
$\ell_{C_1}(e) \ne \ell_{C_2}(e)$ and
$\ell_{C_1}(e') = \ell_{C_2}(e')$.
\begin{figure}[bt]
  \centering
  \begin{subfigure}{.45\textwidth}
    \centering
    \includegraphics{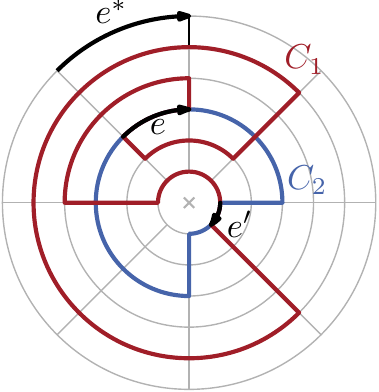}
    \caption{}
    \label{fig:repr:intersection_cycles-complete}
  \end{subfigure}
  \hfil
  \begin{subfigure}{.45\textwidth}
    \centering
    \includegraphics{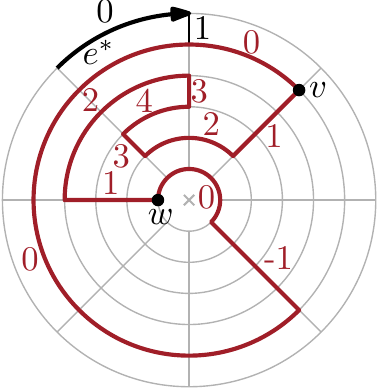}
    \caption{}
    \label{fig:repr:intersection_cycles-C1}
  \end{subfigure}
  \caption{(\protect\subref{fig:repr:intersection_cycles-complete}) Two cycles 
  $C_1$ and $C_2$ may have both common edges with different labels 
  ($\ell_{C_1}(e)=4\neq 0 = \ell_{C_2}(e)$) and ones with equal labels 
  ($\ell_{C_1}(e')=\ell_{C_2}(e')=0$).
  (\protect\subref{fig:repr:intersection_cycles-C1}) All labels of 
  $\subpath{C_1}{v,w}$ are positive, implying that $C_1$ goes down. Note that 
  not all edges of $\subpath{C_1}{v,w}$ point downwards.}
  \label{fig:repr:intersection_cycles}
\end{figure}
The rest of this section is devoted to understanding the relationship
between labelings of essential cycles that share vertices or edges.
The following technical lemma is a key tool in this respect; see also Figure~\ref{lem:two-intersecting-cycles}.

\begin{figure}[t]
  \centering
  \includegraphics{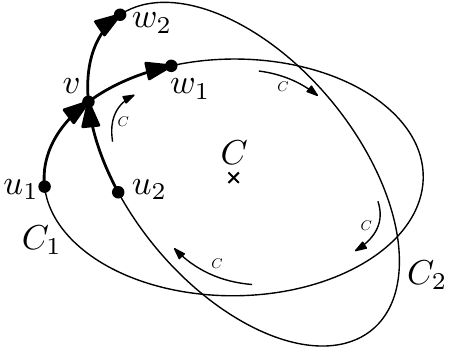}
  \caption{Illustration of proof for
    Lemma~\ref{lem:repr:equal_labels_at_intersection}. The essential
    cycles $C_1$ and $C_2$ intersect each other having a common vertex
    $v$. The essential cycle $C$ bounds the central face of $H$.}
   \label{lem:two-intersecting-cycles}
\end{figure}

\begin{lemma}\label{lem:repr:equal_labels_at_intersection}
  Let $C_1$ and $C_2$ be two essential cycles and let $H=C_1+C_2$ be
  the subgraph of $G$ formed by $C_1$ and~$C_2$. Let $v$ be a common
  vertex of $C_1$ and $C_2$ that is incident to the central face of
  $H$.  For $i=1,2$, let further~$u_i$ and~$w_i$ be the vertices
  preceding and succeeding $v$ on~$C_i$, respectively.  Then
  $\ell_{C_1}(vw_1) = \ell_{C_1}(u_1v) + \rot(u_1vw_1) =
  \ell_{C_2}(u_2v) + \rot(u_2vw_1)$.  Moreover, if $w_1=w_2$, then
  $\ell_{C_1}(vw_1)=\ell_{C_2}(vw_2)$.
\end{lemma}

\begin{proof}
  Let $C$ be the cycle that bounds the central face of $H$.  First
  assume that the edge $vw_1$ is incident to the central face of $H$.

  Let $P$ be a reference path from the reference edge to $vw_1$
  respecting $C_1$. Similarly, let $Q$ be a reference path from the
  reference edge to $vw_2$ respecting $C_2$. By
  Lemma~\ref{lem:label-simple-path} we assume that $P$ and $Q$ start
  at the target of the reference edge and end at $v$.  We observe
  that both $P$ and $Q$ respect the essential cycle $C$.

  Then, Corollary~\ref{cor:measure-direction} applied to $P$ and~$Q$ yields
  \begin{align*}
    \ell_C(vw_1) = \rot(e^\star\join P\join vw_1) = \rot(e^\star\join P) + \rot(u_1vw_1) \\
    \ell_C(vw_1) = \rot(e^\star\join Q\join vw_1) = \rot(e^\star\join Q)+\rot(u_2vw_1)
\end{align*}
By the definition of labelings it is
$\ell_{C_1}(u_1v)=\rot(e^\star\join P)$,
$\ell_{C_2}(u_2v)=\rot(e^\star\join Q)$,
and~$\ell_{C_1}(vw_1) = \ell_{C_1}(u_1v) + \rot(u_1vw_1)$.  Combining
this with the previous equation give the desired result:
\[
 \ell_{C_1}(vw_1) = \ell_{C_1}(u_1v) + \rot(u_1vw_1)= \ell_{C}(vw_1) = \ell_{C_2}(u_2v) + \rot(u_2vw_1).
\]
If $vw_1$ does not lie on $C$, then the edge $vw_2$ does.  By swapping
the roles of~$C_1$ and~$C_2$ and using the same argument as above, we
obtain
\[
  \ell_{C_1}(u_1v) + \rot(u_1vw_2) = \ell_{C_2}(u_2v) + \rot(u_2vw_2).
\]
Since $vw_1$ lies locally to the left of both $u_1vw_2$ and $u_2vw_2$,
it is $\rot(u_ivw_1)=\rot(u_ivw_2)-\alpha$ for $i=1,2$ and the same
constant $\alpha$, which is either $1$ or $2$. Hence, we get
\begin{align*}
  \ell_{C_1}(vw_1) = \ell_{C_1}(u_1v) + \rot(u_1vw_1)
                 &= \ell_{C_1}(u_1v)+ \rot(u_1vw_2)-\alpha\\ &= \ell_{C_2}(u_2v) + \rot(u_2vw_2) -\alpha = \ell_{C_2}(u_2v) + \rot(u_2vw_1).
\end{align*}

Finally, if $w_1=w_2$, i.e., $vw_1$ lies on both $C_1$ and~$C_2$, then
$\ell_{C_1}(vw_1) = \ell_{C_2}(u_2v) + \rot(u_2vw_1) =
\ell_{C_2}(u_2v) + \rot(u_2vw_2) = \ell_{C_2}(vw_2)$.
\end{proof}

\begin{corollary}
  \label{cor:central-face-same-label}
  Let~$C_1$ and~$C_2$ be two essential cycles, let~$H=C_1+C_2$ be the
  subgraph of $G$ formed by $C_1$ and $C_2$, and let~$e$ be an edge
  that lies on both $C_1$ and~$C_2$ and that is incident to the
  central face of $H$.  Then~$\ell_{C_1}(e) = \ell_{C_2}(e)$.
\end{corollary}

This allows us to prove an important criterion to exclude strictly monotone
cycles. We call an essential cycle $C$ \emph{horizontal} if
$\ell_C(e) = 0$ for all edge $e$ of $C$. We show that a strictly monotone cycle
and a horizontal cycle cannot share vertices.

\begin{proposition}
\label{prop:horizontal_cycle}
Let $C_2$ be a horizontal cycle and let $C_1$ be an essential cycle
that shares at least one vertex with $C_2$.  Then $C_1$ is not strictly
monotone.
\end{proposition}

\begin{figure}[t]
  \centering
  \includegraphics[page=2]{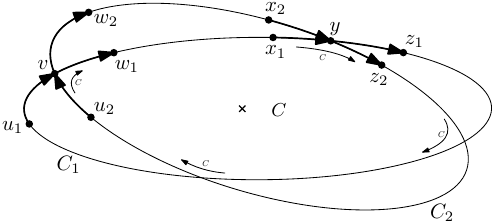}
  \caption{Illustration of proof for
    Proposition~\ref{prop:horizontal_cycle}.  All edges of $C_2$ are
    labeled with 0. In this situation there are edges on $C_1$ with
    labels $-1$ and $1$. Hence, $C_2$ is neither increasing nor
    decreasing.}
   \label{fig:intersecting-horizontal-cycle}
\end{figure}

\begin{proof}
  The situation is illustrated in
  Figure~\ref{fig:intersecting-horizontal-cycle}.

  If the two cycles are equal, the claim clearly holds.  Otherwise, we
  show that one can find two edges on $C_1$ whose labels have opposite
  signs.

  Let $v$ be a shared vertex of $C_1$ and~$C_2$ that is incident to
  the central face $f$ of $H=C_1 + C_2$ and such that the edge $vw_1$
  on $C$ is not incident to $f$.  For $i=1,2$ denote by $u_i$
  and~$w_i$ the vertex preceding and succeeding $v$ on~$C_i$,
  respectively.  By Lemma~\ref{lem:repr:equal_labels_at_intersection}
  it is
\begin{equation}\label{eqn:rect:two_cycles_horizontal:vw}
\ell_{C_1}(vw_1)=\ell_{C_2}(u_2v)+\rot(u_2vw_1) = \rot(u_2vw_1)\,,
\end{equation}
where the second equality follows from the assumption that
$\ell_{C_2}(u_2v)=0$.

Let $y$ be the first common vertex of $C_1$ and $C_2$ on the central
face~$f$ after $v$.  That is, $\subpath{f}{v,y}$ is a part of one of
the cycles $C_1$ and $C_2$, and it intersects the other cycle only at
$v$ and $y$.  For $i=1,2$, we denote by~$x_i$ and~$z_i$ the vertices
preceding and succeeding $y$ on $C_i$.  Again by
Lemma~\ref{lem:repr:equal_labels_at_intersection} (this time swapping
the roles of $C_1$ and~$C_2$), we have
\begin{equation}\label{eqn:rect:two_cycles_horizontal:xy}
  0 = \ell_{C_2}(yz_2) = \ell_{C_1}(x_1y)+\rot(x_1yz_2).
\end{equation}
Overall, we have $\ell_{C_1}(vw_1) = \rot(u_2vw_1)$
and~$\ell_{C_1}(x_1y) = -\rot(x_1yz_2)$.

By construction $vw_1$ and $x_1y$ lie on the same side of
$C_2$. Hence, $u_2vw_1$ and $x_1yz_2$ both make a right turn if $vw_1$
and $x_1y$ lie in the interior of $C_2$ and a left turn otherwise.
Thus, it is $\rot(u_2vw_1)=\rot(x_1yz_2)\neq 0$, and therefore
$\ell_{C_1}(vw_1)$ and $\ell_{C_1}(x_1y)$ have opposite signs.  Hence,
$C_1$ is not strictly monotone.
\end{proof}

In many cases we cannot assume that one of two essential cycles
sharing a vertex is horizontal.  However, we can still draw
conclusions about their intersection behavior from their labelings and
find conditions under which shared edges have the same label on both
cycles.

Intuitively, positive labels can often be interpreted as going
downwards and negative labels as going upwards.  In
Figure~\ref{fig:repr:intersection_cycles-C1} all edges of
$\subpath{C_1}{v,w}$ have positive labels and in total the distance from the center decreases along this path, i.e., the distance of
$v$ from the center is greater than the distance of $w$ from the center.  Yet, the edges on
$\subpath{C_1}{v,w}$ point in all possible directions---even upwards.
One can still interpret a maximal path with positive labels as leading
downwards with the caveat that this is a property of the whole path
and does not impose any restriction on the directions of the
individual edges.

Using this intuition, we expect that a path~$P$ going down cannot
intersect a path~$Q$ going up if $P$ starts below $Q$.  In
Lemma~\ref{lem:repr:illegal_intersection}, we show that this
assumption is correct if we restrict ourselves to intersections on the
central face.

\begin{lemma}\label{lem:repr:illegal_intersection}
  Let $C_1$ and $C_2$ be two simple, essential cycles in $G$ sharing
  at least one vertex.  Let $H = C_1 + C_2$, and denote the central
  face of $H$ by $f$.  Let $v$ be a vertex that is shared by~$C_1$
  and~$C_2$ that is incident to $f$ and, for $i=1,2$, let~$u_i$
  and~$w_i$ be the vertices preceding and succeeding~$v$ on $C_i$,
  respectively.  
  \begin{compactenum}
  \item\label{lem:repr:illegal_intersection-in} If $\ell_{C_1}(u_1v)\geq 0$
    and $\ell_{C_2}(u_2v)\leq 0$, then $u_2v$ lies in the interior of $C_1$.
  \item\label{lem:repr:illegal_intersection-out} If $\ell_{C_1}(vw_1)\geq 0$
    and $\ell_{C_2}(vw_2)\leq 0$, then $vw_2$ lies in the exterior of $C_1$.
  \end{compactenum}
\end{lemma}

\begin{figure}[bt]
  \centering
  \begin{subfigure}{.45\textwidth}
    \centering
     \includegraphics[page=3]{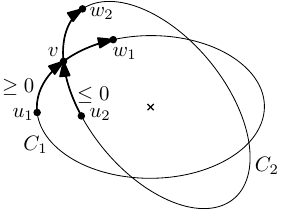}
     \caption{}
     
    \label{fig:repr:illegal_intersection-in}
  \end{subfigure}
  \hfil
  \begin{subfigure}{.45\textwidth}
    \centering
 \includegraphics[page=4]{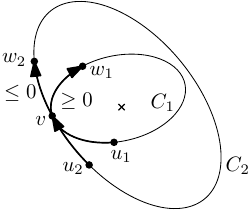}
    \caption{}
    \label{fig:repr:illegal_intersection-out}
  \end{subfigure}
  \caption{Possible intersection of two cycles~$C$ and $C'$ at
    $v$. (\subref{fig:repr:illegal_intersection-in})~The labels of the incoming edges satisfy
    $\ell_{C_1}(u_1v)\geq 0$ and $\ell_{C_2}(u_2v)\leq 0$. The edges
    $vw_1$ and $vw_2$ could be exchanged. (\subref{fig:repr:illegal_intersection-out})~The labels of the outgoing edges satisfy
      $\ell_{C_1}(vw_1)\geq 0$ and $\ell_{C_2}(vw_2)\leq 0$. The edges $u_1v$
      and $u_2v$ could be exchanged. }
  \label{fig:repr:illegal_intersection}
\end{figure}

\begin{proof}
  The second case follows from the first by taking the mirror
  representation; this reverses the order on the cycles and changes the
  sign of each label by Lemma~\ref{lem:mirroring_label}, but does not
  change the notion of interior and exterior.  It therefore suffices to 
  consider the first case.

  Since the central face~$f$ lies in the interior of both $C_1$ and $C_2$
  and $v$ lies on the boundary of $f$, one of the edges $vw_1$ and
  $vw_2$ is incident to $f$.  We denote this edge by $vx$ and it is
  either $x=w_1$ (as in Figure~\ref{fig:repr:illegal_intersection-in})
  or $x=w_2$.  By Lemma~\ref{lem:repr:equal_labels_at_intersection},
  we
  have~$\ell_{C_1}(u_1v) + \rot(u_1vx) = \ell_{C_2}(u_2v) +
  \rot(u_2vx)$.  Applying $\ell_{C_1}(u_1v)\geq 0$ and
  $\ell_{C_2}(u_2v)\leq 0$, we obtain $\rot(u_1vx) \leq
  \rot(u_2vx)$. Therefore, $u_2v$ lies to the right of or on $u_1vx$
  and thus in the interior of $C_1$.
\end{proof}

The next lemma is a direct consequence of
Lemma~\ref{lem:repr:illegal_intersection} when applied to decreasing
and increasing cycles.

\begin{lemma}
  \label{lem:increasing_decreasing_disjoint}
  A decreasing and an increasing cycle do not have any common vertex.
\end{lemma}

\begin{figure}[tb]
  \centering
  \includegraphics[page=1]{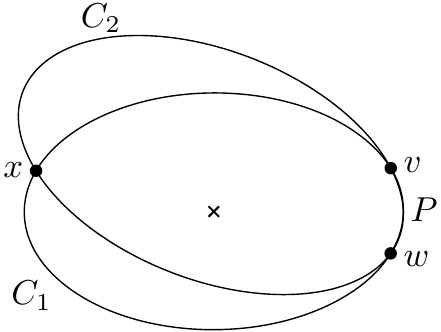}
  \caption{Illustration of proof for
    Lemma~\ref{lem:increasing_decreasing_disjoint}. The cycle $C_1$ is
    increasing and the cycle~$C_2$ is decreasing. The edge of $C_2$
    ending at $v$ strictly lies in the exterior of $C_1$ and the edge
    of $C_2$ starting at $w$ strictly lies in the interior of
    $C_1$. Hence, the edge of $C_2$ ending at $x$ strictly lies in the
    interior of $C_1$, which contradicts
    Lemma~\ref{lem:repr:illegal_intersection}.}
  \label{fig:lem:decreasing_increasing}
\end{figure}

\begin{proof}
Let $C_1$ be an increasing and $C_2$ a decreasing cycle. Assume that they have 
a common vertex. But then there also is a common vertex on the central face~$f$ 
of the subgraph $C_1+C_2$.
Consider any maximal common path~$P$ of $C_1$ and $C_2$ on $f$. We denote the 
start vertex of $P$ by $v$ and the end vertex by $w$. Note that $v$ may equal 
$w$. By Lemma~\ref{lem:repr:illegal_intersection} the edge to $v$ on the 
decreasing cycle $C_2$ lies strictly in the exterior of $C_1$, where the 
strictness follows from the maximality of~$P$. Similarly, the 
edge from $w$ on $C_2$ lies strictly in the interior of $C_1$. Hence, 
$\subpath{C_2}{w,v}$ crosses $C_1$. Let $x$ be the first intersection of $C_1$ 
and $C_2$ on $f$ after $w$. Then the 
edge to $x$ on $C_2$ lies strictly in the interior of $C_1$, contradicting 
Lemma~\ref{lem:repr:illegal_intersection}.
\end{proof}

For the correctness proof in Section~\ref{sec:rectangulation}, a
crucial insight is that for essential cycles using an edge that is
part of a regular face, we can find an alternative cycle without this
edge in a way that preserves labels on the common subpath.

\begin{lemma}\label{lem:reroute-face}
  If an edge $e$ belongs to both a simple essential cycle~$C$ and a
  regular face~$f$, then there exists a simple
  essential cycle~$C'$ that can be decomposed into two paths $P$
  and~$Q$ such that
  \begin{compactenum}[(i)]
  \item $P$ or $\reverse{P}$ lies on $f$,
  \item $Q = C \cap C'$, and
  \item $\ell_C(e) = \ell_{C'}(e)$ for all edges $e$ on~$Q$.
  \end{compactenum}
\end{lemma}

\begin{figure}[bt]
  \centering
  \begin{subfigure}{.45\textwidth}
    \centering
    \includegraphics[page=1]{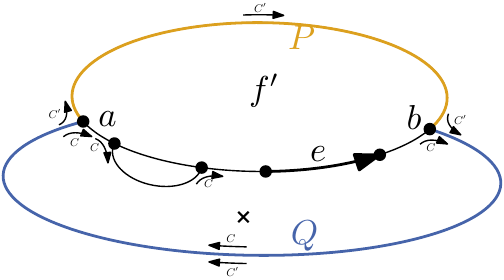}
    \caption{}
    \label{fig:rect:existence_C'-outer}
  \end{subfigure}
  \hfil
  \begin{subfigure}{.45\textwidth}
    \centering
    \includegraphics[page=2]{fig/lemma_reroute_face.pdf}
    \caption{}
    \label{fig:rect:existence_C'-center}    
  \end{subfigure}

  \caption{The edge $e$ cannot lie on both the outer and the central
    face. In both cases $C'$ can be subdivided in two paths $P$ and
    $Q$ on $C$ and $f$, respectively.  Here, these paths are separated
    by the vertices $a$ and
    $b$.~(\subref{fig:rect:existence_C'-outer}) The cycle bounding the
    outer face is $C'$. The edge $e$ does not lie on the outer face,
    and hence the cycle bounding this face is defined as
    $C'$. (\subref{fig:rect:existence_C'-center}) The cycle bounding
    the central face is $C'$. }
  \label{fig:rect:existence_C'}
\end{figure}

\begin{proof}
  Consider the graph $H=C+f$ composed of the essential cycle $C$ and
  the regular face~$f$. Since $e$ is incident to $f$, the edge $e$ cannot lie 
  on both the outer and the central face in $H$.  If $e$ does not lie on the 
  outer face,
  we define $C'$ as the cycle bounding the outer face but directed
  such that it contains the center in its interior; see
  Figure~\ref{fig:rect:existence_C'-outer}.  Otherwise, $C'$ denotes the
  cycle bounding the central face; see
  Figure~\ref{fig:rect:existence_C'-center}.

  Since $C$ lies in the exterior of $f$, the intersection of $C$ with
  $C'$ forms one contiguous path~$Q$.  Setting $P=C'-Q$ yields a path
  that lies completely on $f$ (it is possible though that $P$ and
  $f$ are directed differently).

  By the construction of $C'$ the edges of $Q$ are incident to the
  central face of $C+C'$.  Then
  Lemma~\ref{lem:repr:equal_labels_at_intersection} implies
  that~$\ell_{C}(e) = \ell_{C'}(e)$ for all edges $e$ of $Q$.
\end{proof}

The last lemma of this section shows that we can replace single edges
of an essential cycle~$C$ with complex paths without changing the
labels of the remaining edges on $C$.

\begin{figure}[tb]
  \centering
  \includegraphics[page=1]{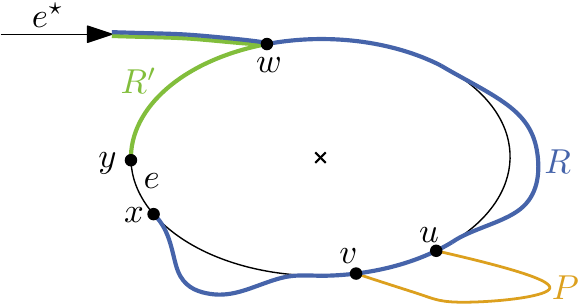}
  \caption{Illustration of the proof for
    Lemma~\ref{lem:same-labels-except-from-edge}.  The reference path
    $R'$ (green) for the edge $xy$ is constructed based on the
    reference path $R$. Replacing the edge $uv$ with a path $P$, does
    not impact $R'$, which implies that the label of $xy$ remains the
    same. }
  \label{fig:lem:same_labels}
\end{figure}

\begin{lemma}\label{lem:same-labels-except-from-edge}
  Let $C$ be an essential cycle in an ortho-radial representation
  $\Gamma$ and let $uv$ be an edge of $C$.  Consider the ortho-radial
  representation $\Gamma'$ that is created by replacing $uv$ with an
  arbitrary path $P$ such that the interior vertices of $P$ do not
  belong to $\Gamma$, i.e., they are newly inserted vertices in
  $\Gamma'$. If the cycle $C'=C[v,u]+P$ is essential, then the labels of
  $C$ and $C'$ coincide on $C[v,u](=C'[v,u])$.
\end{lemma}

\begin{proof} For an illustration of the proof see
  Figure~\ref{fig:lem:same_labels}.  Let $e=xy$ be an arbitrary edge
  on $C[v,u]$ and let $R$ be a reference path from the reference edge
  $e^\star$ to $e$ that respects $C$. We first construct a new
  reference path~$R'$ that does not contain $uv$ as follows. Let $w$
  be the first vertex of $R$ that lies on $C$. If $C[w,x]$ does not
  contain $e$, we define $R'=R-R[w,x]+C[w,x]$ and otherwise
  $R'=R-R[w,x]+\reverse{C}[w,y]$.  We observe that $R'$ is again a
  reference path of $e$ that respects $C$. Further, it can be
  partitioned into a prefix that only consists of edges that do not
  belong to $C$ and a suffix that only consists of edges that belong
  to $C[v,u]$.

  We now show that $R'$ is a reference path of $e$ in $\Gamma'$ that
  respects $C'$. As $R'$ does not use $uv$, it is still contained in
  $\Gamma'$ and hence it is a reference path of $e$. So assume that
  $R'$ does not respect~$C'$. Hence, $R'$ and $C'$ have a vertex $z$
  in common such that the outgoing edge~$e'$ of $R'$ at $z$ strictly
  lies in the interior of $C'$. As $R'$ respects $C$, this vertex lies
  on $P$. It cannot be an intermediate vertex of $P$, because these
  are newly inserted in $\Gamma'$. Hence, $z$ is either $u$ or $v$ and
  thus part of $C[v,u]$. In particular, it occurs on $R'$ after $w$, which
  implies that $e'$ belongs to $C$. This contradicts that $e'$
  strictly lies in the interior of $C'$.  Altogether, this shows that
  $R'$ is a reference path of $e$ both in $\Gamma$ and $\Gamma'$ such
  that $C$ and $C'$ are respected. Consequently, $\ell_C(e) = \ell_{C'}(e)$.
\end{proof}

\section{Characterization of Rectangular Ortho-Radial Representations}
\label{sec:characterization-rect}

Throughout this section, assume that $I=(G,\mathcal E,f_c,f_o)$ is an instance with an ortho-radial
representation~$\Gamma$ and a reference edge $e^\star$.  We prove
Theorem~\ref{thm:main-result:drawable} for the case that $\Gamma$ is
\emph{rectangular}. In a \emph{rectangular ortho-radial
  representation} the central face and the outer face are horizontal
cycles, and every regular face is a rectangle, i.e., it has exactly
four right turns, but no left turns.

We first observe that a bend-free ortho-radial drawing $\Delta$ can be
described by an angle assignment together with the lengths of its
vertical edges and the angles of the circular arcs representing the
horizontal edges; we call the angles of the circular arcs
\emph{central angles}.  We define two flow networks that assign
consistent lengths and central angles to the vertical edges and horizontal
edges, respectively.  These networks are straightforward adaptions of
the networks used for drawing rectangular graphs in the
plane~\cite{bett-gdavg-99}.  In the following, \emph{vertex} and
\emph{edge} refer to the vertices and edges of~$G$, whereas
\emph{node} and \emph{arc} are used for the flow networks.

\newcommand{\Nhor}{N_{\mathrm{hor}}}
\newcommand{\Ahor}{A_{\mathrm{hor}}}
\newcommand{\Fhor}{F_{\mathrm{hor}}}

\newcommand{\Nver}{N_{\mathrm{ver}}}
\newcommand{\Aver}{A_{\mathrm{ver}}}
\newcommand{\Fver}{F_{\mathrm{ver}}}

\begin{figure}[t]
  \begin{subfigure}[b]{.45\textwidth}
    \centering
    \includegraphics{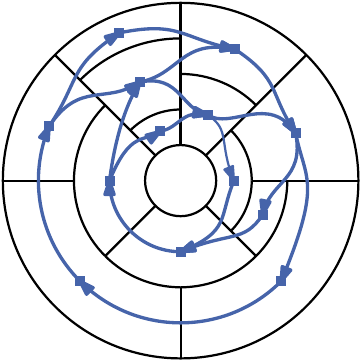}
    \caption{$\Nver$}
    \label{fig:draw:flows-ver}
  \end{subfigure}
  \hfil
  \begin{subfigure}[b]{.45\textwidth}
    \centering
    \includegraphics{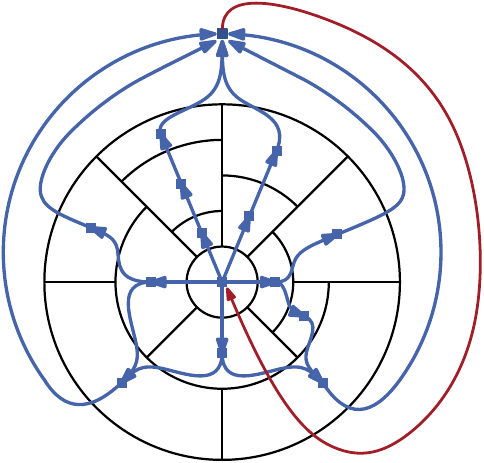}
    \caption{$\Nhor$}
    \label{fig:draw:flows-hor}
  \end{subfigure}
  \caption{Flow networks $\Nver$ and $\Nhor$ (blue arcs) for an example 
  graph~$G$ (black edges). }
  \label{fig:draw:flows}
\end{figure}

The network $\Nver=(\Fver,\Aver)$ with nodes~$\Fver$ and arcs~$\Aver$
for the vertical edges contains one node for each face of $G$ except
for the central and the outer face. All nodes have a demand of~0.  For
each vertical edge $e$ in $G$, which we assume to be directed upwards,
there is an arc $a_e=fg$ in $\Nver$, where $f$ is the face to the left
of $e$ and $g$ the one to its right. The flow on $fg$ has the lower
bound $l(fg)=1$ and upper bound $u(fg)=\infty$. An example of this
flow network is shown in Figure~\ref{fig:draw:flows-ver}.

Similarly, the network~$\Nhor=(\Fhor,\Ahor)$ assigns the central angles of the
horizontal edges. There is a node for each face of $G$, and an arc
$a_e=fg$ for every horizontal edge $e$ in $G$ such that $f$ lies
locally below $e$ and $g$ lies locally above $e$.  Additionally,
$\Nhor$ includes one arc from the outer to the central face.  Again,
all edges require a minimum flow of 1 and have infinite capacity. The
demand of all nodes is 0.  Figure~\ref{fig:draw:flows-hor} shows an
example of such a flow network.  Valid flows in these two flow
networks then yield an ortho-radial drawing of~$\Gamma$.

\begin{lemma}
  \label{thm:rect:flows-to-drawing}
  A pair of valid flows in $\Nhor$ and $\Nver$ corresponds to a
  bend-free ortho-radial drawing of $\Gamma$ and vice versa.
\end{lemma}
\begin{proof}
  Given a feasible flow $\varphi_{\mathrm{ver}}$ in $\Nver$, we set
  the length of each vertical edge~$e$ of $G$ to the flow
  $\varphi_{\mathrm{ver}}(a_e)$ on the arc $a_e$ that crosses $e$. For
  each face $f$ of $G$, the total length of its left side is equal to
  the total amount of flow entering $f$.  Similarly, the length of the
  right side is equal to the amount of flow leaving $f$.  As the flow
  is preserved at all nodes of $\Nver$, the left and right sides of
  $f$ have the same length.  The central angles of the horizontal
  edges are obtained from a flow~$\varphi_{\mathrm{hor}}$ in $\Nhor$.
  Let $\Phi$ be the total amount of flow that leaves the central
  face. Then for each horizontal edge~$e$ we set its central angle to
  $2\pi\varphi_{\mathrm{hor}}(a_e)/\Phi $, where
  $\varphi_{\mathrm{hor}}(a_e)$ is the amount of flow of the arc
  $a_e$ that connects the two adjacent faces of $e$. As the flow is
  preserved at all nodes of $\Nhor$, the top and bottom sides of each
  face have the same central angle.
  
  Conversely, given a bend-free ortho-radial drawing~$\Delta$ of
  $\Gamma$, we can extract flows in the two networks.  For each
  vertical edge $e$ we set the flow~$\varphi_{\mathrm{ver}}(a_e)$ of the
  corresponding arc $a_e$ to $l_e/l_{\min}$, where $l_e$ is the
  length of $e$ in $\Delta$ and $l_{\min}$ is the length of the
  shortest edge in $\Delta$.  With the scaling, we ensure that the
  flow of each arc is at least~$1$.  Similarly, for the horizontal
  edges we assign to each arc $a_e$ of each horizontal edge $e$ the
  flow $\varphi_{\mathrm{hor}}(a_e) = \alpha_e/\alpha_{\min}$,
  where $\alpha_e$ is the central angle of $e$ in $\Delta$ and
  $\alpha_{\min}$ is the smallest central angle of any horizontal
  edge in $\Delta$.  Again, the scaling ensures that each arc has flow
  at least~$1$.  Since the opposing sides of the regular faces have
  the same lengths and central angles, the flow is preserved at all
  nodes.
\end{proof}

Using this correspondence of drawings and feasible flows, we show the
characterization of rectangular graphs.

\begin{thm}\label{thm:draw:rectangle_drawing}
  Let $\Gamma$ be a rectangular ortho-radial representation and
  let $\Nhor=(\Fhor, \Ahor)$ and $N_\text{ver}=(\Fver, \Aver)$ be the flow networks as
  defined above.  The following statements are equivalent:
  \begin{compactenum}[(i)]
  \item\label{item:draw:rectangle_drawing:drawing} $\Gamma$ is drawable.
  \item\label{item:draw:rectangle_drawing:valid} $\Gamma$ is valid.
  \item\label{item:draw:rectangle_drawing:set} For every subset $S\subseteq \Fver$ 
    such that there is an arc from $\Fver \setminus S$ to $S$ in 
  $\Nver$, there is also an arc from $S$ to $\Fver\setminus S$.

  \end{compactenum}  
\end{thm}

\begin{proof}
  ``\eqref{item:draw:rectangle_drawing:drawing} $\Rightarrow$
  \eqref{item:draw:rectangle_drawing:valid}'': Let $\Delta$ be a
  bend-free ortho-radial drawing of $G$ preserving the embedding described by
  $\Gamma$ and let $C$ be an essential cycle.  Our goal is to show that
  $C$ is not strictly monotone.  To this end, we construct a path $P$
  from the reference edge of $\Gamma$ to a vertex on $C$ such that either the
  labeling of $C$ induced by $P$ attains both positive and negative
  values or it is 0 everywhere.

  In $\Delta$ either all vertices of $C$ lie on the same concentric
  circle, or there is a maximal subpath $Q$ of $C$ whose vertices all
  have maximum distance to the center of the ortho-radial grid among
  all vertices of $C$.  In the first case, we may choose the
  endpoint~$v$ of the path~$P$ arbitrarily, whereas in the second case
  we select the first vertex of $Q$ as $v$; for an example see
  Figure~\ref{fig:draw:drawing_to_representation}.

\begin{figure}[bt]
  \centering
  \includegraphics{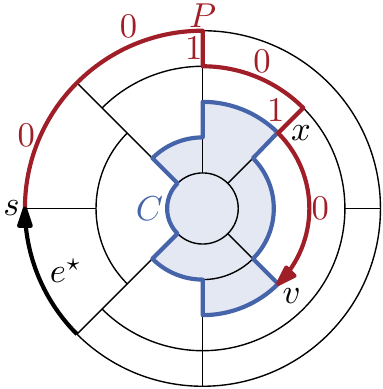}
  \caption{The path $P$ from $s$ to $v$---constructed backwards by going only 
  up or left---does not intersect the interior of $C$. The rotations of the 
  edges on $P$ relative to $e^\star$ are 0 or 1.}
  \label{fig:draw:drawing_to_representation}
\end{figure}

We construct the path $P$ backwards (i.e., the construction yields
$\reverse{P}$) as follows: Starting at $v$ we choose the edge going
upwards from $v$, if it exists, or the one leading left. Since all
faces of $\Gamma$ are rectangles, at least one of these always exists. This
procedure is repeated until the target~$s$ of the reference edge is
reached.  

To show that this algorithm terminates, we assume that this was not
the case.  As $G$ is finite, there must be a first time a vertex~$w$
is visited twice. Hence, there is a cycle~$C'$ in $\Delta$ containing
$w$ that contains only edges going left or up.  As all drawable
essential cycles with edges leading upwards must also have edges that
go down~\cite{hht-orthoradial-09}, all edges of $C'$ are
horizontal. By construction, there is no edge incident to a vertex of
$C'$ that leads upwards. The only cycle with this property, however,
is the one enclosing the outer face because $G$ is connected. But
this cycle contains the reference edge, and therefore the algorithm halts.

This not only shows that the construction of $P$ ends, but also that
$P$ is a path (i.e., the construction does not visit a vertex
twice). Thus, $P$ is a reference path from the reference edge
$e^\star$ to the edge $vv'$, where $v'$ is the vertex following $v$ on
$C$.  Further, $P$ respects $C$ as $\reverse{P}$ starts at the
outermost circle of the ortho-radial grid that is used by $C$ and as
by construction all edges of $\reverse{P}$ point left or upwards.

By the construction of $P$, the label of $vv'$ induced by $P$ is $0$.
If all edges of $C$ are horizontal, this implies $\ell_C(e)=0$ for all
edges $e$ of $C$, which shows that $C$ is not strictly monotone.
Otherwise, we claim that the edges $e_-=uv$ and $e_+=wx$ directly
before and after $Q$ on $C$ have labels $-1$ and $+1$, respectively.
Since all edges on $Q$ are horizontal and $e_-$ goes down, we have
$\rot(\subpath{C}{v,x})=1$ and therefore $\ell_C(e_+)=1$. Similarly,
$\rot(uvv')=1$ implies that $\ell_C(uv)=\ell_C(vv') - \rot(uvv')=-1$.

``\eqref{item:draw:rectangle_drawing:valid} $\Rightarrow$
\eqref{item:draw:rectangle_drawing:set}'': Instead of proving this
implication directly, we show the contrapositive. That is, we assume
that there is a set $S\subsetneq F_\text{ver}$ of nodes in $N_\text{ver}$ such 
that
$S$ has no outgoing but at least one incoming arc.  From this
assumption we derive that $\Gamma$ is not valid, as we find a strictly
monotone cycle.

Let $\Nver[S]$ denote the node-induced subgraph of
$\Nver$ induced by the set $S$. Without loss of generality, $S$
can be chosen such that $\Nver[S]$ is weakly connected, i.e.,
the underlying undirected graph is connected. If $\Nver$ is not weakly connected, at
least one weakly-connected component of $\Nver[S]$ possesses an
incoming arc but no outgoing arc, and we can work with this component
instead.

As each node of $S$ corresponds to a face of $G$, $S$ can also be
considered as a collection of faces of $G$. To distinguish the two
interpretations of $S$, we refer to this collection of faces by
$\mathcal{S}$.  Our goal is to show that the innermost or the
outermost boundary of $\mathcal{S}$ forms a strictly monotone cycle in $\Gamma$.
\begin{figure}[bt]
 \centering
 \includegraphics{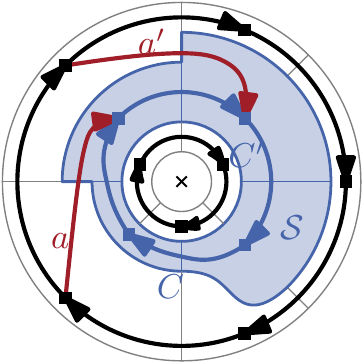}
 \caption{A set~$S$ of nodes in a graph~$G$ such that $N_\text{ver}[S]$ has no 
 outgoing but two incoming arcs~$a$ and $a'$.
 The set of faces~$\mathcal{S}$ corresponding to the nodes in $S$ are shaded 
 with blue.
 The outermost boundary of $\mathcal{S}$ forms an increasing cycle~$C$. The 
 edges on this cycle with label $-1$ are exactly those that are crossed by $a$ 
 or $a'$. All other edges on $C$ are labeled with 0.
 Note that the edge on $C$ at the bottom is curved because $G$ does not admit 
 an ortho-radial drawing.}
 \label{fig:draw:impossible_cycle}
\end{figure}
Figure~\ref{fig:draw:impossible_cycle} shows an example of such a set
$S$ of nodes. Here, the arcs $a$ and $a'$ lead from a node outside of
$S$ to one in $S$. These arcs cross edges on the outer boundary of
$\mathcal{S}$, which point upwards.

Let $\mathcal F$ be the set of faces of $\mathcal E$ (including the
central and outer face). Let $Z$ be a connected component of
$\mathcal F\setminus \mathcal S$ such that there exists an arc from
$Z$ to $S$ in $\Nver$ and let $C$ be the cycle in $G$
that separates $Z$ from $\mathcal S$. If $C$ were non-essential,
then $\rot(C)=4$ and $C$ would therefore contain an upward and a
downward edge. One of these edges would correspond to an incoming arc of $S$ 
and the other edge to an outgoing arc of $S$, contradicting the choice of $S$. 
Thus, $C$ is essential.

As usual we consider $C$ in clockwise direction. We may assume without loss of
generality that $C$ contains $\mathcal S$ in its interior; otherwise, we
consider $\flip{\Gamma}$ and $\reverse C$. Note that for each edge of $C$ 
the face locally to the right belongs to $\mathcal S$ whereas the face locally
to the left does not. Hence, upward edges of $C$ correspond to incoming arcs of
$S$ and downward edges to outgoing arcs.
Since there is an arc from $Z$ to $S$ but not vice versa, the cycle $C$
contains at least one upward but no downward edge.
Hence, there is some integer $k$ such that all labels of $C$ belong to 
$L_k=\{4k, 4k+1, 4k+2\}$. Since the numbers in $L_k$ are either all 
non-negative (if $k\geq 0$) or all negative (if $k<0$), the cycle $C$ is 
monotone. Moreover, $C$ is not horizontal because it has an upward edge.

``\eqref{item:draw:rectangle_drawing:set} $\Rightarrow$
\eqref{item:draw:rectangle_drawing:drawing}'':
By Lemma~\ref{thm:rect:flows-to-drawing} the existence of a drawing is 
equivalent to the existence of feasible flows in $\Nhor$ and 
$\Nver$.
If a flow network $N$ contains for each arc $a$ a cycle $C_a$ that contains 
$a$, then routing one unit flow along each of these cycles $C_a$ and adding all 
flows gives a circulation in $N$ where at least one unit flows along each arc.
Hence, it suffices to prove that in $\Nhor$ and in $\Nver$ each 
arc is contained in a cycle.

Note that $\Nhor$ without the arc from the outer face~$g$ to the central
face~$f$ is a directed acyclic graph with $f$ as its only source and
$g$ as its only sink.  For each arc $a\neq gf$ in $\Nhor$ there
is a directed path~$P_a$ from $f$ to $g$ via $a$.  Adding the
arc~$gf$, we obtain the cycle~$C_a=P_a\join gf$.

For $\Nver$ we consider an arc~$a=fg$, and we
define the set $S_g$ of all nodes $h$ for which there exists a
directed path from $g$ to $h$ in $\Nver$. By definition, there
is no arc from a vertex in $S_g$ to a vertex not in $S_g$.  As
$\Nver$ satisfies~\ref{item:draw:rectangle_drawing:set}, $S_g$
does not have any incoming arcs either. Hence, $f\in S_g$ and there is
a directed path~$P_a$ from $g$ to $f$. Then $C_a=P_a\join fg$ is the desired 
cycle.
\end{proof}

By \cite{hht-orthoradial-09} an ortho-radial drawing of a graph is locally 
consistent. Therefore,
Theorem~\ref{thm:draw:rectangle_drawing} implies the characterization of 
ortho-radial drawings for rectangular graphs.

\begin{corollary}[Theorem~\ref{thm:main-result:drawable} for Rectangular 
Ortho-Radial Representations]\label{cor:draw:characterization}
  A rectangular ortho-radial representation is
  drawable if and only if it is valid.
\end{corollary}

We note that we can construct the flows in $\Nhor$ and $\Nver$ using
standard techniques based on flows in planar graphs with multiple
sinks and sources~\cite{millerN95}. With this a drawing can be computed in $O(n^{\frac{3}{2}})$ 
time.

\section{Drawable Representations of Planar
  4-Graphs}\label{sec:rectangulation}

In the previous section we proved that a rectangular ortho-radial
representation is drawable if and only if it is valid. We extend this
result to general ortho-radial representations by reduction to the
rectangular case.  In Section~\ref{sec:rect:algorithm} we present a procedure that augments a given instance such that all faces become
rectangles. For readability we defer some of the proofs to
Section~\ref{sec:rectangulation-correctness}. In
Section~\ref{sec:rect:main-theorem} we use the rectangulation
procedure and Corollary~\ref{cor:draw:characterization} to show
Theorem~\ref{thm:main-result:drawable}.  We remark that all our proofs are
constructive, but make use of tests whether certain modified
ortho-radial representations are valid.  We develop an efficient
testing algorithm for this in
Section~\ref{sec:finding_monotone_cycles}.

\subsection{Rectangulation Procedure}
\label{sec:rect:algorithm}

Throughout this section, we are given an instance
$I=(G,\mathcal E,f_c,f_o)$ with a valid
ortho-radial
representation~$\Gamma$ and a reference edge $e^\star$.
The core of the argument is a \emph{rectangulation procedure} that
successively augments $G$ with new vertices and edges to a graph
$G'$ along with a valid rectangular ortho-radial representation $\Gamma'$.
Then,~$\Gamma'$ is
drawable by Corollary~\ref{cor:draw:characterization}, and removing
the augmented parts yields a drawing of~$\Gamma$.

\begin{figure}[t]
  \centering
     \begin{minipage}[b]{0.48\textwidth}
    \begin{subfigure}[b]{\textwidth}
      \centering
      \includegraphics[]{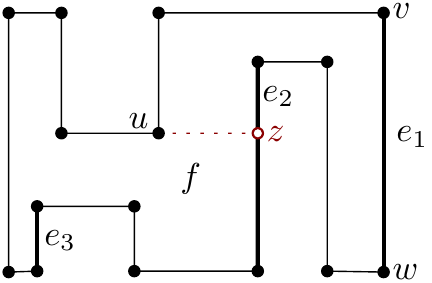}
      \subcaption{Multiple edge candidates.}
      \label{fig:hshape}
    \end{subfigure}    
  \end{minipage}
    \begin{minipage}[b]{0.25\textwidth}
    \begin{subfigure}[b]{\textwidth}
      \centering
      \includegraphics[page=2,scale=1]{fig/hshape.pdf}
      \subcaption{Vertical. }
      \label{fig:aug-vert}
    \end{subfigure}
    \begin{subfigure}[b]{\textwidth}
      \centering
      \includegraphics[page=3]{fig/hshape.pdf}
      \subcaption{Horizontal. }
      \label{fig:aug-horz-decreasing}
    \end{subfigure}
    \end{minipage}
    \begin{minipage}[b]{0.25\textwidth}
    \begin{subfigure}[b]{\textwidth}
      \centering
      \includegraphics[page=4]{fig/hshape.pdf}
      \subcaption{Horizontal. }
      \label{fig:aug-horz-valid}
    \end{subfigure}
    \begin{subfigure}[b]{\textwidth}
      \centering
      \includegraphics[page=5]{fig/hshape.pdf}
      \subcaption{Horizontal. }
      \label{fig:aug-path}
    \end{subfigure}
   \end{minipage}
  
  \caption{Examples of augmentations.
    (\protect\subref{fig:hshape})~The candidate edges of $u$ are $e_1$, $e_2$ 
    and $e_3$.
    (\protect\subref{fig:aug-vert})~Insertion of vertical edge $uz$. 
    (\protect\subref{fig:aug-horz-decreasing})~$\Gamma^{u}_{vw}$ contains a 
    decreasing cycle.
    (\protect\subref{fig:aug-horz-valid})~$\Gamma^{u}_{vw}$ is
    valid.
    (\protect\subref{fig:aug-path})~Insertion of horizontal edge 
    $uw_{i}$ because there is a horizontal path from $w_i$ to $u$.}
    \label{fig:augmentation}
\end{figure}

The rectangulation procedure works by augmenting non-rectangular faces
one by one, thereby successively removing concave angles at the
vertices until all faces are rectangles.  Traversing the boundary of a
face in clockwise direction yields a sequence of left and right turns,
where a degree-1 vertex contributes two left turns.
Note that concave angles correspond exactly to left turns in this
sequence.  Consider a face $f$ with a left turn (i.e., a concave
angle) at $u$ such that the following two turns when walking along $f$
(in clockwise direction) are right turns; see
Figure~\ref{fig:augmentation}.  We call $u$ a \emph{port} of $f$.  We
define a set of \emph{candidate edges} that contains precisely those
edges $vw$ of $f$, for which $\rot(\subpath{f}{u, vw}) = 2$; see
Figure~\ref{fig:hshape}.  We treat this set as a sequence, where the
edges appear in the same order as in $f$, beginning with the first
candidate after $u$.  The \emph{augmentation} $\Gamma^u_{vw}$ with
respect to a candidate edge $vw$ is obtained by splitting the edge
$vw$ into the edges $vz$ and $zw$, where $z$ is a new vertex, and
adding the edge $uz$ in the interior of $f$ such that the angle formed
by $zu$ and the edge following $u$ on $f$ is $90\degree$.  The
direction of the new edge $uz$ in $\Gamma^{u}_{vw}$ is the same for
all candidate edges.  If this direction is vertical, we call $u$ a
\emph{vertical port} and otherwise a \emph{horizontal port}. We note
that any vertex with a concave angle in a face becomes a port during
the augmentation process.  For regular faces Tamassia~\cite{t-emn-87}
shows that they always contain a port.  Moreover, the following
observation can be proven analogously.

\begin{observation}\label{obs:rect:augmentation_conditions_1_to_4}
  If $u$ is a port of a face $f$ and~$vw$ is a candidate edge for $u$,
  then $\Gamma^u_{vw}$ is an ortho-radial representation.
\end{observation}

However, an augmentation~$\Gamma^u_{vw}$ is not necessarily valid.  We
prove that we can always find an augmentation that is valid.  The
crucial ingredient is the following proposition.

\begin{proposition}
  \label{prop:augmentation-properties}
  Let $G$ be a planar 4-graph with valid ortho-radial representation
  $\Gamma$, let $f$ be a regular face of $G$ and let~$u$ be a port of
  $f$ with candidate edges $e_1=v_1w_1,\dots,e_k=v_kw_k$.  Then the following
  facts hold:
  \begin{compactenum}
  \item If $u$ is a vertical port, then~$\Gamma^u_{e_1}$ is a valid
    ortho-radial representation; see Figure~\ref{fig:aug-vert}\label{prop:augmentation-properties:fact:vertical-port}
  \item If $u$ is a horizontal port, then $\Gamma^u_{e_1}$ does not
    contain an increasing cycle and~$\Gamma^u_{e_k}$ does not
    contain a decreasing cycle; see Figure~\ref{fig:aug-horz-decreasing}--\subref{fig:aug-horz-valid}. \label{prop:augmentation-properties:fact:monotone-cycle}
    
  \item   Let~$P_i$ be the maximal path that contains the vertex $w_i$ of the
  candidate edge $e_i=v_iw_i$ and that consists of only horizontal edges.
  If $u$ is a horizontal port and~$\Gamma^u_{e_i}$ contains a
  decreasing cycle and~$\Gamma^u_{e_{i+1}}$ contains an increasing
  cycle, then $u$ is an endpoint of $P_i$ and adding the horizontal
  edge $uz$ to the other endpoint $z$ of $P_i$ yields a horizontal
  cycle; see Figure~\ref{fig:aug-path}.  In particular, $\Gamma + uz$ is valid.
    \label{prop:augmentation-properties:fact:horizontal-cycle}
  \end{compactenum}
\end{proposition}
To increase the readability we split the proof of
Proposition~\ref{prop:augmentation-properties} into the separate
Lemmas~\ref{lem:vertical-edge}--~\ref{lem:horizontal-path}, which we defer to Section~\ref{sec:rectangulation-correctness}.

We are now ready describe the rectangulation procedure. Let~$G$ be a
planar 4-graph with valid ortho-radial representation~$\Gamma$.
Without loss of generality, we can assume that~$G$ is connected,
otherwise we can treat the connected components separately.  We
further insert triangles in both the central and outer face and
suitably connect these to the original graph; see
Figure~\ref{fig:rect:outer_central_face}. Namely, for the central
face~$g$ we identify an edge $e$ on the simple cycle~$C$ bounding $g$
such that $\ell_C(e)=0$. Since $\Gamma$ is valid and $C$ is an
essential cycle, such an edge exists.  We then insert a new cycle~$C'$
of length~3 inside $g$ and connect one of its vertices to a new vertex
on $e$. The new cycle~$C'$ now forms the boundary of the central
face. Analogously, we insert into the outer face a cycle~$C_o$ of
length $3$ which contains $G$ and is connected to the reference
edge~$e^\star$. We choose an arbitrary edge $e^{\star\star}$ on $C_o$
as new reference edge. We observe that there is a path~$P$ from
$e^{\star\star}$ to $e^{\star}$ with rotation $0$. Hence, each
reference path from $e^\star$ to an essential cycle in $\Gamma$ can be
extended by $P$ such that the new path is a reference path with
respect to $e^{\star\star}$ and has the same rotation.

\begin{figure}[tb]
  \centering
  \includegraphics{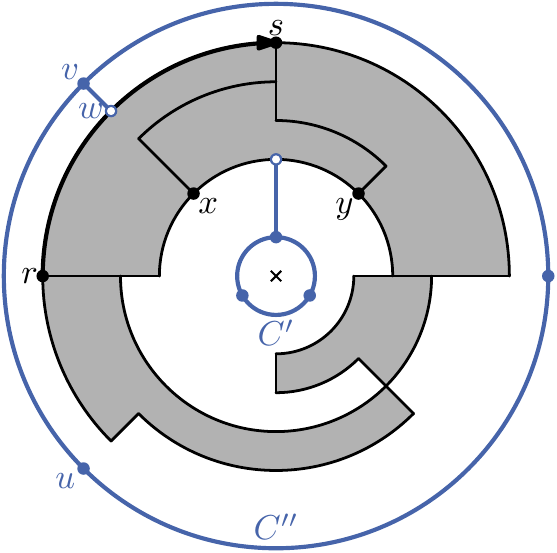}
  \caption{The outer and the central face are rectangulated by adding
    cycles of length~3.  The cycle~$C'$ is connected to an arbitrary
    edge~$xy$ that has label~0 and $C''$ is connected to a new vertex
    on the old reference edge~$rs$.  The edge~$uv$ is selected as the
    new reference edge.}
  \label{fig:rect:outer_central_face}
\end{figure}

After this preprocessing any face $f$ that is not a rectangle is
regular, and it therefore contains a port $u$.  If any of the
candidate augmentations~$\Gamma^u_{e_i}$ is valid,
then~$\Gamma^u_{e_i}$ has fewer concave corners than~$\Gamma$, and we
continue the augmentation procedure with~$\Gamma^u_{e_i}$.  On the
other hand, if none of these augmentations is valid, then
each~$\Gamma^u_{e_i}$ contains a strictly monotone cycle.
Let $i$ be the
smallest index such that $\Gamma^u_{e_i}$ contains an increasing cycle
and note that such an index $i$ exists and that $i>1$ by property 2 of
Proposition~\ref{prop:augmentation-properties}.  Then, by definition
of~$i$, $\Gamma^u_{e_{i-1}}$ contains a decreasing cycle
and~$\Gamma^u_{e_i}$ contains an increasing cycle.  But then property
3 of Proposition~\ref{prop:augmentation-properties} guarantees the
existence of a vertex $z$ such that~$\Gamma+uz$ is valid and has fewer
concave corners than~$\Gamma$.  Using this procedure we can iteratively
augment~$\Gamma$ to a rectangular ortho-radial
representation that contains a subdivision of the
ortho-radial representation~$\Gamma$.

\subsection{Proof of Proposition~\ref{prop:augmentation-properties}}
\label{sec:rectangulation-correctness}

Throughout this section, we assume that we are in the situation
described by Proposition~\ref{prop:augmentation-properties}.  That is,
$G$ is a planar 4-graph with valid ortho-radial
representation~$\Gamma$, and~$u$ is a port of a regular
face~$f$ with candidate edges~$e_1=v_1w_1,\dots,e_k=v_kw_k$.  After
possibly replacing $\Gamma$ with $\flip{\Gamma}$, we may
assume that the edge $uz$ resulting from an augmentation with a
candidate is directed to the right or up. By Lemma~\ref{lem:flip_label} there is a one-to-one correspondence between increasing (decreasing) cycles in $\Gamma$ and increasing (decreasing) cycles in $\flip{\Gamma}$. 

\begin{lemma}\label{lem:vertical-edge}
  If~$u$ is a vertical port, then~$\Gamma^u_{e_1}$ is a valid ortho-radial 
  representation.
\end{lemma}

\begin{figure}[bt]
    \centering
    \includegraphics{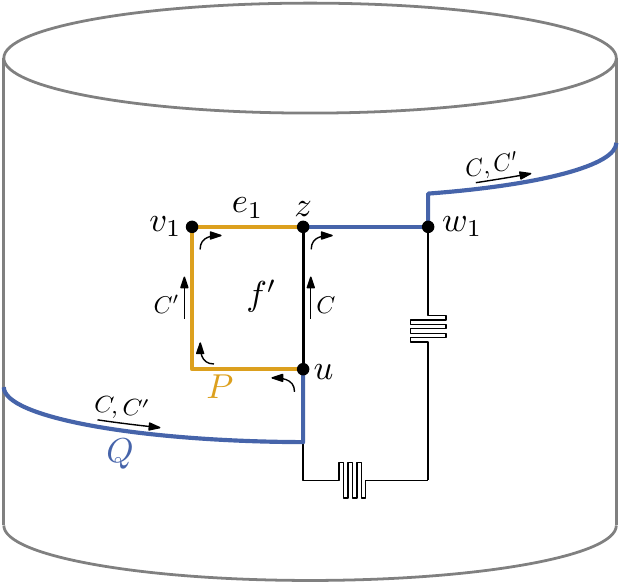}
    \caption{Illustration of proof for
      Lemma~\ref{lem:vertical-edge}. In this illustration it is
      assumed that inserting the edge $uz$ into a valid ortho-radial
      representation creates an increasing cycle $C$ that uses
      $uz$. However, then there is cycle $C'$ not using $uz$ that is
      also increasing. }
    \label{fig:rect:vertical-port}
\end{figure}

\begin{proof}
  Assume for the sake of contradiction that $\Gamma^u_{e_1}$ contains a 
  strictly monotone
  cycle $C$. As $\Gamma$ is valid, $C$
  must contain the new edge $uz$ in either direction (i.e., $uz$ or
  $zu$).  Let $f'$ be the new rectangular face of $G+uz$ containing
  $u$, $v_1$ and $z$, and consider the subgraph $H=C+f'$ of $G+uz$.
  According to Lemma~\ref{lem:reroute-face} there exists an 
  essential cycle $C'$ that does not contain $uz$. Moreover, $C'$ can
  be decomposed into paths $P$ and $Q$ such that $P$ lies on $f'$ and
  $Q$ is a part of~$C$; see Figure~\ref{fig:rect:vertical-port}.

  The goal is to show that $C'$ is increasing or decreasing. We
  present a proof only for the case that $C$ is an increasing
  cycle. The proof for decreasing cycles can be obtained by flipping
  all inequalities.

  For each edge $e$ on $Q$ the labels $\ell_C(e)$ and $\ell_{C'}(e)$
  are equal by Lemma~\ref{lem:reroute-face}, and
  hence $\ell_{C'}(e)\leq 0$.  For an edge $e\in P$, there are two
  possible cases: $e$ either lies on the side of $f'$ parallel to $uz$
  or on one of the two other sides.  In the first case, the label of
  $e$ is equal to the label $\ell_C(uz)$ ($\ell_C(zu)$ if $C$ contains
  $zu$ instead of $uz$). In particular the label is negative.

  In the second case, we first note that $\ell_{C'}(e)$ is even, since
  $e$ points left or right.  Assume that $\ell_{C'}(e)$ was positive
  and therefore at least $2$. Then, let $e'$ be the first edge on $C'$
  after $e$ that points to a different direction. Such an edge exists,
  since otherwise $C'$ would be an essential cycle whose edges all
  point to the right, but they are not labeled with 0.
  This edge $e'$ lies on $Q$ or is parallel to $uz$. Hence, the
  argument above implies that $\ell_{C'}(e')\leq 0$. However,
  $\ell_{C'}(e')$ differs from $\ell_{C'}(e)$ by at most $1$, which
  requires $\ell_{C'}(e')\geq 1$.  Therefore, $\ell_{C'}(e)$ cannot be
  positive.

  We conclude that all edges of $C'$ have a non-positive label. If all
  labels were $0$, $C$ would not be an increasing cycle by
  Proposition~\ref{prop:horizontal_cycle}.  Thus, there exists an edge
  on $C'$ with a negative label and $C'$ is an increasing cycle in
  $\Gamma$. But as $\Gamma$ is valid, such a cycle does not exist, and
  therefore $C$ does not exist either. Hence, $\Gamma^u_{e_1}$ is
  valid.
\end{proof}

  \begin{figure}[bt]
    \centering
    \includegraphics{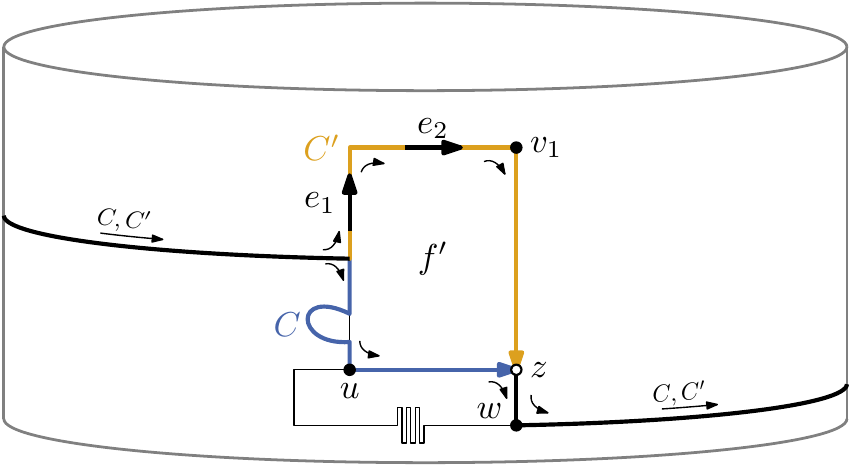}
    \caption{The increasing cycle $C$ contains $uz$. There are three
      possibilities for edges on $C'$ that lie not on $C$: They lie on
      the left side of $f'$ (like $e_1$), on the top (like $e_2$), or
      on the right side formed by only the edge $vz$.}
    \label{fig:rect:first_candidate}
  \end{figure}

\begin{lemma}\label{lem:horizontal-first-candidate}
  If $u$ is a horizontal port, then $\Gamma^u_{e_1}$ contains no
  increasing cycle.
\end{lemma}

\begin{proof}
  Let $f'$ be the new rectangular face of
  $\Gamma^u_{e_1}$ containing $u$, $v_1$ and $z$, and assume for the sake of
  contradiction that there
  is an increasing cycle $C$ in $\Gamma^u_{e_1}$.  This cycle must use
  either $uz$ or $zu$.  Similar to the proof of
  Lemma~\ref{lem:vertical-edge}, we find an increasing cycle~$C'$ in
  $\Gamma$, contradicting the validity of $\Gamma$.

  Applying Lemma~\ref{lem:reroute-face} to $C$ and $f'$ yields an
  essential cycle $C'$ without $uz$ and $zu$ that can be decomposed
  into a path $P$ on $f'$ and a path $Q\subseteq C\setminus f'$ such
  that all edges of $Q$ have non-positive labels.  We show in the
  following that the edges of $P$ also have non-positive labels.

  If $C$ contains $uz$, there are three possibilities for an edge $e$
  of $P$, which are illustrated in
  Figure~\ref{fig:rect:first_candidate}:
  The edge~$e$ lies on the left
  side of $f'$ and points up, $e$ is parallel to $uz$, or $e=v_1z$.  In
  the first case $\ell_{C'}(e)=\ell_{C}(uz)-1<0$ and in the second
  case $\ell_{C'}(e)=\ell_C(uz)\leq 0$. If $e=v_1z$, $C$ cannot contain
  $zv_1$ and therefore $zw_1\in C$.  Then, $\ell_{C'}(e)=\ell_C(zw_1)<
  0$. In all three cases the label of $e$ is at most 0.

  If $C$ contains $zu$, the label of $zu$ has to leave a remainder of
  2 when it is divided by 4 since $zu$ points to the left.  As the
  label is also at most $0$, we conclude $\ell_C(zu)\leq -2$. The
  edges of $P$ lie either on the left, top or right of
  $f'$. Therefore, the label of any edge $e$ on $P$ differs by at most
  1 from $\ell_{C}(zu)$, and thus we get $\ell_{C'}(e)\leq 0$.

  Summarizing the results above, we see that all edges on $C'$ are
  labeled with non-positive numbers. The case that all labels of $C'$
  are equal to 0 can be excluded, since $C$ would not be an increasing
  cycle by Proposition~\ref{prop:horizontal_cycle}.  Hence, $C'$
  is an increasing cycle, which was already present in $\Gamma$,
  contradicting the validity of $\Gamma$.
\end{proof}

\begin{lemma}\label{lem:horizontal-last-candidate}
  If $u$ is a horizontal port, then~$\Gamma^u_{e_k}$ contains no
  decreasing cycle.
\end{lemma}

  \begin{proof}
  Let $uz$ be the new edge inserted in~$\Gamma^u_{e_k}$.
  In~$\Gamma^u_{e_k}$, the face $f$ is split in two parts.  Let
  $f'$ be the face containing $v_k$ and $f''$ the one containing $w_k$.
  Assume for the sake of contradiction that there is a decreasing
  cycle~$C$ in $\Gamma^u_{e_k}$.  Then, either $uz$ or $zu$ lies on
  $C$.  By Lemma~\ref{lem:reroute-face} there exists an essential
  cycle~$C'$ that can be decomposed into a path $P$ on $f''$ and
  $Q=C \cap C'$; see Figure~\ref{fig:rect:last_candidate:a}.  For all
  edges $e\in E(Q)$, we have $\ell_C(e)=\ell_{C'}(e)\geq 0$ by
  Lemma~\ref{lem:reroute-face}.  Since~$C'$ is already present
  in~$\Gamma$ and~$\Gamma$ is valid, $C'$ cannot be decreasing.
  Moreover, since $C$ and~$C'$ intersect, $C'$ cannot be horizontal by
  Proposition~\ref{prop:horizontal_cycle}.  Therefore, $C'$ must
  contain an edge $xy$ with~$\ell_{C'}(xy) < 0$, which hence has to
  lie on $P$.

  \begin{figure}[bt]
    \centering
    \begin{subfigure}{0.45\textwidth}
      \centering
      \includegraphics{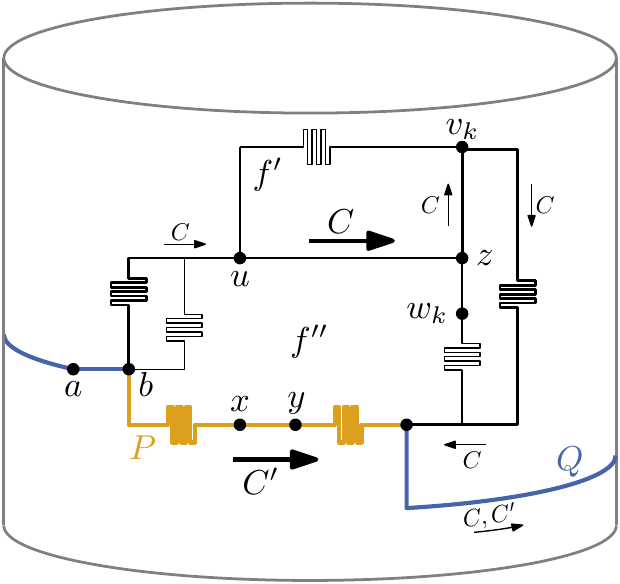}
      \caption{}
      \label{fig:rect:last_candidate:a}  
    \end{subfigure}
    \hfill
    \begin{subfigure}{0.45\textwidth}
      \centering
      \includegraphics[page=2]{fig/last_candidate.pdf}
      \caption{}
      \label{fig:rect:last_candidate:b}  
    \end{subfigure}
    \caption{The situation in the proof of
      Lemma~\ref{lem:horizontal-last-candidate}.  The cycle $C$ is
      decreasing and it is assumed that $\ell_{C'}(xy)<0$. (a)~The 
      decomposition of $C'$ into $P$ and $Q$. (b)~The reference paths $R$ and 
      $\subpath{C'}{b,x}$ from $ab$ to $xy$.}
    \label{fig:rect:last_candidate}
  \end{figure}
  
  Our goal is to show that there must be a candidate on $f$ after $y$
  and in particular after the last candidate $e_k$---a contradiction.
  The following claim gives a sufficient condition for the existence
  of such a candidate.

  \begin{Claim}
    \label{claim:candidate}
    If~$\rot(\subpath{f}{u,yx}) \le 2$, then there is a candidate on
    $\subpath{f}{yx, u}$.
  \end{Claim}
  To prove Claim~\ref{claim:candidate}, we determine for each edge $e$ on $f$
  the value $r(e):=\rot(\subpath{f}{u,e})$.  By assumption, it is $r(yx)\leq
  2$. For the last edge $e_{\text{last}}$ on $\subpath{f}{yx, u}$ it
  is $r(e_{\text{last}})=\rot(f)-\rot(tuv) = 5$, where $t$ and $v$
  are the preceding and succeeding vertices of $u$ on $f$,
  respectively.  Here, we use that $f$ is a regular face (i.e.,
  $\rot(f)=4$) and $\rot(tuv)=-1$ since $u$ is a port.

  Note that for two consecutive edges $e,e'$ on the boundary of $f$,
  it is~$r(e') \le r(e)+1$.  Therefore, there exists an edge~$e$ that
  lies between $yx$ and~$e_{\text{last}}$ on the boundary of $f$ that
  satisfies $r(e) = 2$.  Hence, $e$ is a candidate that lies after $yx$
  on the boundary of $f$.  \hfill ${\diamond}$ %

  To finish the proof of the lemma it hence suffices to show
  that~$\rot(\subpath{f}{u,yx}) \le 2$. 
  As $e_k$ is a candidate, we have $\rot(\subpath{f}{u, zw_k}) = 2$ and 
  therefore 
  \begin{equation*}
    \rot(\subpath{f}{u, yx})
    = \rot(\subpath{f}{u, zw_k}) + \rot(\subpath{f}{zw_k, yx})
    = \rot(\subpath{f}{zw_k, yx}) + 2.
  \end{equation*}
  Thus, it suffices to show $\rot(\subpath{f}{zw_k, yx}) \leq 0$. 
  
  We present a detailed argument for the case that $C$ uses $uz$ as
  illustrated in Figure~\ref{fig:rect:last_candidate}. At the end of
  the proof, we briefly outline how the argument can be adapted if $C$
  uses $zu$.
 
  If $C$ uses $uz$, then $P$ is directed such that $f''$
  lies to the left of $P$.
  Thus, $C'$ lies in the interior of~$C$.
  Let now $ab$ be the last edge of $Q$, and let $R$ be the path defined by 
  $\subpath{C}{b, uz} + \subpath{f}{zw_k,y}$; see 
  Figure~\ref{fig:rect:last_candidate:b}.
  Both $R$ and $\subpath{C'}{b,x}$ are reference paths from $ab$ to $xy$ that 
  lie in the interior of $C$ and in the exterior of $C'$. Applying the second 
  statement of Lemma~\ref{lem:measure-direction} hence gives
  \begin{equation}
    \dir(ab, R, xy) = \dir(ab, \subpath{C'}{b,x}, xy). 
    \label{eqn:rect:last_candidate:equal_dirs}
  \end{equation}
  The direction along~$R$ is defined as
  \begin{align*}
    \dir(ab, R, xy)
    &= \rot(ab + R + yx) - 2 \\
    &= \rot(ab + \subpath{C}{b, uz} + \subpath{f}{zw_k,y} + yx) - 2\\
    &= \rot(\subpath{C}{ab, uz} + \subpath{f}{zw_k,yx}) - 2\\
    &= \rot(\subpath{C}{ab, uz}) + \rot(uzw_k) + \rot(\subpath{f}{zw_k,yx}) - 2\\
    &= \ell_C(uz) - \ell_C(ab) +  \rot(\subpath{f}{zw_k,yx}) - 1.    \label{eqn:rect:last_candidate:dir_R}
  \end{align*}
  The last step uses
  $\rot(\subpath{C}{ab, uz}) = \ell_C(uz) - \ell_C(ab)$ and
  $\rot(uzw_k)=1$.

  The rotation along~$\subpath{C'}{b,x}$ is defined as
  \begin{align*}
    \dir(ab, \subpath{C'}{b,x}, xy)
    &= \rot(ab + \subpath{C'}{b, x} + xy)\\
    &= \rot(\subpath{C'}{ab, xy})\\
    &= \ell_{C'}(xy) - \ell_{C'}(ab)\\
     &= \ell_{C'}(xy) - \ell_{C}(ab).
  \end{align*}
  The last step uses that $\ell_C(ab) = 
  \ell_{C'}(ab)$ by Lemma~\ref{lem:reroute-face}.
  Altogether, we obtain  
   \begin{equation*}
    \ell_{C}(uz) - \ell_{C}(ab) + \rot(\subpath{f}{zw_k, yx}) - 1 = 
    \ell_{C'}(xy) - \ell_{C}(ab),
  \end{equation*}
  which can be rearranged to
  \begin{equation*}
    \rot(\subpath{f}{zw_k, yx}) = \ell_{C'}(xy) - \ell_{C}(uz) + 1.
  \end{equation*}
  With $\ell_{C'}(xy) \leq -1$ and $\ell_C(uz) \geq 0$ we obtain 
  $\rot(\subpath{f}{zw_k,yx}) \leq 0$. This completes the proof 
  for the case 
  that $uz$ lies on $C$.
  
  If $zu$ lies on 
  $C$, we consider the flipped representation~$\flip{\Gamma^u_{e_k}}$. In 
  $\flip{\Gamma^u_{e_k}}$ the cycle $\reverse{C}$ is decreasing and contains 
  the edge~$uz$. The cycle~$\reverse{C'}$ is not decreasing and contains the 
  edge~$yx$ with label $\reverse{\ell}_{\reverse{C'}}(yx) = \ell_{C'}(xy) < 0$.
  Moreover, the cycle $\reverse{C}$ contains $\reverse{C'}$ in its interior. 
  Thus, the argument above can be applied to $\reverse{C}$, $\reverse{C'}$, and 
  $yx$ instead of $C$, $C'$, and $xy$.
  \end{proof}

\begin{lemma}
  \label{lem:horizontal-path}
  Let~$P_i$ be the maximal path that contains the vertex $w_i$ of the
  candidate edge $e_i=v_iw_i$ and that consists of only horizontal edges.
  If $u$ is a horizontal port and~$\Gamma^u_{e_i}$ contains a
  decreasing cycle and~$\Gamma^u_{e_{i+1}}$ contains an increasing
  cycle, then $u$ is an endpoint of $P_i$ and adding the horizontal
  edge $uz$ to the other endpoint $z$ of $P_i$ yields a horizontal
  cycle.  In particular, $\Gamma + uz$ is valid.
\end{lemma}
\begin{figure}[bt]
  \centering
  \begin{subfigure}{.3\textwidth}
    \includegraphics{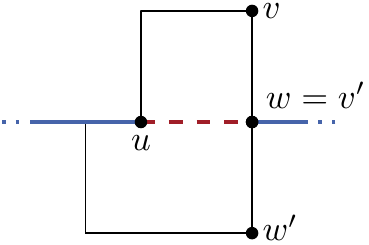}
    \caption{}
    \label{fig:rect:horizontal_between-neighbors}
  \end{subfigure}
  \hfil
  \begin{subfigure}{.3\textwidth}
    \includegraphics{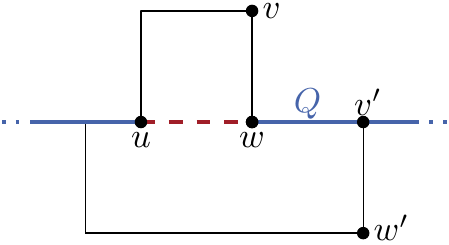}
    \caption{}
    \label{fig:rect:horizontal_between-w}
  \end{subfigure}
  \hfil
    \begin{subfigure}{.3\textwidth}
    \includegraphics{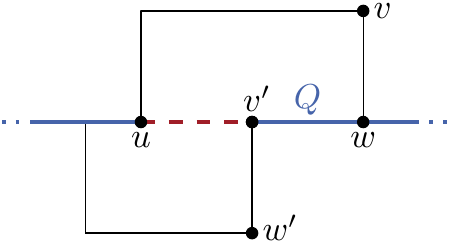}
    \caption{}
    \label{fig:rect:horizontal_between-v}
  \end{subfigure}
  \caption{Three possibilities how the path between $w$ and $v'$ can look 
  like: \protect\subref{fig:rect:horizontal_between-neighbors}~$w=v'$, 
  \protect\subref{fig:rect:horizontal_between-w}~all edges point right, and 
  \protect\subref{fig:rect:horizontal_between-v} all edges point left. In the 
  first two cases the edge $uw$ is inserted and in (c) $uv'$ is added.}
 \label{fig:rect:horizontal_between}
\end{figure}
\begin{proof}

\begin{figure}[tb]
  \centering
  \begin{subfigure}{.3\textwidth}
    \centering
    \includegraphics[page=1]{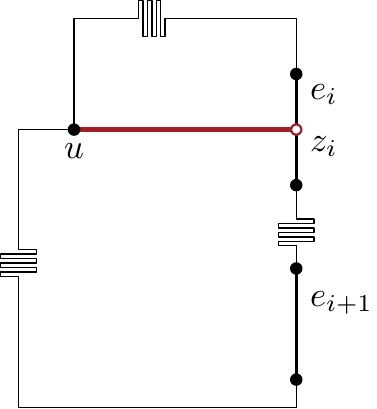}
    \caption{}
    \label{fig:rect:between_structure:a}
  \end{subfigure}
  \begin{subfigure}{.3\textwidth}
 \centering
    \includegraphics[page=2]{fig/between_structure.pdf}
    \caption{}
    \label{fig:rect:between_structure:b}
  \end{subfigure}
   \begin{subfigure}{.3\textwidth}
 \centering
     \includegraphics[page=3]{fig/between_structure.pdf}
     \caption{}
     \label{fig:rect:between_structure:c}
   \end{subfigure}
   \caption{Illustration of proof for Lemma~\ref{lem:horizontal-path}. (\subref{fig:rect:between_structure:a})~$\Gamma^u_{e_i}$ is obtained by inserting the edge $uz_{i}$ into $\Gamma$. (\subref{fig:rect:between_structure:b})~$\Gamma^u_{e_{i+1}}$ is obtained by inserting the edge $uz_{i+1}$ into $\Gamma$. (\subref{fig:rect:between_structure:c})~$\tilde{\Gamma}$ is obtained by combining $\Gamma^u_{e_{i}}$ and $\Gamma^u_{e_{i+1}}$.}
\label{fig:rect:between_structure}
\end{figure}

Let $z_{i}$ be the new vertex inserted in $\Gamma^u_{e_i}$ and
$z_{i+1}$ the one in $\Gamma^u_{e_{i+1}}$.  Since both $uz_{i}$ and
$uz_{i+1}$ point to the right, there is no augmentation of $\Gamma$
containing both edges.  We compare $\Gamma^u_{e_i}$ and
$\Gamma^u_{e_{i+1}}$ by the following construction (see also
Figure~\ref{fig:rect:between_structure}), which models all important
aspects of both representations: Starting from $\Gamma$ we insert new
vertices $z_{i}$ on $e_i$ and $z_{i+1}$ on $e_{i+1}$.  We connect $u$
and $z_i$ by a path of length 2 that points to the right and denote
its internal vertex by $x$.  Furthermore, a path of length~2 from $x$
via a new vertex $y$ to $z_{i+1}$ is added.  The edge $xy$ points down
and $yz_{i+1}$ to the right. In the resulting ortho-radial
representation $\tilde{\Gamma}$ the edge $uz_i$ in $\Gamma^u_{e_i}$ is
modeled by the path $uxz_i$. Similarly, the edge $uz_{i+1}$ in
$\Gamma^u_{e_{i+1}}$ is modeled by the path $uxyz_{i+1}$ in
$\tilde{\Gamma}$.  

Take any decreasing cycle in $\Gamma^u_{e_i}$. As $\Gamma$ is valid, 
this cycle must contain either $uz_i$ or $z_iu$.
We obtain a cycle~$C_1$ in $\tilde\Gamma$ by replacing $uz_i$ with $uxz_i$ (or 
$z_iu$ 
with $z_ixu$).
Note that $ux$ and $xz_i$ have the same label as $uz_i$, and the labels of all 
other edges on the cycles stay the same by 
Lemma~\ref{lem:same-labels-except-from-edge}. Therefore, $C_1$ is a decreasing 
cycle.

Similarly, there exists an increasing cycle~$C'_1$ in $\Gamma^u_{e_{i+1}}$,
which contains $uz_{i+1}$ or $z_{i+1}u$.  Replacing $uz_{i+1}$ with
$uxyz_{i+1}$ (or $z_{i+1}u$ with $z_{i+1}yxu$) we get a cycle $C_2$ in
$\tilde\Gamma$. By Lemma~\ref{lem:same-labels-except-from-edge} the labels of 
$C_2$ are non-positive outside of $uxyz_{i+1}$ (or $z_{i+1}yxu$) and there is 
an edge with negative label.
Consequently, the only edge of $C_2$ that may have a positive label is the 
edge~$e$ between $x$ and $y$.
Since $C_1$ and $C_2$ intersect, Lemma~\ref{lem:increasing_decreasing_disjoint}
implies that $C_2$ is not increasing.
Thus, the label of $e$ is positive. If $e=yx$, then
$\ell_{C_2}(e)\geq 3$ and its succeeding edge has a positive label as
well. Hence, the cycle $C_2$ contains the edge $xy$ and consequently
the edge $ux$.

Using this construction we show that one endpoint of $P_i$ is $u$ and
the other is another vertex of $f$. To that end, we prove the following claims. 

\begin{Claim}\label{claim:same-labels}
  The cycles $C_1$ and $C_2$ both contain the edge $ux$ and it is  
  $\ell_{C_1}(ux)=\ell_{C_2}(ux)=0$. 
\end{Claim}

\begin{Claim}\label{claim:vertices}
  The vertices $w_i$ and $v_{i+1}$ have a degree of at least 2.
  Further, $C_1$ contains $w_i$ and $C_2$ contains $v_{i+1}$.
\end{Claim}

\begin{Claim}\label{claim:part-of-path}
  The edges of $f[w_i,v_{i+1}]$ are part of $P_i$. In particular, $C_2$ contains $w_i$ or $C_1$ contains $v_{i+1}$.
\end{Claim}

\begin{Claim}\label{claim:endpoints-of-path}
  The right endpoint of $P_i$ is $u$ and the left endpoint~$z$ lies on $f$. 
\end{Claim}

Using these claims we prove the lemma as follows. 
Due to Claim~\ref{claim:endpoints-of-path} we can insert the
horizontal edge $uz$ into $\Gamma$ obtaining a horizontal cycle
$P_i+uz$. The resulting ortho-radial representation is valid, because
any strictly monotone cycle necessarily contains $uz$ and hence shares
a vertex with a horizontal cycle, contradicting Lemma~\ref{prop:horizontal_cycle}. This
finishes the proof. In the following we prove the claims.

\begin{figure}[t]
  \centering
  \includegraphics[page=2]{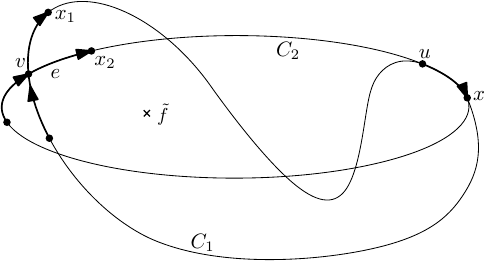}
  \caption{Illustration of proof for
    Claim~\ref{claim:same-labels} of Lemma~\ref{lem:horizontal-path}.}
   \label{fig:lemma-augmentation-horizontal-cycle}
\end{figure}

\textit{Claim~\ref{claim:same-labels}.}  Let $H=C_1+C_2$ be the graph
formed by the two cycles $C_1$ and $C_2$ and let $\tilde f$ be its
central face. Assume that the edge $ux$ is not incident to $\tilde
f$. As both $C_1$ and $C_2$ contain $ux$ (or $xu$),
the face $\tilde f$ consists of edges of both $C_1$ and $C_2$. In
particular, there is an edge $e$ of $C_2\setminus C_1$ on $\tilde f$
whose source $v$ lies on $C_1$; see Figure~\ref{fig:lemma-augmentation-horizontal-cycle}. Let $x_1$ and $x_2$
be the succeeding vertices of $v$ on $C_1$ and $C_2$, respectively.
As $C_1$ is a decreasing cycle, we have $\ell_{C_1}(vx_1)\geq 0$. If
$\ell_{C_2}(vx_2)\leq 0$, then by
Lemma~\ref{lem:repr:illegal_intersection} the edge $e=vx_2$ lies in
the exterior of $C_1$, which contradicts the choice of $e$. Hence, we
have $\ell_{C_2}(vx_2) > 0$, which implies that $e = xy$. Further, the
cyclic order of the vertex $x$ implies that $ux$ is the predecessor of
$xy$ on the central face, which contradicts the assumption that $ux$
is not incident to $\tilde f$.

Since $C_2$ contains $ux$, the central face $\tilde f$ lies to the
right of $ux$. Since $C_1$ is also directed such that $\tilde f$ lies
to the right of $C_1$, it cannot contain $xu$, but it contains $ux$.
By Lemma~\ref{lem:repr:equal_labels_at_intersection} we further obtain that $\ell_{C_1}(ux)=\ell_{C_2}(ux)=0$.

\textit{Claim~\ref{claim:vertices}.} We consider the setting in
$\tilde\Gamma$; see Figure~\ref{fig:rect:between_structure:c}. From
Claim~\ref{claim:same-labels} we obtain
$\ell_{C_1}(xz_i)=0$. Consequently, since $C_1$ is a decreasing cycle
and since $\rot(xz_iw_i)=1$ and $\rot(xz_iv_i)=-1$, the cycle contains
the edge $z_iw_i$ but not the edge $z_iv_i$. In particular, this
implies that $w_i$ has a degree of at least 2, as otherwise $C_1$
would not be simple. Similarly, from Claim~\ref{claim:same-labels} we
obtain $\ell_{C_1}(yz_{i+1})=0$. Consequently, since $C_2$ has only
non-positive labels except on $xy$ and since
$\rot(xz_{i+1}v_{i+1})=-1$ and $\rot(xz_{i+1}w_{i+1})=1$, the cycle
contains the edge $z_{i+1}v_{i+1}$ but not the edge
$z_{i+1}w_{i+1}$. In particular, this implies that $v_{i+1}$ has a
degree of at least 2, as otherwise $C_2$ would not be simple.

\textit{Claim~\ref{claim:part-of-path}.} Let $e_s$ be the direct
successor of $e_i$ and let $e_p$ be the direct predecessor of
$e_{i+1}$ on the boundary of $f$.  In order to show the claim, we do a
case distinction on the rotations $\rot(e_i,e_s)$ and
$\rot(e_p,e_{i+1})$. From Claim~\ref{claim:vertices} it follows that
$\rot(e_i,e_s)\in\{-1,0,1\}$ and $\rot(e_p,e_{i+1}) \in \{-1,0,1\}$.
Further, it cannot be both $\rot(e_i,e_f)= 1$ and
$\rot(e_p,e_{i+1})= 1$. Otherwise, since $\rot(f[u,e_i])=2$ and
$\rot(f[u,e_{i+1}])=2$, it would hold $\rot(f[u,e_s])=3$ and
$\rot(f[u,e_p])=1$. In that case, there would be a further candidate
edge $e$ with $\rot(f[u,e])=2$ between $e_i$ and $e_{i+1}$.  Thus, we
obtain $\rot(e_i,e_s)\in\{-1,0\}$ or $\rot(e_p,e_{i+1})\in\{-1,0\}$,

First, assume that $\rot(e_i,e_s)=0$. We obtain $e_s=e_{i+1}$ and hence
$e_p=e_i$. Otherwise, $e_s$ would be a candidate edge between $e_i$ and
$e_{i+1}$. Consequently, $f[w_i,v_{i+1}]$, which consist of the single
vertex $v=w_{i}=v_{i+1}$, trivially belongs to $P_i$. Moreover, by
Claim~\ref{claim:same-labels} the cycle $C_2$ then also contains $w_i$
and $C_1$ contains $v_{i+1}$. The very same argument can be applied
for $rot(e_{p},e_{i+1})=0$, as we also obtain $e_s=e_{i+1}$ and
$e_p=e_i$. In particular, we obtain $\rot(e_i,e_s)=0$ if and only if
$\rot(e_p,e_{i+1})=0$.

\begin{figure}[tb]
  \centering
  \begin{subfigure}{.49\textwidth}
    \centering
    \includegraphics[page=4]{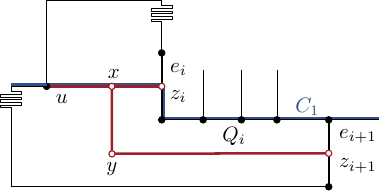}
    \caption{}
    \label{fig:rect:horz-path:a}
  \end{subfigure}
  \begin{subfigure}{.49\textwidth}
 \centering
    \includegraphics[page=5]{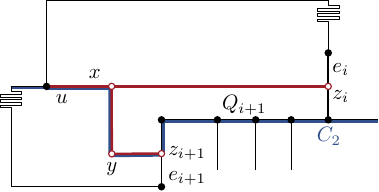}
    \caption{}
    \label{fig:rect:horz-path:b}
  \end{subfigure}
   \caption{Illustration of proof for
    Claim~\ref{claim:part-of-path} of Lemma~\ref{lem:horizontal-path}.}
\label{fig:rect:horz-path}
\end{figure}

Now, assume that $\rot(e_i,e_s)=-1$ or $\rot(e_p,e_{i+1})=-1$. We
first consider the case that $\rot(e_i,e_s)=-1$; see
Figure~\ref{fig:rect:horz-path:a}. We observe that $C_1$ contains
$e_s$ as it also contains the vertex $w_i$ by
Claim~\ref{claim:same-labels}. Let $Q_i$ be the maximal horizontal
path on $f$ that starts at $w_i$. We observe that $Q_i$ is a prefix of
$P_i$, and the cycle $C_1$ shares at least $e_s$ with $Q_i$. Since by
Claim~\ref{claim:same-labels} the cycle $C_1$ has label $0$ on $ux$,
it also has label $0$ on any of its edges that lie on $Q_i$.  In fact,
this implies that $C_1$ completely contains the path $Q_i$. Otherwise,
it would contain an edge that does not belong $Q_i$, but starts at an
intermediate vertex of $Q_i$. By the choice of $Q_i$ such an edge has
a negative label since its preceding edge has label $0$, which
contradicts that $C_1$ is decreasing. For the same reason, the
edge~$e$ that succeeds $Q_i$ on $f$ has a negative label. Hence, the edge $e$ is the candidate $e_{i+1}$. Consequently,
$C_1$ contains $v_{i+1}$ and $f[w_i,v_{i+1}]=Q_i$, which shows Claim~\ref{claim:part-of-path}. For the case that
$\rot(e_l,e_{i+1})=-1$ we use similar arguments; see
Figure~\ref{fig:rect:horz-path:b}. Let $Q_{i+1}$ be the maximal 
horizontal path on $f$ that starts at $v_{i+1}$ when going along $f$
counterclockwise. The first edge of $Q_{i+1}$ is $e_p$, which also
belongs to $C_2$. Further, $C_2$ has label $0$ on any of its edges
that belong to $Q_{i+1}$. Consequently, $C_2$ contains $Q_{i+1}$
completely, as otherwise it would contain an edge that points
downwards, contradicting that $C_{2}$ has only non-positive labels
except for the edge $xy$. Hence, the edge of $f$ that is incident with
the endpoint of $Q_{i+1}$ and does not belong to $Q_{i+1}$ is the
candidate edge $e_i$, which concludes the proof.

\textit{Claim~\ref{claim:endpoints-of-path}.} We first show that $u$
is the right endpoint of $P_i$. To that end, let $H=C_1+C_2$ be the
graph formed by the two cycles $C_1$ and $C_2$ and let $\tilde f$ be
its central face. From the proof of Claim~\ref{claim:same-labels} we
know that $ux$ is incident to the central face.  Further, from
Claim~\ref{claim:vertices} and Claim~\ref{claim:part-of-path} it
follows that $C_1$ and $C_2$ have the vertex $w_{i}$ or the vertex
$v_{i+1}$ in common; see Figure~\ref{fig:rect:horz-path:a} and
Figure~\ref{fig:rect:horz-path:b}, respectively. Let $v=w_i$ in the
former case and let $v=v_{i+1}$ in the latter case.

We first show that $C_1[v,x]=C_2[v,x]$. Assume that this is not
the case. Consider the maximal common prefix of $C_1[v,x]$ and
$C_2[v,x]$, and let $w$ be the endpoint of that prefix.  As $v$ is
incident to $\tilde f$, this also applies to $w$. Further, by
Lemma~\ref{lem:repr:illegal_intersection} the outgoing edge of $w$ on
$C_2$ lies in the exterior of $C_1$. Let $t$ be the first intersection
of $C_1$ and $C_2$ on $\tilde f$ after $w$. Then the edge to $t$ on $C_2$
lies strictly in the exterior of $C_1$ contradicting
Lemma~\ref{lem:repr:illegal_intersection}. 

Hence, we have $C_1[v,x]=C_2[v,x]$; we denote that path
by $P$. Since $ux$ is incident to $\tilde f$, all edges of $P$ are
incident to $\tilde f$ as well. Hence, by
Corollary~\ref{cor:central-face-same-label} all edges of $P$ have
label 0 on $C_1$ and $C_2$ and are therefore horizontal. As
$f[w_i,v_{i+1}]$ also belongs to $P_i$, we conclude that $u$ is the
right endpoint of $P_i$.

\begin{figure}[tb]
  \centering
  \begin{subfigure}{.49\textwidth}
    \centering
    \includegraphics[page=6]{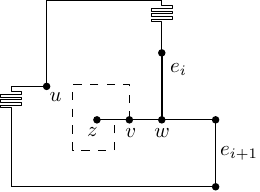}
    \caption{}
    \label{fig:rect:left-endpoint:a}
  \end{subfigure}
  \begin{subfigure}{.49\textwidth}
 \centering
    \includegraphics[page=7]{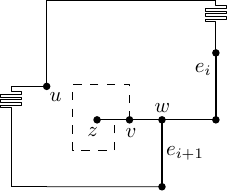}
    \caption{}
    \label{fig:rect:left-endpoint:b}
  \end{subfigure}
   \caption{Illustration of proof for
    Claim~\ref{claim:endpoints-of-path} of Lemma~\ref{lem:horizontal-path}.}
\label{fig:rect:left-endpoint}
\end{figure}

Finally, we argue that $z$ belongs to the face $f$. As by
Claim~\ref{claim:part-of-path} $f[w_i,v_{i+1}]$ belongs to $P_i$, the
vertex $z$ belongs to a prefix $P$ of $P_i$ that does not contain any
edge of $f[w_i,v_{i+1}]$ and ends either at $w_i$ or at $v_{i+1}$; see
Figure~\ref{fig:rect:left-endpoint:a} and
Figure~\ref{fig:rect:left-endpoint:b}, respectively. If $P$ is empty,
$w_i$ or $v_{i+1}$ is $z$, respectively, which shows the claim.  So
assume that $P$ is not empty. Let $vw$ be the last edge of $P$, i.e.,
depending on the considered case we either have $w=w_i$ or
$w=v_{i+1}$. By the local ordering of the incident edges of $w$, the
path $P$ lies in the interior of $f$. Further, $vw$ is incident to $f$
on both sides, as otherwise the two other incident edges of $w$ could
not be incident to $f$ as well. Hence, we obtain $\rot(f[u,wv])=3$ and
$\rot(f[u,vw])=0$. Thus, the path $f[wv,vw]$  of $f$ consists only of
horizontal edges, as otherwise there would be an edge $e$ on that path
with $\rot(f[u,e])=2$, which contradicts the choice of $e_i$ and
$e_{i+1}$. Altogether, the path $f[wv,vw]$ is $\reverse{P}+P$,
which proves that $z$ belongs to $f$.
\end{proof}

\subsection{Proof of the Main Theorem}\label{sec:rect:main-theorem}

We are now ready to prove the characterization of drawable ortho-radial representations.

\mrDrawable

\begin{proof}
\begin{figure}
  \centering
  \begin{subfigure}{.45\textwidth}
    \centering
    \includegraphics[page=1,width=\textwidth]{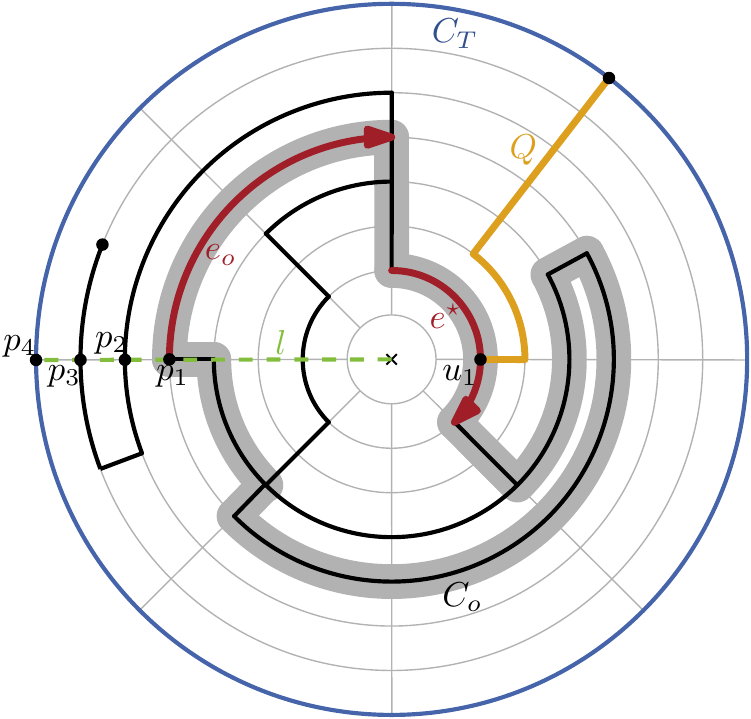}
    \caption{}
    \label{fig:rect:geometric-rectangulation:setting}
  \end{subfigure}\hfill
    \begin{subfigure}{.45\textwidth}
    \centering
    \includegraphics[page=2,width=\textwidth]{fig/geometric_rectangulation_complete.pdf}
    \caption{}
    \label{fig:rect:geometric-rectangulation:cycles}
  \end{subfigure}
      \begin{subfigure}{.45\textwidth}
    \centering
    \includegraphics[page=3,width=\textwidth]{fig/geometric_rectangulation_complete.pdf}
    \caption{}
    \label{fig:rect:geometric-rectangulation:rect}
  \end{subfigure}
 \caption{Illustration of proof for Theorem~\ref{thm:main-result:drawable}.}
  \end{figure}

If $\Gamma$ is a valid ortho-radial representation, then using the
rectangulation procedure of Section~\ref{sec:rect:algorithm} we obtain a
valid rectangular ortho-radial representation~$\Gamma'$ that contains a
subdivision of $\Gamma$.  Then~$\Gamma'$ has a
drawing~$\Delta'$ by Theorem~\ref{cor:draw:characterization}, and
a drawing~$\Delta$ of~$\Gamma$ is obtained from~$\Delta'$ by
removing edges that are in~$\Gamma'$ but not in~$\Gamma$, and by
undoing subdivisions. It remains to show that the reference edge $e^\star$ of $\Gamma$ is an outlying edge in $\Delta$, which then implies that $\Gamma$ is drawable.

Let $C_o$ be the essential cycle that is part of the outer face of
$\Gamma$ and let $e_o$ be an outermost horizontal edge of $C_o$ in
$\Delta$, i.e., an horizontal edge of $C_o$ with maximum distance to
the center; see Figure~\ref{fig:rect:geometric-rectangulation:setting}. If $e_o$ can be chosen as the reference edge $e^\star$,
then $e^\star$ is an outlying edge of $\Delta$. Hence, assume that $e_o$ cannot be chosen as $e^\star$.

Let $C_T$ be the essential cycle that forms the outer face of
$\Gamma'$. Further, let $l$ be the ray that emanates from the center
of $\Delta'$ and goes through the source vertex~$v$ of $e_o$. Let $p_1,\dots,p_k$ be the intersection points of $l$
with $\Delta'$ from $v$ on. We observe that $p_1=v$ and $p_k$ lies on
$C_T$. For each intersection point $p_i$ with $1\leq i \leq k$ we insert
a vertex $v_i$ at $p_i$ if there is no vertex so far. Further, we add
vertical edges to $\Delta'$ such that the new drawing contains the
path $P=v_1\dots v_k$. We denote the resulting drawing by $\Delta''$;
see Figure~\ref{fig:rect:geometric-rectangulation:cycles}.

By the construction of the rectangulation the drawing $\Delta''$ further contains
a path~$Q$ that connects a subdivision vertex $u_1$
on $e^\star$ with a vertex~$u_2$ on $C_T$. We observe that the cycle
$C=C_o[v_1,u_1]+Q+T[u_2,v_k]+\reverse{P}$ is essential. Moreover, by
construction the path $C[u_1,v_1]=Q+C_T[u_2,v_k]+\reverse{P}$ has
rotation 2. Since $C$ has a left bend at both $v_1$ and $u_1$ it
follows that $C_o[v_1,u_1]$ has rotation 0. Consequently, $e^\star$ is
an outlying edge in $\Delta$. Altogether, we obtain that $\Gamma$ is
drawable.

Conversely, assume that $\Gamma$ is drawable. Hence, $\Gamma$ has a drawing $\Delta$ in which the reference edge $e^\star$ of $\Gamma$ is an outlying edge.
By Lemma~\ref{lem:reference-edge:change} we can assume without loss of generality that $e^\star$ lies  on
the outermost circle that is used by an essential cycle of $\Delta$.
By~\cite{hht-orthoradial-09} the representation~$\Gamma$ satisfies
Conditions~\ref{cond:repr:sum_of_angles}
and~\ref{cond:repr:rotation_faces} of
Definition~\ref{def:local-conditions}.  To prove that
$\Gamma$ is valid, we show
how to reduce the general case to the more restricted one, where all
faces are rectangles.  By Corollary~\ref{cor:draw:characterization}
the existence of a drawing and the validity of the ortho-radial
representation are equivalent.

Given the drawing~$\Delta$, we augment it such that all faces are
rectangles. This rectangulation is similar to the one described in
Section~\ref{sec:rect:algorithm} but works with a drawing and not only
with a representation.  We first insert the missing parts of the
innermost and outermost circle that are used by $\Delta$ such that the
outer and the central face are already rectangles.  For each left turn
on a face~$f$ at a vertex $u$, we then cast a ray from $v$ in $f$ in
the direction in which the incoming edge of $u$ points.  This ray
intersects another edge in $\Delta$. Say the first intersection occurs
at the point~$p$. Either there already is a vertex~$z$ drawn at $p$ or
$p$ lies on an edge. In the latter case, we insert a new vertex, which
we call~$z$, at $p$.  We then insert the edge $uz$ in $G$ and update
$\Delta$ and $\Gamma$ accordingly.

Repeating this step for all left turns, we obtain a drawing~$\Delta'$
and an ortho-radial representation~$\Gamma'$ of the augmented
graph~$G'$; see Figure~\ref{fig:rect:geometric-rectangulation:rect}
for an example of $\Delta'$.  As the labelings of essential cycles
are unchanged by the addition of edges elsewhere in the graph, any
increasing or decreasing cycle in $\Gamma$ would also appear in
$\Gamma'$.  But by Corollary~\ref{cor:draw:characterization} $\Gamma'$
is valid, and hence neither $\Gamma$ nor $\Gamma'$ contain increasing
or decreasing cycles. Thus, $\Gamma$ is valid.
\end{proof}

We remark that the proof of Theorem~\ref{thm:main-result:drawable} is in
fact constructive and it can easily be implemented in polynomial time,
provided that we can check in polynomial time whether a given
ortho-radial representation is valid.  We develop such an algorithm
with running time $\O(n^2)$ in
Section~\ref{sec:finding_monotone_cycles}. In
Section~\ref{sec:efficient_rectangulation} we show how compute a rectangulation
of~$\Gamma$ in time $\O(n^2)$.

\section{Validity Testing}\label{sec:finding_monotone_cycles}
In this section, we show how to check whether a given ortho-radial
representation is valid in polynomial time, which yields the following statement. 

\mrValidity

The two conditions for ortho-radial representations are local and
checking them can easily be done in linear time. Throughout this
section, we therefore assume that we are given an instance
$I=(G,\mathcal E,f_c,f_o)$ with an ortho-radial
representation~$\Gamma$ and a reference edge $e^\star$.  The condition
for validity however references all essential cycles of which there
may be exponentially many. We present an algorithm that checks whether
$\Gamma$ contains a strictly monotone cycle and computes such a cycle
if one exists.
The main difficulty is that the labels on a decreasing cycle
$C$ depend on a reference path $P$ that runs from $e^\star$ to $C$ and
respects $C$.  However, we know neither the path $P$ nor the cycle $C$
in advance, and choosing a specific cycle $C$ may rule out certain
paths $P$ and vice versa.

We only describe how to search for decreasing cycles; increasing
cycles can be found by searching for decreasing cycles in the mirrored
representation by Lemma~\ref{lem:mirroring_label}.  A decreasing cycle
$C$ is \emph{outermost} if it is not contained in the interior of any
other decreasing cycle.  Clearly, if $\Gamma$ contains a decreasing
cycle, then it also has an outermost one.  We first show that in this
case this cycle is uniquely determined.

\begin{lemma}
  \label{lem:outermost_decreasing_cycle}
  If $\Gamma$ contains a decreasing cycle, there is a unique outermost
  decreasing cycle.
\end{lemma}

\begin{proof}
  Assume that $\Gamma$ has two outermost decreasing cycles $C_1$ and
  $C_2$, i.e., $C_1$ does not lie in the interior of $C_2$ and vice
  versa.  Let $C$ be the cycle bounding the outer face of the subgraph
  $H=C_1+C_2$ that is formed by the two decreasing cycles. By
  construction, $C_1$ and $C_2$ lie in the interior of $C$, and we
  claim that $C$ is a decreasing cycle contradicting that $C_1$ and
  $C_2$ are outermost. To that end, we show that
  $\ell_C(e)=\ell_{C_1}(e)$ for any edge~$e$ that belongs to both $C$
  and $C_1$, and $\ell_C(e)=\ell_{C_2}(e)$ for any edge $e$ that
  belongs to both $C$ and $C_2$. Hence, all edges of $C$ have a
  non-negative label since $C_1$ and~$C_2$ are decreasing. By
  Proposition~\ref{prop:horizontal_cycle} there is at least
  one label of $C$ that is positive, and hence $C$ is a decreasing
  cycle.

  It remains to show that $\ell_C(e)=\ell_{C_1}(e)$ for any edge~$e$
  that belongs to both $C$ and $C_1$; the case that $e$ belongs to
  both $C$ and $C_2$ can be handled analogously.  Let $\Gamma_H$ be
  the ortho-radial representation~$\Gamma$ restricted to $H$.  We flip
  the cylinder to exchange the outer face with the central face and
  vice versa. More precisely, Lemma~\ref{lem:flip_label} implies that
  the reverse edge $\reverse{e}$ of $e$ lies on the central face of
  the flipped representation $\flip{\Gamma_H}$ of $\Gamma_H$. Further,
  it proves that
  $\reverse{\ell}_{\reverse{C}}(\reverse{e})=\ell_{C}(e)$ and
  $\reverse{\ell}_{\reverse{C_1}}(\reverse{e})=\ell_{C_1}(e)$, where
  $\reverse{\ell}$ is the labeling in $\reverse{\Gamma_H}$.  Hence, by
  Corollary~\ref{cor:central-face-same-label} we obtain
  $\reverse{\ell}_{\reverse{C}}(\reverse{e})=
  \reverse{\ell}_{\reverse{C_1}}(\reverse{e})$.  Flipping back the cylinder, again by
  Lemma~\ref{lem:flip_label} we obtain $\ell_C(e)=\ell_{C_1}(e)$. %
\end{proof}

The core of our algorithm is an adapted left-first DFS. Given a
directed edge~$e$ it determines the outermost decreasing cycle $C$ in
$\Gamma$ such that $C$ contains $e$ in the given direction and $e$ has
the smallest label among all edges on $C$, if such a cycle exists.  By
running this test for each directed edge of $G$ as the start edge, we
find a decreasing cycle if one exists.

\newcommand{\reference}{\mathrm{ref}}
More precisely, the DFS visits each vertex at most once and it maintains for each visited vertex $v$ a reference
edge $\reference(v)$, the edge of the search tree via which $v$ was
visited. Whenever it has a choice which vertex to visit next, it
picks the first outgoing edge in clockwise direction after the
reference edge that leads to an unvisited vertex.  In addition to
that, we employ a filter that ignores certain outgoing edges during
the search.  To that end, we define for all outgoing edges $e$
incident to a visited vertex $v$ a \emph{search label}~$\tilde\ell(e)$
by setting
$\tilde\ell(e) = \tilde\ell(\reference(v)) + \rot(\reference(v) + e)$
for each outgoing edge $e$ of $v$.  In our search we ignore edges with
negative search labels.  For a given directed edge $vw$ in $G$ we 
initialize the search by
setting $\reference(w) = vw$, $\tilde\ell(vw) = 0$ and then start
searching from $w$.

Let $T$ denote the directed search tree with root $w$ constructed by
the DFS in this fashion.  If $T$ contains $v$, then this determines a
\emph{candidate cycle} $C$ containing the edge $vw$. If $C$ is a
decreasing cycle, which we can easily check by determining a
reference path, we report it.
Otherwise, we show that there is no outermost decreasing cycle $C$
such that $vw$ lies on $C$ and has the smallest label among
all edges on $C$.

\begin{figure}[t]
  \centering
  \includegraphics[]{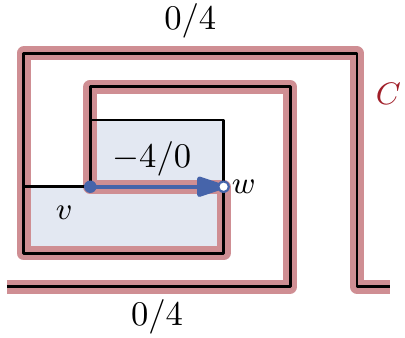}
  \caption{The search from $vw$ finds the non-decreasing cycle $C$.  Edges are
  labeled $\ell_{C}(e) / \tilde\ell(e)$.}
  \label{fig:nondecreasing_cycle_found}
\end{figure}

It is necessary to check that $C$ is essential and
decreasing.  For example the cycle in
Figure~\ref{fig:nondecreasing_cycle_found} is found by the search and
though it is essential, it is non-decreasing.  This is caused by the
fact that the label of $vw$ is actually $-4$ on this cycle but the
search assumes it to be $0$.

\begin{figure}
  \centering
    \includegraphics{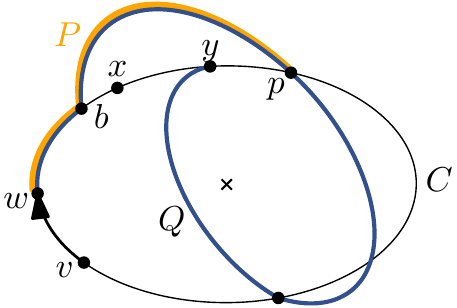}
        \caption{Path $Q$ and its prefix $P$ that leaves $C$ once and ends 
      at a vertex $p$ of $C$.}
    \label{fig:prefix}
\end{figure}

\begin{lemma}\label{lem:dfs_correctness}
  Assume that $\Gamma$ contains a decreasing cycle.  Let $C$ be the
  outermost decreasing cycle of $\Gamma$ and let $vw$ be an edge on
  $C$ with the minimum label, i.e., $\ell_C(vw) \leq \ell_C(e)$ for
  all edges~$e$ of~$C$.  Then the left-first DFS from $vw$ finds $C$.
\end{lemma}
\begin{proof}
  Assume for the sake of contradiction 
  that the search does not find
  $C$.  Let $T$ be the tree formed by the edges visited by the search.
  Since the search does not find $C$ by assumption, a part of
  $\subpath{C}{w,v}$ does not belong to $T$.  Let $xy$ be the first
  edge on $\subpath{C}{w,v}$ that is not visited, i.e.,
  $\subpath{C}{w,x}$ is a part of $T$ but $xy\not\in T$.  There are
  two possible reasons for this.  Either $\tilde\ell(xy) < 0$ or $y$
  has already been visited before via another path $Q$ from $w$ with
  $Q\neq \subpath{C}{w,y}$.

  The case $\tilde{\ell}(xy)<0$ can be
  excluded as follows.  By the construction of the labels
  $\tilde\ell$, for any path $P$ from $w$ to a vertex $z$ in $T$ and any edge
  $e'$ incident to $z$ we have $\tilde\ell(e') = \rot(vw+P+e')$.  In
  particular, $\tilde{\ell}(xy) = \rot(\subpath{C}{vw, xy}) = \ell_C(xy) -
  \ell_C(vw) \geq 0$ since the rotation can be rewritten as a label difference 
  (see Observation~\ref{obs:repr:label_difference}) and $vw$ has the smallest label 
  on $C$.
  
  Hence, $T$ contains a path $Q$ from $w$ to $x$ that was found by the
  search before and $Q$ does not completely lie on $C$. There is a
  prefix of $Q$ (possibly of length $0$) lying on $C$ followed by a
  subpath not on $C$ until the first vertex $p$ of $Q$ that again
  belongs to $C$; see Figure~\ref{fig:prefix}.  We set $P=\subpath{Q}{w, p}$  
  and denote the vertex where $P$ leaves $C$
  by $b$.  By construction the edge $vw$ lies on $\subpath{C}{p,b}$.
  The subgraph $H=P+C$ that is formed by the decreasing cycle $C$ and
  the path $P$ consists of the three internally
  vertex-disjoint paths $\subpath{P}{b,p}$, $\subpath{C}{b,p}$ and
  $\subpath{\reverse{C}}{b,p}$ between $b$ and $p$.  Since edges that
  are further left are preferred during the search, the clockwise
  order of these paths around $b$ and $p$ is fixed as in Figure~\ref{fig:prefix}.  In $H$ there are
  three faces, bounded by $C$,
  $\subpath{\reverse{C}}{b,p}+\subpath{\reverse{P}}{p,b}$ and
  $\subpath{P}{b,p}+\subpath{\reverse{C}}{p,b}$, respectively.  Since
  $C$ is an essential cycle and it bounds a face in $H$, it bounds the central face
  and one of the two other faces is the outer face. These two
  possibilities are shown in Figure~\ref{fig:search_path}.  We denote
  the cycle bounding the outer face but in which the edges are
  directed such that the outer face lies locally to the left by
  $C'$. That is, the boundary of the outer face is $\reverse{C'}$.  We
  distinguish cases based on which of the two possible cycles
  constitutes $\reverse{C'}$.

    \begin{figure}
    \centering
    \begin{subfigure}[b]{0.43\textwidth}
      \centering
      \includegraphics[]{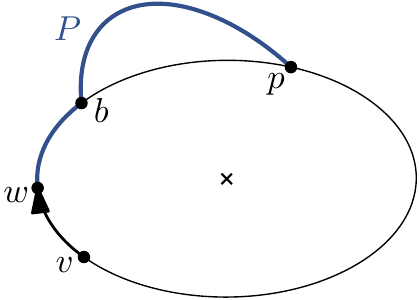}
      \subcaption{The edge $vw$ lies on the outer face of~$H$.}
      \label{fig:search_path-forward}
    \end{subfigure}
    \hspace{1em}
    \begin{subfigure}[b]{0.52\textwidth}
      \centering
      \includegraphics[page=2]{fig/path_forward.pdf}
      \subcaption{The edge $vw$ does not lie on the outer face of $H$.}
      \label{fig:search_path-backward}
    \end{subfigure} 
    \caption{The two possible embeddings of the subgraph formed by the  
    decreasing cycle $C$ and the path $P$, which was found by the  
    search.}
    \label{fig:search_path}
  \end{figure}
  
  If
  $\reverse{C'}=\subpath{\reverse{C}}{b,p}+\subpath{\reverse{P}}{p,b}$
  forms the outer face of $H$, $vw$ lies on $C'$ as illustrated in
  Figure~\ref{fig:search_path-forward} and we show that $C'$ is a
  decreasing cycle, which contradicts the assumption that $C$ is the
  outermost decreasing cycle.  Since $P$ is simple and lies in the
  exterior of $C$, the path $P$ is contained in $C'$, which
  means $\subpath{C'}{w, p}= P$. The other part of $C'$ is formed by
  $\subpath{C}{p, w}$. Since $C$ forms the central face of $H$, the
  labels of the edges on $\subpath{C}{p,w}$ are the same for $C$ and
  $C'$ by Proposition~\ref{lem:repr:equal_labels_at_intersection}. In
  particular, $\ell_C(vw) = \ell_{C'}(vw)$ and all the labels of edges
  on $\subpath{C}{p,w}$ are non-negative because $C$ is
  decreasing. The label of any edge $e$ on both $C'$ and $P$ is
  $\ell_{C'}(e) = \ell_{C'}(vw) + \rot(vw + \subpath{P}{w, e}) =
  \ell_C(vw) + \tilde{\ell}(e) \geq 0$.  Thus, the labeling of $C'$ is
  non-negative. Further, not all labels of $C'$ are $0$ since
  otherwise $C$ would not be a decreasing cycle by
  Proposition~\ref{prop:horizontal_cycle}.  Hence, $C'$
  is decreasing and contains $C$ in its interior, a contradiction.
  
  If $\reverse{C'} = \subpath{\reverse{C}}{p,b}+\subpath{P}{b,p}$, 
  the edge $vw$ does not lie on $C'$; see 
  Figure~\ref{fig:search_path-backward}.
  We show that $C'$ is a decreasing cycle containing $C$ in 
  its interior, again contradicting the choice of $C$.
  As above, Proposition~\ref{lem:repr:equal_labels_at_intersection}
  implies that the common edges of $C$ and $C'$ have the same labels
  on both cycles.  It remains to show that all edges $xy$ on
  $\subpath{\reverse{P}}{p,b}$ have non-negative labels.  To establish
  this we use paths to the edge that follows $b$ on $C$.
  This edge $bc$ has the same label on both cycles and thus provides a handle
  on $\ell_{C'}(xy)$.  We make use of the following equations, which
  follow immediately from the definition of the (search) labels.
  \begin{align*}
    \ell_{C'}(bc) &= \ell_{C'}(xy) + \rot(\subpath{\reverse{P}}{xy,db}) + \rot(dbc) = \ell_{C'}(xy) - \rot(\subpath{P}{bd,yx}) - \rot(cbd)\\
    \ell_{C}(bc)  &= \ell_{C}(vw)  + \rot(\subpath{C}{vw,ab})           + \rot(abc)\\
    \tilde\ell(yx) &= \rot(\subpath{C}{vw,ab}) + \rot(abd) + \rot(\subpath{P}{bd,yx})
  \end{align*}
  Since $\ell_C(bc) = \ell_{C'}(bc)$ and $\rot(abd) = -\rot(dba)$, we
  thus get
  \begin{align*}
    \ell_{C'}(xy) &= \ell_{C}(vw) + \rot(\subpath{C}{vw,ab}) + \rot(abc) +\rot(\subpath{P}{bd,yx}) + \rot(cbd)\\
                 &= \ell_{C}(vw) + \tilde\ell(yx) + \rot(dba) + \rot(abc) + \rot(cbd).
  \end{align*}
  Since $\ell_C(vw) \ge 0$ and $\tilde\ell(yx) \ge 0$ (as $yx$ was not filtered 
  out), it follows that
  $\ell_{C'}(xy) \ge \rot(abc) + \rot(dba) + \rot(cbd) = 2$ as this is
  the sum of clockwise rotations around a degree-3 vertex.  Hence,
  $C'$ is decreasing and contains $C$ in its interior, a
  contradiction.
  Since both embeddings of $H$ lead to a contradiction, we obtain a
  contradiction to our initial assumption that the search fails to
  find $C$.
\end{proof}

  The left-first DFS clearly runs in $\O(n)$ time.  In order to
  guarantee that the search finds a decreasing cycle if one exists, we
  run it for each of the $\O(n)$ directed edges of $G$.  Since some
  edge must have the lowest label on the outermost decreasing cycle,
  Lemma~\ref{lem:dfs_correctness} guarantees that we eventually find a
  decreasing cycle if one exists.  Increasing cycles can be found by
  searching for decreasing cycles in the mirror representation
  $\mirror{\Gamma}$ (Lemma~\ref{lem:mirroring_label}).
  Altogether, this proves Theorem~\ref{thm:main-result:validity}.

\section{Efficient Rectangulation Procedure}\label{sec:efficient_rectangulation}

Let~$G$ be a planar 4-graph with valid ortho-radial
representation~$\Gamma$.  By Theorem~\ref{thm:main-result:drawable}
$\Gamma$ is drawable.  The proof of Theorem~\ref{thm:main-result:drawable}
is constructive and shows how to augment~$\Gamma$ to a valid
rectangular ortho-radial representation~$\Gamma^\star$.  Then a
drawing $\Delta^\star$ of~$\Gamma^\star$ can be computed by
determining flows in two flow networks by
Theorem~\ref{thm:rect:flows-to-drawing}.  A drawing $\Delta$
of~$\Gamma$ can be finally obtained by undoing augmentation steps.
Hence, it remains to show how a valid ortho-radial representation can
be rectangulated efficiently.

To follow the construction of the proof of Theorem~\ref{thm:main-result:drawable}, one needs
to determine efficiently whether an augmentation yields a valid
ortho-radial representation~$\Gamma^u_{vw}$.  We call such an
augmentation a \emph{valid augmentation}.
Since each valid augmentation reduces the number of concave angles, we
obtain a rectangulation after $\O(n)$ valid augmentations.  Moreover,
there are $\O(n)$ candidates for each augmentation, each of which can
be tested for validity (and increasing/decreasing cycles can be
detected) in $\O(n^2)$ time by Theorem~\ref{thm:main-result:validity}. Thus, the
rectangulation algorithm can be implemented to run in $\O(n^{4})$ time.

In the remainder of this section we present an improvement to
$\O(n^{2})$ time, which is achieved in three steps.  First, we show
that, due to the nature of augmentations, each validity test can be done
in $\O(n)$ time.  This reduces the running time of the augmentation
procedure to $\O(n^3)$; see Section~\ref{sec:faster-validity-test:first}.  Second, we show how to find a valid
augmentation for a given port $u$ using only $\O(\log n)$ validity
tests, thus improving the running time to $\O(n^2 \log n)$; see Section~\ref{sec:faster-validity-test:second}.  Finally,
we design an algorithm that can be used to post-process a valid
augmentation and which reduces the number of validity tests to $\O(n)$ and the
running time to $\O(n^2)$ in total; see
Section~\ref{sec:faster-validity-test:third}.
Altogether, this proves our third main result.

\mrDrawing

\subsection{1st Improvement -- Faster Validity Test}
\label{sec:faster-validity-test:first}

The general test for strictly monotone cycles performs one left-first
DFS per edge and runs in $\O(n^2)$ time. However, we can
exploit the special structure of the augmentation to reduce the
running time to $\O(n)$. For the proof we restrict
ourselves to the case that the inserted edge $uz$ points to the right.
The case that it points left can be handled by flipping the
representation using Lemma~\ref{lem:flip_label}.
 
The key result is that in any decreasing cycle of an augmentation the new
edge $uz$ has the minimum label.  Thus, performing only one left-first DFS 
starting at $uz$ is sufficient.  For increasing
cycles the arguments do not hold, but in a second step we show that
the test for increasing cycles can be replaced by a simple test for
horizontal paths.

Recall that the augmentations $\Gamma^u_{vw}$ that are tested during
the rectangulation are built by adding one edge $uz$ to a valid
representation~$\Gamma$. Hence, any strictly monotone cycle in
$\Gamma^u_{vw}$ contains the edge $uz$.
 
We first show that the new edge $uz$ has label $0$ on any
decreasing cycle in the augmentation $\Gamma^u_{vw}$ if $vw$ is the first
candidate. We extend this result afterwards to augmentations to all
candidates. Since the label of edges on decreasing cycles is
non-negative, this implies in particular that the label of $uz$ is
minimum, which is sufficient for the left-first DFS to succeed (see
Lemma~\ref{lem:dfs_correctness}).

\begin{figure}[t]
    \centering
    \includegraphics{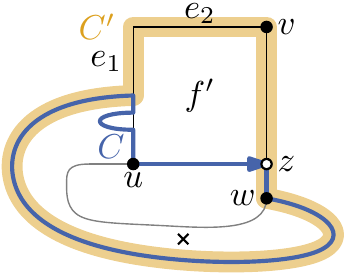}
    \caption{A decreasing cycle $C$ that uses $uz$ and an essential 
      cycle $C'$ derived from $C$.}
    \label{fig:first_candidate}
\end{figure}

\begin{lemma}
  \label{lem:label_first_candidate}
  Let $\Gamma$ be a valid ortho-radial representation and let $u$ be a
  horizontal port on $f$ with first candidate $vw$. If $\Gamma^u_{vw}$
  contains a decreasing cycle~$C$, then $C$ contains $uz$ in this
  direction and $\ell_C(uz) = 0$.
\end{lemma}
  
\begin{proof}
  We first consider the case that $C$ uses $uz$ (and not $zu$) and
  assume for the sake of contradiction that $\ell_C(uz)\neq 0$; see
  Figure~\ref{fig:first_candidate}. Since $uz$ points right,
  $\ell_C(uz)$ is divisible by $4$. Together with $\ell_C(uz) \ge 0$
  because $C$ is decreasing, we obtain $\ell_C(uz) \geq 4$. By
  Lemma~\ref{lem:reroute-face} there is an essential cycle $C'$
  without $uz$ in the subgraph $H$ that is formed by the new
  rectangular face $f'$ and $C$.  We show that $C'$ is a decreasing
  cycle. We observe that each edge of $C'$ either lies on $f'$ or is
  an edge of $C$.  For any $e \in C' \cap C$ Lemma~\ref{lem:reroute-face} gives
  $\ell_{C'}(e) = \ell_C(e) \geq 0$.
  Since $f'$ is rectangular, the labels of edges in $C'\cap f'$ differ by
  at most $1$ from $\ell_C(uz)$. By assumption it is
  $\ell_C(uz) \geq 4$ and therefore $\ell_{C'}(e)\geq 3$ for all edges
  $e\in C'\cap f'$.  Hence, $C'$ is a decreasing cycle in $G$
  contradicting the validity of $\Gamma$.
  
  If $zu\in C$, it is $\ell_C(zu)\geq 2$ and a similar argument yields a 
  decreasing cycle in~$\Gamma$.
\end{proof}

While the same statement does not generally hold for all candidates,
it does hold if the first candidate creates a decreasing cycle.

\begin{figure}[t]
  \centering
  \includegraphics{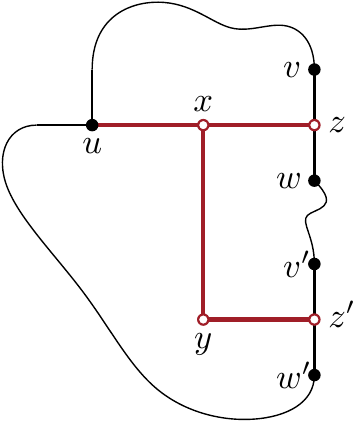}
  \caption{The structure used to simulate the simultaneous insertion of $uz$ to 
    $vw$ and $uz'$ to $v'w'$.}
  \label{fig:structure}
\end{figure}

\begin{lemma}
  \label{lem:label_any_candidate}
  Let $\Gamma$ be a valid ortho-radial representation and let $u$ be a
  horizontal port on $f$ with first candidate $vw$. Further, let
  $v'w'$ be another candidate and denote the edge inserted in
  $\Gamma^u_{v'w'}$ by $uz'$.  If $\Gamma^u_{vw}$ contains a
  decreasing cycle, any decreasing cycle $C'$ in $\Gamma^u_{v'w'}$
  uses $uz'$ in this direction and $\ell_{C'}(uz') = 0$.
\end{lemma}

\begin{proof}
  In order to simulate the simultaneous insertion of 
  two new edges to both $vw$ and $v'w'$ we use the structure from
  the proof of Lemma~\ref{lem:horizontal-path}; see 
  Figure~\ref{fig:structure}. We denote the resulting 
  augmented representation by $\tilde\Gamma$.
  There is a one-to-one correspondence between decreasing cycles in 
  $\Gamma^u_{vw}$ and decreasing cycles in $\tilde{\Gamma}$ containing $uxz$. 
  Let $C$ be a decreasing cycle in $\tilde{\Gamma}$ containing $uxz$. By 
  Lemma~\ref{lem:label_first_candidate} the cycle $C$ contains $uxz$ in this 
  direction, and 
  we have $\ell_{C}(ux) = 0$.
  
  Similarly, for any decreasing cycle in $\Gamma^u_{v'w'}$ there is a 
  cycle in $\tilde{\Gamma}$ where $uz'$ ($z'u$) is replaced by the path $uxyz'$ 
  ($z'yxu$). Let $\tilde{C'}$ be the cycle in $\tilde{\Gamma}$ that 
  corresponds to the decreasing cycle $C'$ in $\Gamma^u_{v'w'}$.

  Suppose for now that $C'$ uses $uz'$ in this direction, which means that
  $\tilde{C'}$ uses $ux$. In particular, since 
  $\ell_{\tilde{C'}}(ux) = \ell_{C'}(uz')$, it holds $\ell_{\tilde{C'}}(xy)=
  \ell_{\tilde{C'}}(ux)+1\geq 1$.
  Let $\tilde{f}$ be the central face of 
  $H=\tilde{C}+\tilde{C'}$. We note that $\tilde{f}$ is a decreasing cycle by 
  Lemma~\ref{lem:repr:equal_labels_at_intersection} and 
  Proposition~\ref{prop:horizontal_cycle}. Since $\Gamma$ is valid, $\tilde{f}$ 
  is not exclusively formed by edges of~$G$.
  Thus, by the construction of $H$, the path $uxyz'$ lies on~$\tilde{f}$.
  Lemma~\ref{lem:repr:equal_labels_at_intersection} therefore implies that
  $\ell_{\tilde{C'}}(ux)=\ell_{C}(ux)=0$, where the last equality
  follows from Lemma~\ref{lem:label_first_candidate}.
  
  Above we assumed that $C'$ uses $uz'$ in this direction. This is in
  fact the only possibility. Assume for the sake of contradiction that
  $C'$ contains $z'u$ and hence $xu\in\tilde{C'}$. As above we can
  argue that the central face $\tilde{f}$ of $\tilde{C}+\tilde{C}'$ is not exclusively
  formed by edges of $G$ and that the path $z'yxz$ lies on
  $\tilde{f}$. By Lemma~\ref{lem:repr:equal_labels_at_intersection} we
  have
  $\ell_{\tilde{C'}}(yx) = \ell_{C}(xz) - \rot(yxz) = -1$.
  Hence, we obtain
  $\ell_{C'}(z'u) = \ell_{\tilde{C'}}(xu) = \ell_{\tilde{C'}}(yx) +
  \rot(yxu) = -2$, which contradicts that $C'$ is a
  decreasing cycle.
\end{proof}

\begin{figure}[t]
    \centering
    \includegraphics{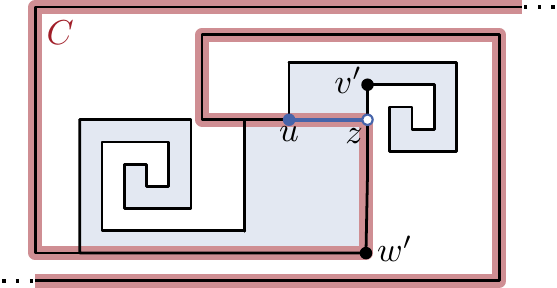}
    \caption{Here, the insertion of the edge~$uz$ to the last 
      candidate~$v'w'$ introduces an increasing cycle~$C$ with $\ell_C(uz) = 
      -4$.}
    \label{fig:increasing_cycle_negative}
\end{figure}

Altogether, we can efficiently test which of the candidates~$e_1,\dots,e_k$ 
produce decreasing cycles as follows. By Lemma~\ref{lem:label_first_candidate}, 
if the first candidate is not valid, then $\Gamma^{u}_{e_{1}}$ has
a decreasing cycle that contains the new edge $uz$ with label~$0$,
which is hence the minimum label for all edges on the cycle.  This can
be tested in $\O(n)$ time by Lemma~\ref{lem:dfs_correctness}.
Fact~\ref{prop:augmentation-properties:fact:monotone-cycle} of Proposition~\ref{prop:augmentation-properties} guarantees that we either find a valid
augmentation or a decreasing cycle.  In the former case we are done,
in the second case Lemma~\ref{lem:label_any_candidate} allows us to
similarly restrict the labels of $uz$ to $0$ for the remaining
candidate edges, thus allowing us to detect decreasing cycles in
$\Gamma^u_{e_{i}}$ in $\O(n)$ time for $i=2,\dots,k$.

It is tempting to use the mirror symmetry
(Lemma~\ref{lem:mirroring_label}) to exchange
increasing and decreasing cycles to deal with increasing cycles in an
analogous fashion.  However, this fails as mirroring invalidates the
property that $u$ is followed by two right turns in clockwise
direction. For example, in Figure~$\ref{fig:increasing_cycle_negative}$ 
inserting the edge to the last candidate introduces an increasing cycle~$C$ 
with $\ell_{C}(uz)=-4$. We therefore give a direct algorithm for
detecting increasing cycles in this case.  %

Let $e_{i}=v_iw_i$ and $e_{i+1}={v_{i+1}w_{i+1}}$ be two consecutive
candidates for $u$ such that $\Gamma^u_{e_{i}}$ contains a decreasing
cycle but $\Gamma^u_{e_{i+1}}$ does not. If $\Gamma^u_{e_{i+1}}$
contains an increasing cycle, then by
Fact~\ref{prop:augmentation-properties:fact:horizontal-cycle} of
Proposition~\ref{prop:augmentation-properties} the vertices $w_{i}$,
$v_{i+1}$ and $u$ lie on a horizontal path that starts at a vertex $z$
incident to $f$ and ends at $u$. The presence of such a horizontal
path $P$ can clearly be checked in linear time, thus allowing us to
also detect increasing cycles provided that the previous candidate
produced a decreasing cycle. If $P$ exists, we insert the edge
$uz$. By Proposition~\ref{prop:horizontal_cycle} this does not produce
strictly monotone cycles. Otherwise, if $P$ does not exist, the
augmentation
$\Gamma^u_{e_{i+1}}$ is valid. In both cases we have resolved the
horizontal port~$u$ successfully. %

Summarizing, the overall algorithm for augmenting from a horizontal
port $u$ now works as follows.  By exploiting
Lemmas~\ref{lem:label_first_candidate}
and~\ref{lem:label_any_candidate}, we test the candidates in the order
as they appear on $f$ until we find the first candidate $e$ for which
$\Gamma^u_{e}$ does not contain a decreasing cycle.  Using
Fact~\ref{prop:augmentation-properties:fact:horizontal-cycle} of Proposition~\ref{prop:augmentation-properties} we either find that $\Gamma^{u}_{e}$
is valid, or we find a horizontal path as described above. In both
cases this allows us to determine an edge whose insertion does not
introduce a strictly monotone cycle.  Since in each test for a
decreasing
cycle the edge $uz$ can be restricted to have label~$0$, each of the
tests takes linear time.  This improves the running time of the
rectangulation algorithm to $\O(n^{3})$.

\subsection{2nd Improvement -- Fewer Validity Tests}\label{sec:faster-validity-test:second}

Instead of linearly searching for a suitable candidate for $u$, we can
employ a binary search on the candidates, which reduces the number of
validity tests for $u$ from linear to logarithmic. To do this
efficiently, we first compute the list of all candidates
$e_1,\dots,e_k$ for $u$ in time linear in the size of $f$. Next, we
test if the augmentation $\Gamma^u_{e_1}$ is valid. If it is, we are
done.

Otherwise, we start the binary search on the list $e_1,\dots,e_k$,
where $k$ is the number of candidates for $u$.  The search maintains a
sublist $e_i,\dots,e_j$ of consecutive candidates such that
$\Gamma^u_{e_i}$ contains a decreasing cycle and $\Gamma^u_{e_j}$ does
not. Note that this invariant holds in the beginning, because we
explicitly test for a decreasing cycle in $\Gamma^u_{e_1}$ and there
is no decreasing cycle in $\Gamma^u_{e_{k}}$ by
Fact~\ref{prop:augmentation-properties:fact:monotone-cycle} of
Proposition~\ref{prop:augmentation-properties}.  If the list consists
of only two consecutive candidates, i.e., $j=i+1$, we stop.
Otherwise, we set $m=\lfloor (i+j)/2\rfloor$ and test if
$\Gamma^u_{e_m}$ contains a decreasing cycle. If it does, we recurse
on $e_m,\dots,e_j$ and otherwise on $e_i,\dots,e_m$.  As the invariant
is preserved, we end up with two consecutive candidates $e_i$ and
$e_{i+1}$ such that $\Gamma^u_{e_i}$ contains a decreasing cycle and
$\Gamma^u_{e_{i+1}}$ does not.  In this situation
Fact~\ref{prop:augmentation-properties:fact:horizontal-cycle} of
Proposition~\ref{prop:augmentation-properties} guarantees that we find
a valid augmentation. Clearly this only requires $O(\log n)$ validity
tests in total. Further, as argued in the previous subsection each of
these tests can be performed in linear time. Altogether, we obtain the
following lemma.

\begin{lemma}\label{lem:binary_search}
  Using binary search we find a valid augmentation for $u$ in $\O(n\log n)$ 
  time.
\end{lemma}

Since there are $\O(n)$ ports to remove, we obtain that any
planar $4$-graph with valid ortho-radial representation can be
rectangulated in $\O(n^2\log n)$ time.

\begin{thm}
  \label{thm:augmentation-binary-search}
  Given a valid ortho-radial representation $\Gamma$, a
  corresponding rectangulation can be computed in $\O(n^2 \log n)$
  time.
\end{thm}

\subsection{3rd Improvement -- Linear Number of Validity Tests}
\label{sec:faster-validity-test:third}
\label{sec:2-phase-appendix}

In this section we describe an improvement of our algorithm that
reduces the total number of validity tests to $\O(n)$ such that the
running time of our algorithm becomes $\O(n^2)$. The improvement
adapts the augmentation step for horizontal ports of the
rectangulation procedure. Let $u$ be a horizontal port of a face $f$
and let $e_1,\dots,e_k$ be its candidate edges. The adapted
augmentation step resolves $u$ in two steps. 

\begin{figure}[t]
  \centering
  \begin{subfigure}[b]{0.49\textwidth}
    \centering
    \includegraphics[page=4,width=\textwidth]{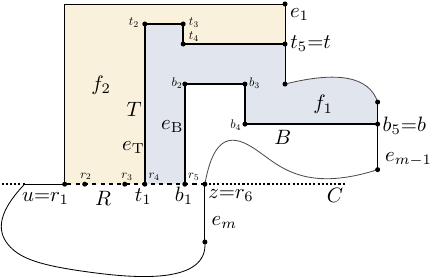}
    \subcaption{$\Gamma_1$: $r_2r_3$ points right.}
    \label{fig:second-phase:step1:right}
  \end{subfigure}
  \begin{subfigure}[b]{0.49\textwidth}
    \centering
    \includegraphics[page=2,width=\textwidth]{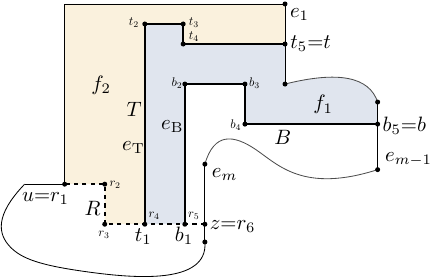}
    \subcaption{$\Gamma_1$: $r_2r_3$ points downwards.}
    \label{fig:second-phase:step1:downwards}
  \end{subfigure} 
  \caption{Illustration of Step 1, which inserts $R$, $T$ and $B$
    into $\Gamma_0$. Depending on whether $uz$ lies on a
    horizontal cycle~$C$ in $\Gamma_0$, the edge $r_2r_3$ points
    (\subref{fig:second-phase:step1:right}) to the right or
    (\subref{fig:second-phase:step1:downwards}) downwards.}
  \label{fig:second-phase:step1}
\end{figure}

\begin{compactitem}
\item[\textit{Step 1.}] We do a linear scan on $e_{1},\dots,e_k$ to
  search for the first candidate~$e_m$ of $u$ that gives rise to a
  valid augmentation with an additional edge $uz$. Recall that $z$
  either subdivides $e_m$ or $uz$ completes a horizontal cycle that
  contains the source of $e_m$. We note that we apply $m$ validity
  tests in this step. In case that $m<4$, the augmentation step is
  stopped. Otherwise we continue with the following step.

\item[\textit{Step 2.}]  In the following let $\Gamma_0$ be the valid
  ortho-radial representation that we obtain after Step~1. We observe
  that the inserted edge $uz$ splits the face~$f$ into two smaller
  faces. In Step 2.1 we partition the face that contains the
  edges $e_1,\dots,e_{m-1}$ by inserting two paths $T$ and $B$ of
  constant size that start at subdivision vertices on $uz$ and end at
  subdivision vertices on $e_1$ and $e_{m-1}$, respectively; see
  Figure~\ref{fig:second-phase:step1}. These two paths separate two sub-faces
  $f_1$ and $f_2$ of $f$ that contain the candidate edges
  $e_1,\dots,e_{m-2}$, which are all but a constant number of edges
  for which we have performed validity tests in Step 1. In Step
  2.2. we rectangulate $f_1$ without performing any validity test and
  in Step 2.3 we rectangulate $f_2$ performing a number of validity
  tests that is proportional to the size of $f_2$.  In the following
  we describe the three steps in greater detail.
  \begin{compactitem}
  \item[\textit{Step 2.1.}] We adapt $\Gamma_0$ as follows.  We introduce a
    path $B=b_1\dots b_5$ in $f$ that connects a subdivision
    vertex~$b_1$ on $uz$ with a subdivision vertex on $e_{m-1}$; see
    Figure~\ref{fig:second-phase:step1}. The edge $b_1b_2$ points
    upwards, the edge~$b_3b_4$ points downwards, and the other two
    edges point to the right. The face that lies to the left of $B$ is
    the face~$f'$ we seek to rectangulate. Similarly, we introduce a
    path $T=t_1\dots t_5$ in $f'$ that connects a subdivision
    vertex~$t_1$ on $ub_1$ with a subdivision vertex on $e_{1}$. The
    edge $t_1t_2$ points upwards, the edge~$t_3t_4$ points downwards,
    and the other two edges point to the right.  Finally, we subdivide
    the edge $ut_1$ by two additional vertices. Altogether, the edge
    $uz$ has been replaced by a path $R=r_1\dots r_6$ with $r_1=u$,
    $r_4=t_1$, $r_5=b_1$ and $r_6=z$. As we have obtained the edges of
    $R$ by subdividing $uz$, they all point to the right. In the case
    that $uz$ does not lie on a horizontal cycle, we orient $r_2r_3$
    such that it points downwards.

  We denote the resulting ortho-radial representation by
  $\Gamma_1$. Further, let $f_1$ be the face that lies to the
  right of $T$, and let $f_2$ be the face that lies to the
  left of $T$.

\item[\textit{Step 2.2.}]  We iteratively resolve the ports in $f_1$
  until the face is rectangulated. For each port $u'$ of $f_1$ we
  augment the ortho-radial representation as follows.  If $u'$ is a
  vertical port, we augment with respect to the first candidate edge
  and if $u'$ is a horizontal port we augment with respect to the last
  candidate edge of $u'$.  The procedure stops when $f_1$ is completely
  rectangulated. We denote the resulting ortho-radial representation
  by~$\Gamma_2$. We emphasize that this step does not execute any
  validity test.

\item[\textit{Step 2.3.}] Starting with $\Gamma_2$, we rectangulate
  the face $f_2$, which has a constant number of ports, by iteratively
  applying the original augmentation step. We denote the resulting
  ortho-radial representation by $\Gamma_3$.
\end{compactitem}
\end{compactitem}
In the following we show that the modified rectangulation procedure
runs in $\O(n^2)$ time and yields a valid ortho-radial
representation.

\begin{lemma}\label{lem:second-phase:correctness}
  The modified rectangulation procedure produces a valid, rectangulated ortho-radial
  representation.
\end{lemma}
Since the proof is rather technical, we defer it to Section~\ref{lem:second-phase:correctness}. In
the following we argue the running time.  To that end, we first prove
that the output rectangulation has linear size.

\begin{lemma}\label{lem:linear-output-size}
  The rectangulated ortho-radial representation produced by the modified
  rectangulation procedure has size $\O(n)$.
\end{lemma}

\begin{proof}
  For the proof we define the potential function
\[\Phi=3\cdot \text{horizontal corners of }\Gamma + \text{vertical
    corners of }\Gamma,\] where a \emph{horizontal (vertical)} corner
is a concave corner that becomes a horizontal (vertical) port of a
face during the rectangulation procedure.  At the beginning of the
rectangulation procedure it holds $\Phi \leq 4 n$, because each vertex
with degree greater than 1 can be either a horizontal or a vertical
corner, but not both. Further, each vertex with degree 1 is both a
horizontal and vertical corner.  We show that for each augmentation
step of the rectangulation procedure the potential $\Phi$ decreases by
some value $\Delta\Phi\geq 1$ and that the number~$\Delta V$ of
inserted vertices is proportional to $\Delta\Phi$. Since the
rectangulation procedure terminates with $\Phi=0$,
Lemma~\ref{lem:linear-output-size} follows.

In case that the augmentation step handles a vertical port, this
vertical corner is resolved but no new corner is created. Hence,
$\Phi$ decreases by $\Delta\Phi=1$. Moreover, $\Delta V = 1$.
For the case that the augmentation step handles a horizontal port, we
distinguish two sub-cases. If $m<4$, this horizontal corner is resolved
but no new corner is created. Hence, $\Phi$ decreases by
$\Delta\Phi=3$. Moreover, $\Delta V \leq 1$  (in
case that the augmentation step closes a horizontal cycle we have
$\Delta V = 0$). Now consider the case that
$m\geq 4$. Let $k_1$ and $k_2$ be the number of the horizontal and
vertical corners, respectively, that are resolved during the
rectangulation of $f_1$ and $f_2$ in Step 2.2 and Step 2.3, excluding
those that lie on $R$, $T$ and $B$.

Due to the insertion of the vertices $r_3$ and $b_4$ the potential
increases by at most 2. Further, by resolving the horizontal corner
$u$ the potential decreases by $3$. Further, rectangulating $f_1$ and
$f_2$ decreases the potential by $3k_1 + k_2$. Altogether, we obtain
$\Delta\Phi \geq (3 - 2) + 3k_1 + k_2\geq 1 + k_1 + k_2$.  Moreover,
in Step 2.1 we add at most $13$ vertices. In Step 2.2 we add $k_1+k_2$
vertices for the corners not on $T$ and $B$, and we add $3$ vertices
for $t_4$, $b_2$, and $b_3$.  In Step 2.3 we add at most $3$ vertices
for $t_2$, $t_3$, and possibly $r_2$.  In total, we add 
$\Delta V \leq 19 + k_1 + k_2$ vertices.

In all cases $\Delta V \leq 19 \Delta \Phi$.  Altogether, we obtain
that the rectangulation procedure terminates (since $\Phi$ decreases
in every rectangulation step) and the resulting rectangulation has
$\O(n)$ vertices, and therefore $\O(n)$ edges.
\end{proof}

\begin{lemma}
  The modified rectangulation procedure applies $\O(n)$ validity tests
  and runs in $\O(n^2)$ time.
\end{lemma}

\begin{proof}
We now show that by replacing the original augmentation step with this
adaption, the number of validity tests is linear and the
rectangulation procedure runs in $\O(n^2)$ time.  We use a charging
argument that assigns to each vertex and to each edge of the output
rectangulation a constant number of the validity tests that have been
applied during the rectangulation procedure. Further, we distribute
the total running time such that running time linear in the output size is
assigned to each vertex and each edge. Hence, by
Lemma~\ref{lem:linear-output-size} the rectangulation procedure
applies $\O(n)$ validity tests and runs in $\O(n^2)$ time. 

Let $u$ be a port that is considered in a rectangulation step. If $u$
is a vertical port, we determine its first candidate edge in $\O(n)$
time, which we charge on $u$. Further, we do not apply a validity
test.

If $u$ is a horizontal port, we determine its candidate edges in
$\O(n)$ time, which we charge on $u$. Further, we apply $m$ validity
tests in Step 1. If $m<4$, only a constant number of validity tests is
applied, which we charge on $u$. Otherwise, we charge the validity
tests of $e_1$, $e_{m-1}$ and $e_{m}$ on $u$. We observe that after
the augmentation step, the horizontal port $u$ is resolved and no
further validity tests can be charged on $u$. Further, we charge the
validity test of $e_i$ with $2\leq i \leq m-2$ on $e_i$. Hence, since
by construction the candidates
$e_2,\dots,e_{m-2}$ belong to rectangles after the augmentation step,
each edge can only be charged twice (once from each side).

Step 2.1 has constant running time and applies no validity tests. We
charge the running time on $u$. Step~2.2 requires no validity tests
and resolves each concave corner~$v$ on $f_1$ in $\O(n)$ time, which
we charge on $v$. Afterwards, the face $f_1$ is rectangulated and $v$
cannot be charged again. In Step 2.3 we need $\O(n)$ time for each
concave corner $v$ to identify the candidate edges, which we charge on
$v$. Further, when applying a validity test in this step, we charge it
on the corresponding candidate edge. As $f_2$ has only a constant
number of concave corners, each edge of $f_2$ is charged with at most
a constant number of validity tests. Further, since $f_2$ is
rectangulated afterwards its edges cannot be charged again (from the
side of $f_2$).

Hence, over all rectangulation steps we obtain $\O(n)$
validity tests and running time $\O(n^2)$.
\end{proof}

Altogether, we obtain that any planar $4$-graph with valid
ortho-radial representation can be rectangulated in $\O(n^2)$
time.  Using Corollary~\ref{cor:draw:characterization} this further
implies that a corresponding bend-free ortho-radial drawing can be computed in
$\O(n^2)$ time.

\begin{thm}
  \label{thm:augmentation-third-improvement}
  Given a valid ortho-radial representation $\Gamma$, a corresponding
  rectangulation can be computed in $\O(n^2)$ time.
\end{thm}

\subsubsection{Proof of Lemma~\ref{lem:second-phase:correctness}}
We prove Lemma~\ref{lem:second-phase:correctness} by showing that each
rectangulation step yields a valid ortho-radial representation.  By
Proposition~\ref{prop:augmentation-properties} Step 1 yields a valid
ortho-radial representation $\Gamma_0$. If Step 2 is not considered,
$\Gamma_0$ is the output of the rectangulation step. So assume that
Step 2 is executed.  We use the same notation as in the description of
the algorithm. In particular, Step 2.1, Step 2.2 and Step 2.3 produce
the ortho-radial representations $\Gamma_1$, $\Gamma_2$ and $\Gamma_3$
from the ortho-radial representation of the preceding step. In the
following we consider the sub-steps of Step 2 separately and show that
$\Gamma_1$, $\Gamma_2$ and $\Gamma_3$ are valid.

\paragraph{Correctness of Step 2.1.}
In order to show the correctness of the first step we successively add
the paths $R$, $T$ and $B$ to $\Gamma_0$ and prove the validity of
each created ortho-radial representation. To that end, let
$\Gamma_R=\Gamma_0-uz+R$, $\Gamma_T=\Gamma_R+T$ and
$\Gamma_B=\Gamma_T+B=\Gamma_1$.

\begin{lemma}\label{lem:second-phase:no-monotone-cycle-on-R}
  The ortho-radial representation $\Gamma_R$ is valid.
\end{lemma}

\begin{proof}
  Let $C$ be an essential cycle in $\Gamma_R$. If $C$ does not contain $R$, it
  is contained in $\Gamma_0$. Since $\Gamma_0$ is valid, $C$ is not strictly monotone.
  It remains to consider the case that $C$ uses $R$.

  If the edge $r_2r_3$ of $R$ points to the right, $R$ lies on a
  horizontal cycle by construction. By
  Proposition~\ref{prop:horizontal_cycle} $C$ is not strictly
  monotone.  If $r_2r_3$ points downwards, the cycle $C$ corresponds
  to an essential cycle~$C'=C[z,u]+uz$ in $\Gamma_0$. By
  Lemma~\ref{lem:same-labels-except-from-edge} the labels of $C$ and
  $C'$ coincide on $C[z,u]$. By construction $C'$ is not
  horizontal. Hence, it contains edges $e$ and $e'$ with
  $\ell_{C'}(e)=1$ and $\ell_{C'}(e') = -1$.  Since $uz$ is a
  horizontal edge, we have $e\neq uz \neq e'$. Therefore,
  $\ell_{C}(e) = 1$ and $\ell_{C}(e') = -1$, that is, $C$ is not
  strictly monotone.
\end{proof}

Next, we prove that $\Gamma_T$ and $\Gamma_B$ are valid. To that end,
we introduce the following definition. A \emph{cascading cycle} is a
non-monotone essential cycle that can be partitioned into two paths
$P$ and $Q$ such that the labels on $P$ are $-1$ and the labels on $Q$
are non-negative.  We further require that the edges incident to the
internal vertices of $P$ either all lie in the interior of $C$ or they
all lie in the exterior of $C$. In the first case we call $C$ an
\emph{outer} cascading cycle and in the second case an \emph{inner}
cascading cycle. The path $P$ is the \emph{negative path} of the
cycle.  We first prove the following two general lemmata on cascading
cycles. The first lemma shows that a cascading cycle cannot be
\emph{crossed} by an increasing cycle.

\begin{lemma}\label{lem:cascading_and_increasing_cycles}
  Let $C_1$ be a cascading cycle and $C_2$ an increasing cycle. Either
  $C_1$ lies in the interior of $C_2$ or vice versa.
\end{lemma}

  \begin{figure}
    \centering
    \includegraphics{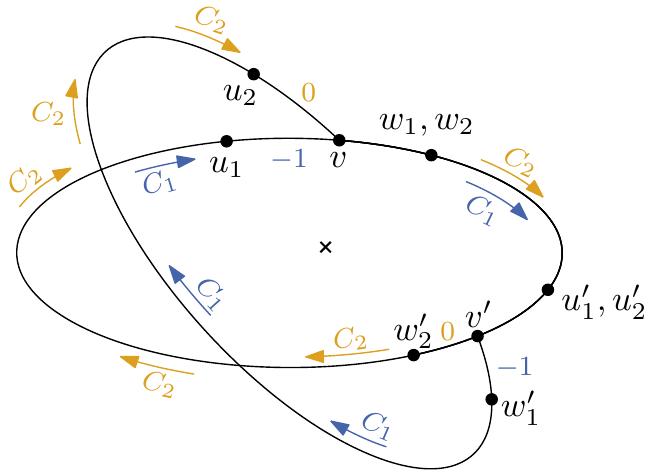}
    \caption{Illustration of proof for
      Lemma~\ref{lem:cascading_and_increasing_cycles}. It is assumed
      that the central face $g$ is neither the outer cascading cycle
      $C_1$ nor the increasing cycle $C_1$. Further, $C_1$ and $C_2$
      have the edges $g[v,v']$ in common.  It is proven that any edge of $C_1[v,v']$ 
      has label $0$ and any edge of $C_1[v',v]$ has label
      $-1$, which contradicts that $C_1$ is a cascading cycle.  }
    \label{fig:lem:cascading-increasing-cycles}
  \end{figure}

\begin{proof}
  We assume without loss of generality that $C$ is an outer cascading cycle. 
  The 
  case that it is an inner cascading cycle can be handled by flipping the 
  cylinder, which exchanges the exterior and interior of essential cycles but 
  keeps the labels.

  Let $g$ be the central face of the subgraph formed by the cycles
  $C_1$ and $C_2$. Assume that $g$ is neither $C_1$ nor $C_2$. Hence,
  there exist (not necessarily distinct) vertices $v$ and $v'$ on $g$
  such that $g[v,v']$ belongs to both $C_1$ and $C_2$, the edge that
  precedes $v$ on $g$ does not belong to $C_2$, and the edge that
  succeeds $v'$ on $g$ does not belong to $C_1$; see Figure~\ref{fig:lem:cascading-increasing-cycles}.
  
  For $i=1,2$ let $u_i$ and $w_i$ be the vertices of $C_i$ that
  precede and succeed $v$, respectively. Further, let $u'_i$ and
  $w'_i$ be the vertices of $C_i$ that precede and succeed $v'$,
  respectively. Since $u_2v$ strictly lies in the exterior of $g$, we
  obtain $\rot(u_2vw_1) < \rot(u_1vw_1)$. Hence, Lemma~\ref{lem:repr:equal_labels_at_intersection} gives
  $\ell_{C_1}(u_1v)<\ell_{C_2}(u_2v)$. Further, by definition of $C_1$
  and $C_2$ it holds $\ell_{C_1}(u_1v)\geq -1$ and
  $\ell_{C_2}(u_2v)\leq 0$. Thus, we obtain $\ell_{C_1}(u_1v)=-1$ and
  $\ell_{C_2}(u_2v)=0$. Since $C_1$ is an outer cascading cycle this
  also implies that $v$ is the endpoint of the negative path $P$ of
  $C_1$.

  Since $v'w'_1$ strictly lies in the exterior of $g$, we obtain
  $\rot(u'_2v'w'_2)>\rot(u'_2v'w'_1)$. Further, by the definition of
  labels we obtain
  $\ell_{C_2}(u'_2v')=\ell_{C_2}(v'w'_2)-\rot(u'_2v'w'_2)$. Applying
  this in Lemma~\ref{lem:repr:equal_labels_at_intersection} gives
  $\ell_{C_1}(v'w'_1)=\ell_{C_2}(v'w'_2)-\rot(u'_2v'w'_2)+\rot(u'_2v'w'_1)$. Thus,
  we obtain $\ell_{C_1}(v'w'_1) < \ell_{C_2}(v'w'_2)$.  Analogously to
  the arguments above, $\ell_{C_2}(v'w'_2) = 0$ and
  $\ell_{C_1}(v'w'_1)=-1$. Thus, $v'$ lies on $P$ and it is not the
  endpoint of $P$. It follows $v\neq v'$ and $C_1-P$ is contained in
  $g[v,v']$. Hence, there is an edge $e$ on $g[v,v']$ with label
  $\ell_{C_1}(e)>0$. By Corollary~\ref{cor:central-face-same-label}
  $\ell_{C_2}(e)=\ell_{C_1}(e)$, which contradicts that $C_2$ is an
  increasing cycle. 
\end{proof}

We use the next lemma to show that the introduced paths $T$ and $B$ do
not impact the validity of the ortho-radial representation.

\begin{lemma}\label{lem:second-phase:no-decreasing-cycle-on-T}
  Let $C$ be a cascading cycle in an ortho-radial
  representation~$\Gamma$ and let $P$ be a sub-path of $C$ such that
  \begin{compactenum}
  \item the intermediate vertices of $P$ have degree 2 in the graph,
  \item $P$ contains the negative path of $C$, and
  \item $P$ contains an edge~$e$ with label $\ell_{C}(e)>0$. 
  \end{compactenum}
  If $\Gamma$ without $P$ is valid, then $\Gamma$ is valid. 
\end{lemma}

\begin{proof}
  First assume that $\Gamma$ contains a decreasing cycle $C'$, which
  implies that $P$ is contained in $C'$ (in either direction).  Let
  $H=C+C'$ be the common sub-graph of $C$ and $C'$, and let $g$ be the
  central face of $H$. We distinguish the following two cases.

  \textit{Case 1, $P$ is part of $g$.}  First assume that $C$ and $C'$
  use $P$ in opposite directions. Since the central face locally lies
  to the right of any essential cycle, this implies that the central
  face lies to the left and right of $P$. Consequently, the central
  face is not simple, which contradicts that $H$ is biconnected. So
  assume that $C$ and $C'$ use $P$ in the same direction. By
  Proposition~\ref{lem:repr:equal_labels_at_intersection} it holds
  $\ell_{C}(e')=\ell_{C'}(e')=-1$ for any edge $e'$ on the negative
  path of $C$. Thus, $C'$ is not a decreasing cycle.

  \textit{Case 2, $P$ is not part of $g$.} Let $C_g$ be the essential
  cycle formed by $g$. Since $C_g$ consists of edges of $C'$ and $C$,
  the corresponding labels of $C'$ and $C$ also apply on $C_g$ by
  Proposition~\ref{lem:repr:equal_labels_at_intersection}. Further,
  since $C'$ is a decreasing cycle and $P$ is the only part of $C$
  that has a negative label on $C$, the cycle $C_g$
  only has non-negative labels. Since $P$ does not lie on $C_g$ but on
  $C$, $C'$ has at least one vertex with $C_g$ in common. By
  Proposition~\ref{prop:horizontal_cycle} the cycle $C_g$ is not
  horizontal. Altogether, $C_g$ is a decreasing cycle that also exists
  in $\Gamma-P$, which contradicts its validity.

  Finally, assume that $\Gamma$ contains an increasing cycle $C'$,
  which implies that $P$ is contained in $C'$ in either direction.
  Lemma~\ref{lem:cascading_and_increasing_cycles} implies that the
  central face $g$ of the subgraph formed by the two essential cycles
  $C$ and $C'$ is either $C$ or $C'$.  In particular, $P$ or
  $\reverse{P}$ lies on $g$. Hence, both $C$ and $C'$ use $P$ in the
  same direction as otherwise $g$ would lie in the exterior of one of
  these cycles. But this would contradict that they are essential.
  Hence, they both contain $P$ in this direction and $P$ also lies on
  $g$.  By Proposition~\ref{lem:repr:equal_labels_at_intersection}
  both cycles have the same labels on $P$. Since $P$ contains the edge
  $e$ with $\ell_{C}(e)>0$ and thus $\ell_{C'}(e)>0$, the cycle $C'$
  is not an increasing cycle.  
\end{proof}

We construct a cascading cycle $C_\T$ in $\Gamma_T$ as follows. Let
$C$ be the outermost decreasing cycle in $\Gamma^{u}_{e_{1}}$ and
let $ut_5$ be the newly inserted edge in $\Gamma^{u}_{e_{1}}$. We
replace $ut_5$ by $R[u,r_4]+T$ obtaining the cycle $C_\T$, which is
well-defined because $C$ uses $ut_5$ in that direction by
Lemma~\ref{lem:label_first_candidate}.

\begin{lemma}\label{lem:second-phase:cascading-cycles-CT}
  $C_\T$ is a cascading cycle with negative path $t_1t_2$ no matter
  whether $r_2r_3$ points to the right or downwards.
\end{lemma}

\begin{proof}
  Let $C$ be the outermost decreasing cycle in
  $\Gamma^{u}_{e_{1}}$. By
  Lemma~\ref{lem:same-labels-except-from-edge} the labels of $C$ and
  $C_\T$ coincide on $C[t_5,u](=C_\T[t_5,u])$. Hence, since $C$ is
  a decreasing cycle all labels on $C_\T[t_5,u]$ are non-negative. Further, by
  Lemma~\ref{lem:label_first_candidate} the edge $ut_5$ has label $0$
  on $C$. If $r_2r_3$ points to the right the sequence of the labels
  on $R[u,r_4]+T$ is therefore $0$, $0$, $0$, $-1$, $0$, $1$, $0$. If
  $r_2r_3$ points downwards the sequence is $0$, $1$, $0$, $-1$, $0$,
  $1$, $0$. In both cases $C_\T$ is a cascading cycle.
\end{proof}

Applying Lemma~\ref{lem:cascading_and_increasing_cycles} to the
situation of $C_\T$ proves that $\Gamma_T$ does not contain any
increasing cycles. Together with
Lemma~\ref{lem:second-phase:no-decreasing-cycle-on-T} this yields that
$\Gamma_T$ is valid.  We analogously prove the validity of $\Gamma_B$
as for $\Gamma_T$.  Let $C$ be the outermost decreasing cycle 
in $\Gamma^{u}_{e_{m-1}}$ and let $ub_5$ be the newly inserted edge in
$\Gamma^{u}_{e_{m-1}}$. We replace $ub_5$ by $R[u,r_5]+B$ obtaining
the cycle $C_\B$, which is well-defined because $C$ uses $ub_5$ in
that direction by Lemma~\ref{lem:label_first_candidate}.

\begin{lemma}\label{lem:second-phase:cascading-cycles-CB}
  $C_\B$ is a cascading cycles no matter whether $r_2r_3$ points to
  the right or downwards. In particular, $b_1b_2$ is the negative
  path of $C_\B$.
\end{lemma}
\begin{proof}
  Let $C$ be the outermost decreasing cycle in
  $\Gamma^{u}_{e_{m-1}}$. By
  Lemma~\ref{lem:same-labels-except-from-edge} the labels of $C$ and
  $C_\B$ coincide on $C[b_5,u](=C_\B[b_5,u])$. Hence, since $C$ is
  a decreasing cycle all labels on $C_\B[b_5,u]$ are non-negative. Further, by
  Lemma~\ref{lem:label_any_candidate} the edge $ub_5$ has label $0$
  on $C$. If $r_2r_3$ points to the right the sequence of the labels
  on $R[u,r_4]+B$ is therefore $0$, $0$, $0$, $0$, $-1$, $0$, $1$, $0$. If
  $r_2r_3$ points downwards the sequence is $0$, $1$, $0$, $0$, $-1$, $0$,
  $1$, $0$. In both cases $C_\B$ is a cascading cycle.
\end{proof}

Applying Lemma~\ref{lem:cascading_and_increasing_cycles} to the
situation of $C_\T$ proves that $\Gamma_B$ does not contain any
increasing cycles. Together with
Lemma~\ref{lem:second-phase:no-decreasing-cycle-on-T} this yields that
$\Gamma_B$ is valid. The following lemma summarizes the result.

\begin{lemma}\label{lem:second-phase:no-monotone-cycle-on-B}
 The ortho-radial representation $\Gamma_B=\Gamma_1$ is valid.
\end{lemma}

\paragraph{Correctness of Step 2.2.} By
Lemma~\ref{lem:second-phase:no-monotone-cycle-on-B} the ortho-radial
representation $\Gamma_1$ of \textit{Step 1} is valid. We now prove
that $\Gamma_2$ is a valid ortho-radial representation. We use the
same notation as in the description of the algorithm.

Starting with the valid ortho-radial representation~$\Gamma_1$, the
procedure iteratively resolves ports in the face $f_1$, which locally
lies to the right of $T$. In case that we resolve a vertical port $u'$
in a representation $\Pi$, the resulting ortho-radial
representation~$\Pi^{u'}_{e'_1}$ is valid by
Fact~\ref{prop:augmentation-properties:fact:vertical-port} of
Proposition~\ref{prop:augmentation-properties}, where $e'_1$ is the
first candidate of $u'$.  So assume that $u'$ is a horizontal port. In
that case we take $\Pi^{u'}_{e'_l}$ for the next iteration, where
$e'_l$ is the last candidate of $u'$. We observe that the augmentation
of $f_1$ may subdivide edges on the negative paths of $C_\T$ and
$C_\B$, but the added edges lie in the interior of $C_\T$ and the
exterior of $C_\B$. Hence, $C_\T$ remains an outer cascading cycle and
$C_\B$ an inner cascading cycle.

\begin{lemma}
 The ortho-radial representation $\Pi^{u'}_{e'_l}$ is valid. 
\end{lemma}

\begin{proof}
  Assume that $\Pi^{u'}_{e'_l}$ is not valid. Hence, there is a
  strictly monotone cycle $C$ that uses $e=u'z'$, where $z'$ is the vertex
  subdividing $e'_l$. Since $e'_l$ is the last candidate of $u'$, the
  cycle $C$ is increasing by Fact~\ref{prop:augmentation-properties:fact:monotone-cycle} of Proposition~\ref{prop:augmentation-properties}.
  By construction $e$ strictly lies in the interior of $C_\T$ and the exterior 
  of $C_\B$. This implies that $C$ lies in the interior of $C_\T$ and the 
  exterior of $C_\B$ by Lemma~\ref{lem:cascading_and_increasing_cycles}.
  In other words, $C$ is contained in the subgraph $H$ formed by the 
  intersection of the interior of $C_\T$ and the exterior of $C_\B$. As 
  $\subpath{R}{r_1,r_4}$ belongs to both $C_\T$ and $C_\B$, it is incident to 
  the outer and the central face of $H$. Hence, removing $\subpath{R}{r_1,r_4}$ 
  leaves a subgraph without essential cycles. Thus, the essential cycle $C$ 
  includes $\subpath{R}{r_1,r_4}$.
  
  By Proposition~\ref{lem:repr:equal_labels_at_intersection} the
  labels of $C$ and $C_\T$ are the same on $\subpath{R}{r_1,r_4}$. If
  $r_2r_3$ points downwards, its label is $1$, which contradicts that
  $C$ is increasing.  If otherwise $r_2r_3$ points right, it lies on a
  horizontal cycle. But then $C$ is not increasing by
  Proposition~\ref{prop:horizontal_cycle}.
\end{proof}

Altogether, applying the lemma inductively on the inserted edges, we
obtain that $\Gamma_2$ is valid.

\paragraph{Correctness of Step 2.3.}  As we only apply the first phase
of the augmentation step on $f_2$, the resulting ortho-radial
representation~$\Gamma_3$ is also valid due to the correctness of the
first phase. This concludes the correctness proof of the second phase.

\section{Bend Minimization} 
\label{sec:bend-minimization}
We have considered bend-free ortho-radial drawings so
far. In this section we shortly describe how to extend our results to
ortho-radial drawings with bends. In particular, we present complexity
results on bend-minimzation for ortho-radial drawings.

We model bends by subdividing each edge of $G$ with $K$ degree-2
vertices. Let $S$ denote the set of all these subdivision vertices. A
vertex $v\in S$ is a \emph{bend} in an ortho-radial representation if
for its two incident edges $e_1$ and $e_2$ it holds
$\rot(e_1,e_2)\neq 0$. A valid ortho-radial representation of $I$ is
\emph{bend-minimal} if there is no other valid ortho-radial
representation of $I$ that has fewer bends.  Niedermann and
Rutter~\cite{Niedermann2020} recently showed that $K=2n-4$ subdivision vertices per
edge are sufficient for creating bend-minimal valid ortho-radial
representations. The next theorem shows that from a computational
point of view there is a substantial difference between creating
orthogonal drawings and ortho-radial drawings with a minimum number of
bends even if the embedding is prescribed.

\begin{thm}\label{thm:complexity-valid-ortho-radial}
  Given an instance $I=(G,\mathcal E,f_c,f_o)$, finding a bend-minimal
  valid ortho-radial representation of $I$ is $\mathcal{NP}$-hard.
\end{thm}

\begin{proof}
  Given a valid bend-minimal ortho-radial representation~$\Gamma$, we
  utilize the TSM framework to create a bend-minimal ortho-radial
  drawing~$\Delta$ based on $\Gamma$ in polynomial time. Consequently,
  if $\Gamma$ could be created in polynomial time, we could also
  derive $\Delta$ in polynomial time.  On the other hand, Garrido and
  Marquez~\cite{Garrido1997} proved that given an embedding
  $\mathcal E$ of a graph $G$ in a surface $S$ with genus $g$ it is
  $\mathcal{NP}$-hard to find an orthogonal grid drawing of $\mathcal E$ on $S$
  that minimizes the number of bends. In particular, for $g=0$ the
  surface $S$ corresponds to a sphere. As the constructions used in
  the proposed reduction from 3SAT do not use the poles of the sphere,
  this result directly transfers to orthogonal grid drawings on
  cylinders and hence to ortho-radial drawings.
\end{proof}

Niedermann and Rutter~\cite{Niedermann2020}
recently presented an integer linear programming formulation for
creating valid bend-minimal ortho-radial representations assuming a pre-defined embedding of the graph.

If we do not assume a pre-defined embedding of the graph, finding a
bend-minimal drawing becomes \NP-hard even for the orthogonal case.
Garg and Tamassia~\cite{gt-ccurpt-01} showed that the problem
\textsc{Orthogonal Embeddability} to decide whether a given planar
4-graph admits an orthogonal drawing without bends is \NP-complete.
In the remaining part of this section we study the analogous problem
\textsc{Ortho-radial Embeddability} for ortho-radial drawings and prove
that it is \NP-complete as well.  We say a graph~$G$ admits an
\emph{ortho-radial (or orthogonal) embedding} if there is an embedding
of $G$ such that $G$ can be drawn ortho-radially (or orthogonally)
without bends.

We give a reduction from \textsc{Orthogonal Embeddability}.  To do so,
we note that the reduction by Garg and Tamassia~\cite{gt-ccurpt-01}
actually produces instances $G=(V,E)$ with a fixed edge $e \in E$ such
that it is \NP-complete to decide whether $G$ has an orthogonal
embedding where $e$ is incident to the outer face.

Given such a graph, we build a structure around $G$ that yields a
graph $G'$ such that in any ortho-radial embedding of~$G'$ the induced
representation~$\Gamma$ of $G$ does not contain any essential cycles.
In other words, $\Gamma$ is actually an orthogonal representation of
$G$. Hence, an ortho-radial embedding of $G'$ can only exist if $G$
admits an orthogonal embedding.  We may assume without loss of
generality that $G$ is connected as otherwise, we handle each
component separately.

\begin{figure}[bt]
  \centering
  \begin{subfigure}{.45\textwidth}
    \centering
    \includegraphics{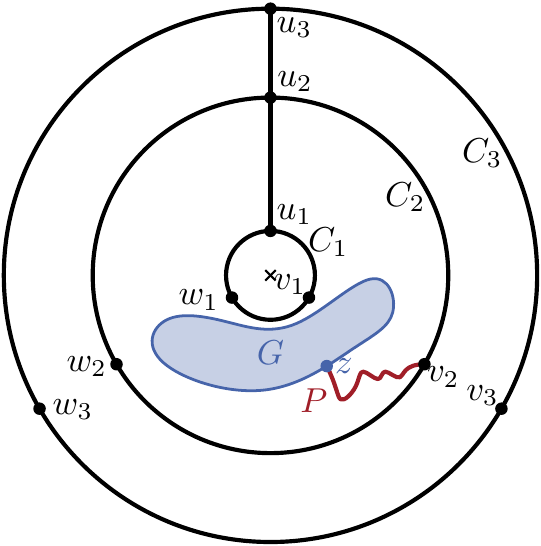}
    \caption{$G$ lies between $C_1$ and $C_2$.}
    \label{fig:embed:embedding-23}
  \end{subfigure}
  \hfil
  \begin{subfigure}{.45\textwidth}
    \centering
    \includegraphics{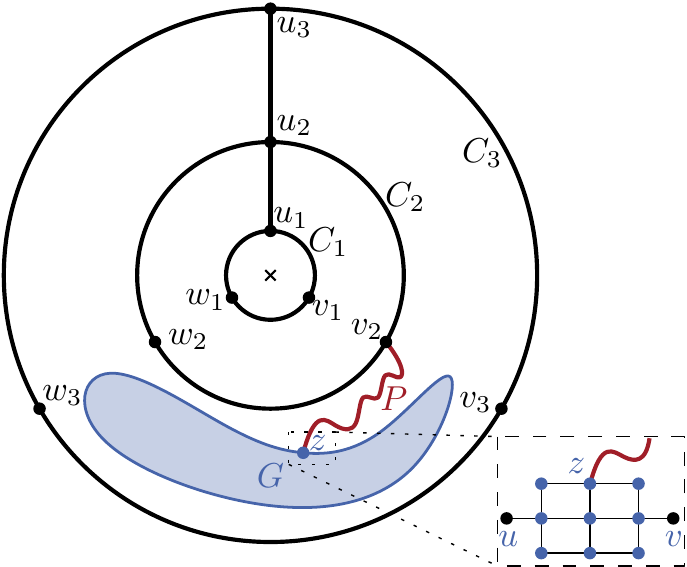}
    \caption{$G$ lies between $C_2$ and $C_3$.}
    \label{fig:embed:embedding-12}
  \end{subfigure}
  \caption{Possible embeddings of $G'$: In both cases $G$ contains no
    essential cycles. The roles of $C_1$ and $C_3$ can be exchanged.}
  \label{fig:embed:embedding}
\end{figure}

The construction of $G'$ from $G$ is based on the fact that there is
only one way to ortho-radially draw a triangle~$C$, i.e., a cycle of
length~3, without bends: as an essential cycle on one circle of the
grid.  We build a graph $H$ consisting of three triangles called
$C_1$, $C_2$ and $C_3$ and denote the vertices on $C_i$ by $u_i$,
$v_i$ and $w_i$.  Furthermore, $H$ contains the edges $u_1u_2$ and
$u_2u_3$.  In Figure~\ref{fig:embed:embedding} $H$ is formed by the
black edges.  To connect $H$ and $G$, we replace the special edge $e=uv$
of $G$ by a $3 \times 3$-grid and connect one of the
degree-3 vertices $z$ of that grid by a path $P$ to $v_2$, where we
choose the length of the path equal to the number of edges in $G$.
The resulting graph is $G'$; see Figure~\ref{fig:embed:embedding}.  The
reduction clearly runs in polynomial time.  Moreover, if~$G'$ admits
an ortho-radial embedding, then the triangles of $H$ must be drawn as
essential cycles, and therefore $G$ must be contained in one of the
two regular faces of $H$, and can hence not contain any essential
cycles.  We therefore find an orthogonal embedding of $G$.
Conversely, if we have an orthogonal embedding of $G$ with $e$ on the
outer face, then it can be inserted into a face of the drawing of $H$
and the path $P$ can be drawn as it is sufficiently long.  We
therefore obtain an ortho-radial embedding of $G'$.  This proves the
following theorem.

\begin{thm}
  \textsc{Ortho-radial Embeddability} is \NP-complete.
\end{thm}

\section{Conclusion}\label{sec:conclusion}
In this paper we considered orthogonal drawings of graphs on
cylinders. Our main result is a characterization of the plane 4-graphs
that can be drawn bend-free on a cylinder in terms of a combinatorial
description of such drawings. These ortho-radial representations
determine all angles in the drawing without fixing any lengths, and
thus are a natural extension of Tamassia's orthogonal
representations. However, compared to those, the proof that every
valid ortho-radial representation has a corresponding drawing is
significantly more involved. The reason for this is the more global
nature of the additional property required to deal with the cyclic
dimension of the cylinder.

Our ortho-radial representations establish the existence of an
ortho-radial TSM framework in the sense that they are a combinatiorial
description of the graph serving as interface between the ``Shape''
and ``Metrics'' step.

For rectangular plane 4-graphs, we gave an algorithm producing a
drawing from a valid ortho-radial representation. Our proof reducing
the drawing of general plane 4-graphs with a valid ortho-radial
representation to the case of rectangular plane 4-graphs is
constructive. We have described an algorithm that checks the validity
of an ortho-radial representation in $\O(n^{2})$ time. In the positive
case, we can also produce a corresponding drawing in the same running
time, whereas in the negative case we find a strictly monotone cycle.
These algorithms correspond to the ``Metrics'' step in a TSM framework
for ortho-radial drawings.

While bend-minimal orthogonal representations can be created in
polynomial time, we showed that creating valid bend-minimal
ortho-radial representations is \NP-hard. This directly rises the
research question of developing approximation algorithms and fixed
parameter tractable algorithms for such representations. Further,
powerful and fast heuristics may help to carry over the framework into
practice, e.g., for creating ortho-radial drawings of large
transportation networks.

Finally, we want to emphasize that we deem the generalization of
ortho-radial drawings from the cylinder to the torus or even more
complex surfaces an interesting and promising research question. It is
far from clear how to transfer our results to the torus as its two
cyclic dimensions lead to different types of essential cycles.

\paragraph{Acknowledgments}
Lukas Barth was funded by the German Research Foundation (DFG) as
part of the Research Training Group GRK 2153: Energy Status Data --
Informatics Methods for its Collection, Analysis and Exploitation.
Matthias Wolf was funded by the Helmholtz Association Program Storage
and Cross-linked Infrastructures, Topic 6 Superconductivity, Networks
and System Integration and by the German Research Foundation (DFG) as
part of the Research Training Group GRK 2153: Energy Status Data --
Informatics Methods for its Collection, Analysis and Exploitation.

\bibliographystyle{abbrv}      %
\bibliography{paper}              %

\end{document}